\documentclass[sigconf]{acmart}
\usepackage{caption}
\usepackage{subcaption}

\usepackage{amsmath}
\usepackage{physics}
\usepackage{siunitx}
\usepackage{mathpartir}
\usepackage{mathtools}
\usepackage{extpfeil}
\usepackage{multirow}
\usepackage{algorithm}
\usepackage{algpseudocode}
\usepackage{xspace}
\usepackage{enumitem}

\usepackage{pbalance}

\usepackage{tikz}
\usepackage{pgfplots}
\usetikzlibrary{tikzmark,fit,calc,positioning,arrows,graphs,arrows.meta}

\usepackage{pifont}

\usepackage[english]{babel}
\newtheorem{theorem}{Theorem}

\usepackage{minted}
\setminted{linenos,breaklines,breakanywhere,tabsize=4,xleftmargin=1.5em,escapeinside=||,fontsize=\footnotesize}

\newcommand{\BashFancyFormatLine}{%
  \def\FancyVerbFormatLine##1{\$\ \,##1}%
}
\setminted[bash]{
    linenos=false,
    formatcom=\BashFancyFormatLine,
    xleftmargin=0pt,
    xrightmargin=5pt,
}

\usepackage{adjustbox}

\usepackage{accsupp} 
\newcommand*{\llbrace}{%
  \BeginAccSupp{method=hex,unicode,ActualText=2983}%
    \textnormal{\usefont{OMS}{lmr}{m}{n}\char102}%
    \mathchoice{\mkern-4.05mu}{\mkern-4.05mu}{\mkern-4.3mu}{\mkern-4.8mu}%
    \textnormal{\usefont{OMS}{lmr}{m}{n}\char106}%
  \EndAccSupp{}%
}
\newcommand*{\rrbrace}{%
  \BeginAccSupp{method=hex,unicode,ActualText=2984}%
    \textnormal{\usefont{OMS}{lmr}{m}{n}\char106}%
    \mathchoice{\mkern-4.05mu}{\mkern-4.05mu}{\mkern-4.3mu}{\mkern-4.8mu}%
    \textnormal{\usefont{OMS}{lmr}{m}{n}\char103}%
  \EndAccSupp{}%
}

\usepackage{wasysym}

\newif\ifdraft
\draftfalse

\newif\ifextendedversion
\extendedversiontrue

\definecolor{myblue}{HTML}{0066CC}
\definecolor{mygreen}{HTML}{009900}
\definecolor{matchagreen}{HTML}{5B7132}
\definecolor{darkblue}{HTML}{004C99}
\definecolor{darkpurple}{HTML}{4C0099}

\newcommand{\shixin}[1]{\ifdraft\noindent{\footnotesize\textcolor{teal}{\fbox{SS} {#1}}}\fi}

\newcommand{\todo}[1]{\ifdraft\noindent{\footnotesize\textcolor{blue}{\fbox{\bf TODO} {#1}}}\fi}

\newcommand{\revision}[1]{#1}

\newcommand{\extver}[1]{\ifextendedversion#1\fi}
\newcommand{\noextver}[1]{\ifextendedversion\else#1\fi}
\newcommand{\bothver}[2]{\ifextendedversion#1\else#2\fi}

\newcommand{\pgheading}[1]{\noindent\textbf{#1.}}

\newcommand{\sys}{\emph{SecSep}\xspace}
\newcommand{\sysisa}{x86-64\xspace}
\newcommand{\systal}{\emph{Octal}\xspace}

\newcommand{\highlight}[1]{\colorbox{blue!20}{#1}}
\newcommand{\highlightmath}[1]{\colorbox{blue!20}{\ensuremath{#1}}}

\newcommand{\reg}[1]{\texttt{#1}}

\algnewcommand\algorithmicswitch{\textbf{switch}}
\algnewcommand\algorithmiccase{\textbf{case}}
\algnewcommand\algorithmicdefault{\textbf{default}}
\algdef{SE}[SWITCH]{Switch}{EndSwitch}[1]{\algorithmicswitch\ #1\ \algorithmicdo}{\algorithmicend\ \algorithmicswitch}%
\algdef{SE}[CASE]{Case}{EndCase}[1]{\algorithmiccase\ #1}{\algorithmicend\ \algorithmiccase}%
\algdef{SE}[DETAULT]{Default}{EndDefault}{\algorithmicdefault}{\algorithmicend\ \algorithmicdefault}%
\algtext*{EndCase}%
\algtext*{EndDefault}

\newcommand{\hasType}[3]{#1 \vdash #2 : #3}

\newcommand{\constraint}[1]{\highlightmath{\Rrightarrow #1}}

\newcommand{\hasVal}[2]{#1 \vdash #2}

\newcommand{\typeChecked}[2]{#1 \vdash_{\textsc{tal}} #2}

\newcommand{\dynNext}[4]{#1 \vdash #2 \xrightarrow{#4} #3}

\newcommand{\dom}[1]{\mathop{\mathrm{dom}}(#1)}

\newcommand{\pcVar}{\mathit{pc}}
\newcommand{\memOp}[4]{#4(#1,#2,#3)}
\newcommand{\opVar}{\mathit{op}}
\newcommand{\instVar}{\mathit{inst}}
\newcommand{\instr}[1]{\mathop{\mathbf{#1}}}
\newcommand{\inst}[2]{\instr{#1}{#2}}
\newcommand{\opMode}[1]{^{\textcolor{blue}{#1}}}

\newcommand{\slotAnno}[1]{\textcolor{darkpurple}{\mathbf{#1}}}
\newcommand{\taintAnno}[1]{\textcolor{darkblue}{\mathbf{#1}}}

\newcommand{\spVar}{\mathit{sp}}
\newcommand{\ptrVar}{\mathit{ptr}}
\newcommand{\sigmaOp}{\sigma_\text{op}}
\newcommand{\sigmaFp}{\sigma_{f_p}}
\newcommand{\sigmaBp}{\sigma_{b_p}}
\newcommand{\sigmaCall}{\sigma_\text{call}}
\newcommand{\sigmaRet}{\sigma_\text{ret}}
\newcommand{\sOp}{s_\text{op}}
\newcommand{\tauOp}{\tau_\text{op}}
\newcommand{\sRet}{s_\text{ret}}
\newcommand{\sCalleeStack}{s_\text{calleeStack}}
\newcommand{\deltaOp}[1][]{\delta_\text{op#1}}
\newcommand{\deltaPtr}[1][]{\delta_\text{ptr#1}}

\newcommand{\typeRange}[3]{[#1, #2]_{#3}}

\newcommand{\PType}{\mathcal{P}}
\newcommand{\MType}{\mathcal{M}}
\newcommand{\RType}{\mathcal{R}}
\newcommand{\SType}{\mathcal{S}}

\newcommand{\sval}{s^\text{valid}}

\newcommand{\Compiler}{\mathcal{C}}

\newcommand{\CompilerOp}{\mathcal{C}_\text{op}}
\newcommand{\CompilerCallee}{\mathcal{C}_\text{callee}}
\newcommand{\CompilerPtr}{\mathcal{C}_\text{ptr}}
\newcommand{\transOp}{\texttt{TransOp}\xspace}
\newcommand{\transPtr}{\texttt{TransPtr}\xspace}
\newcommand{\stackPub}{s_\text{stackPub}}
\newcommand{\stackSec}{s_\text{stackSec}}
\newcommand{\otherPub}{s_\text{otherPub}}
\newcommand{\otherSec}{s_\text{otherSec}}

\newcommand{\swcontract}[3]{\llbracket#1\rrbracket_{\text{#2}}^{\text{#3}}}
\newcommand{\hwcontract}[3]{\llbrace #1 \rrbrace_{\text{#2}}^{\text{#3}}}

\newcommand{\typingMovqMR}{\textsc{Typing-Movq-m-r}\xspace}
\newcommand{\typingMovqRM}{\textsc{Typing-Movq-r-m}\xspace}

\newcommand{\typingJne}{\textsc{Typing-Jne}\xspace}
\newcommand{\typingCallq}{\textsc{Typing-Callq}\xspace}

\AtBeginDocument{%
  }

\setcopyright{acmlicensed}
\copyrightyear{2025}
\acmYear{2025}
\acmDOI{XXXXXXX.XXXXXXX}
\acmConference[Conference acronym 'XX]{Make sure to enter the correct
  conference title from your rights confirmation email}{June 03--05,
  2018}{Woodstock, NY}
\acmISBN{978-1-4503-XXXX-X/2018/06}

\begin{document}

\title{Securing Cryptographic Software via Typed Assembly Language\extver{ (Extended Version)}}

\author{Shixin Song}
\authornote{Both authors contributed equally to this research.}
\email{shixins@mit.edu}
\orcid{0009-0007-5638-5164}
\affiliation{%
  \institution{Massachusetts Institute of Technology}
  \city{Cambridge}
  \country{United States}
}

\author{Tingzhen Dong}
\authornotemark[1]
\email{rogerdtz@mit.edu}
\orcid{0000-0002-0514-923X}
\affiliation{%
  \institution{Massachusetts Institute of Technology}
  \city{Cambridge}
  \country{United States}
}

\author{Kosi Nwabueze}
\email{kosinw@mit.edu}
\orcid{0009-0002-0770-0344}
\affiliation{%
  \institution{Massachusetts Institute of Technology}
  \city{Cambridge}
  \country{United States}
}

\author{Julian Zanders}
\email{jzanders@mit.edu}
\affiliation{%
  \institution{Massachusetts Institute of Technology}
  \city{Cambridge}
  \country{United States}
}

\author{Andres Erbsen}
\email{andreser@mit.edu}
\orcid{0000-0002-9854-7500}
\affiliation{%
  \institution{Google}
  \city{Cambridge}
  \country{United States}
}

\author{Adam Chlipala}
\email{adamc@csail.mit.edu}
\orcid{0000-0001-7085-9417}
\affiliation{%
  \institution{Massachusetts Institute of Technology}
  \city{Cambridge}
  \country{United States}
}

\author{Mengjia Yan}
\email{mengjiay@mit.edu}
\orcid{0000-0002-6206-9674}
\affiliation{%
  \institution{Massachusetts Institute of Technology}
  \city{Cambridge}
  \country{United States}
}

\renewcommand{\shortauthors}{Song et al.}

\begin{abstract}
Authors of cryptographic software are well aware that their code should not leak secrets through its timing behavior, and, until 2018, they believed that following industry-standard \emph{constant-time} coding guidelines was sufficient.
However, the revelation of the Spectre family of speculative execution attacks injected new complexities.

To block speculative attacks, prior work has proposed annotating the program's source code to mark secret data, with hardware using this information to decide when to speculate (i.e., when only public values are involved) or not (when secrets are in play).
While these solutions are able to track secret information stored on the heap, they suffer from limitations that prevent them from correctly tracking secrets on the stack, at a cost in performance.

This paper introduces \sys{}, a transformation framework that rewrites assembly programs so that they partition secret and public data on the stack. By moving from the source-code level to assembly rewriting, \sys{} is able to address limitations of prior work. 
The key challenge in performing this assembly rewriting stems from the loss of semantic information through the lengthy compilation process. The key innovation of our methodology is a new variant of typed assembly language (TAL), \systal, which allows us to address this challenge.
Assembly rewriting is driven by compile-time inference within \systal.
We apply our technique to cryptographic programs and demonstrate that it enables secure speculation efficiently, incurring a low average overhead of $1.2\%$.
\end{abstract}



\maketitle

\section{Introduction}
\label{sec:intro}
Cryptographic software has strong security requirements and is often strengthened to prevent information leakage through timing side channels by adhering to \textit{constant-time} coding,
which forbids secret-dependent values as branch conditions or memory addresses.

However, recent speculative-execution attacks, notably various Spectre attacks~\cite{Kocher2018spectre,kiriansky2018speculative,koruyeh2018spectre,maisuradze2018ret2spec,schwarz2019netspectre}, have invalidated the security guarantees offered by constant-time programming.
Modern processors employ aggressive speculative-execution mechanisms that predict upcoming instructions to be executed and roll back architectural state if the prediction is later found to be incorrect. 
While offering significant performance benefits, such speculative-execution mechanisms introduce a large attack surface, enabling attackers to trigger a program to execute unintended instructions speculatively to access secrets and transmit them via timing side channels.

Recent work~\cite{guarnieri2020spectector} has uncovered multiple vulnerabilities in real-world cryptographic libraries even under constrained speculative-execution models, such as only mispredicting limited types of branches.
\todo{Better to double check}
As modern processors evolve with ever-more-complex speculation mechanisms, we need mitigation solutions that protect broader speculative behaviors.
Practical mitigation needs to navigate the complex trade-offs between security guarantees, performance overhead, and hardware complexity.

Many mitigation solutions~\cite{yu2019speculative,choudhary2021speculative,yan2018invisispec,weisse2019nda,schwarz2020context,daniel2023prospect} share a common philosophy: identify secret data and then delay speculative execution for operations that may transmit such data.
The key research challenge in these approaches lies in how to identify the secret data precisely without incurring high overhead. 
One promising design~\cite{daniel2023prospect} is to augment the hardware with fine-grained taint tracking at the register level and coarse-grained taint tracking at the memory level (e.g., at page or section granularity).
This architecture avoids the prohibitive costs of byte- or word-level tracking while retaining sufficient granularity to enforce secure speculation.

However, this hardware design requires the software to partition secret and public data into distinct memory regions explicitly, so that the hardware can interpret the secrecy status of the data accurately using its coarse-grained taint-tracking capability. Performance and security of the hardware design are contingent on precise annotation of secret data.
Prior projects, ProSpeCT~\cite{daniel2023prospect} and ConTExT~\cite{schwarz2020context}, set out to add this partitioning capability to software through requiring fine-grained source-code annotations. Specifically, these methods require programmers to mark variables in the source code (e.g., C) as either secret or public. This information is then used as follows: For heap data, a customized memory allocator allocates secret and public objects in different memory pools. Stack data is protected by annotating secret and public stack variables manually to relocate them to different regions.

These source-level approaches work well for heap data but less so for the stack.
Critically, they are unable to partition the stack accurately, requiring conservative partitioning and thereby suffering from performance loss.
These limitations are \textit{inherent} to source-level annotation methodologies. 
First, operating at the source level gives no visibility or control over register spills.
Consequently, if a secret register is spilled to the stack, the programmer is forced to mark the whole stack as secret conservatively, leading to overtainting and unnecessary performance overhead.
Second, some approaches relocate stack variables into global memory regions, 
which may compromise functional correctness under concurrency.
Most importantly, source-level transformations heavily rely on strong assumptions about compiler internals.
However, given the complexity and opaqueness of modern compilers, source-level transformation suffers from a significantly enlarged trusted computing base (TCB) and fragile compilation process that is difficult to verify.

\subsection{This Paper}
In this paper, we introduce \sys, an assembly-transformation framework that partitions stack data securely.

We allow key information (usually lost during compilation) to be reintroduced in code through a new variant of typed assembly language (TAL), \systal, which facilitates sound program transformation.
\systal is designed to enable \textit{static} fine-grained taint tracking.

We design the type system by assigning dependent types and taint types to all registers and data objects in memory.
The key idea is to leverage the dependent types to track the value ranges of registers and memory objects, so that we can construct the full picture of their points-to relationships throughout the program.
While precise points-to analysis is infeasible for arbitrary programs, we take advantage of a domain-specific property, that is, cryptographic software is typically written following the constant-time programming discipline, making it amenable to our analysis.
\systal also ensures well-typed programs are memory-safe.

Building upon \systal, we design a program-transformation framework \sys consisting of two important components.
The frontend is a heuristic type-inference algorithm that operates on off-the-shelf x86-64 assembly programs (plus debug tables already produced by Clang, plus a table mapping each function interface to a type in our new type system).
The analysis involves a set of heuristic rules to reason about pointer arithmetic and loop counters.
The outcome of the inference tool is an \systal program where the taint status of every memory operation is identified explicitly.

The backend of our framework is the code-transformation tool that rewrites programs based on their taint types.
It supports real-world cryptographic programs, which involve complex interleaving of secret and public reads/writes and shared pointer-based structures.
After locating memory operations with taint types, they are rewritten depending on their secrecy statuses.

We formally prove the type safety of \systal. We also prove that \sys's transformation separates secret and public data while guaranteeing functional correctness.
We implement a hardware extension that achieves secure speculation with register-level and memory-segment-level taint tracking on the gem5 simulator~\cite{Binkert:2011:gem5,Lowe-Power:2020:gem5-20}.
We evaluate \sys's transformation with the hardware extension using six cryptographic benchmarks~\cite{boringssl} and show that
it enables secure speculation with a negligible overhead of $1.2\%$ on average.

In summary, we make the following contributions:
\begin{itemize}
[leftmargin=*]
    \item We propose \systal, a variant of typed assembly language (TAL) with static fine-grained taint tracking for assembly programs.
    \item We design a program-transformation framework \sys that (1) heuristically infers types for off-the-shelf cryptographic assembly programs and (2) rewrites them to split their secret and public data across coarse-grained memory regions.
    \item We prove soundness of the technique\noextver{~\cite{secsep-proof}}, provide a prototype implementation, and carry out an empirical evaluation. 
\end{itemize}
\pgheading{Availability} Our prototype for \sys is open-sourced at \href{https://github.com/MATCHA-MIT/secsep}{\url{https://github.com/MATCHA-MIT/secsep}}.

\section{Background}
\subsection{Microarchitectural Side-Channel Attacks}
Microarchitectural side-channel attacks exploit transmitter instructions that leave visible side effects on microarchitectural state like caches~\cite{osvik2006cache,yan2019attack,yarom2014flush+,yarom2017cachebleed}, TLBs~\cite{gras2018translation}, branch predictors~\cite{aciiccmez2007predicting,evtyushkin2018branchscope}, and others~\cite{wang2017leaky,xu2015controlled,pessl2016drama,aldaya2019port,andrysco2015subnormal,coppens2009practical,evtyushkin2016covert,grossschadl2009side,gruss2016prefetch,moghimi2019memjam,shin2018unveiling}.
Cryptographic programs prevent these attacks by following the constant-time coding discipline, which avoids executing such transmitter instructions with secret-dependent operands.

However, speculative-execution attacks~\cite{Kocher2018spectre,kiriansky2018speculative,koruyeh2018spectre,maisuradze2018ret2spec,schwarz2019netspectre} exploit the side effects of speculatively executed transmitters to leak the secrets, which are not blocked by the constant-time discipline.

\subsection{Typed Assembly Language}
Conventional assembly language omits most high-level semantic information, making static analysis challenging.
Typed assembly language (TAL)~\cite{morrisett1999system,grossman2000scalable,crary1999talx86,glew1999type,morrisett1998stack} was introduced to regain some of that information. 
We extend past results around memory safety to information-flow tracking to guide program transformation.

\section{Threat Model \& Security Properties}
\label{sec:sw-hw-contract}
We aim to protect cryptographic applications against transient-execution attacks.
Specifically, we assume the software is written following the constant-time coding discipline, in which the program 
avoids using secret-dependent values as branch conditions and memory addresses. 
However, the underlying hardware employs aggressive speculative-execution mechanisms, including prediction on both direct and indirect branches, which can result in transient instruction sequences that violate the constant-time requirements.

Our proposed assembly-rewriting technique is a key component in a software-hardware codesign mitigation.
On the software side, our rewriting tool transforms the constant-time cryptographic programs to separate secret and public data into distinct regions.
On the hardware side, we use an existing Spectre mitigation~\cite{daniel2023prospect}
with a fine-grained taint-tracking mechanism at the register level and coarse-grained taint tracking at the memory level, for a good trade-off between performance, cost, and security.
The hardware uses the taint-tracking information to delay the execution of any potential transmitter instructions, which may leak information via timing side channels, when their operands are tainted.


\revision{We use two observation models $\swcontract{\cdot}{ct}{}$ and $\swcontract{\cdot}{pub}{}$ to constrain our software requirements: constant-time and separating secret/public data.}
Specifically, the notion $\swcontract{\cdot}{}{}$ represents the observation trace of executing a program at the architectural level. Supposing the architectural trace of executing program $P$ is $S_0\xrightarrow{o_1}S_1\xrightarrow{o_2}\dots $, the observation trace is then defined as $\swcontract{P}{}{}(S_0)=o_1o_2\dots$. 

Here, $\swcontract{\cdot}{ct}{}$ records the trace of load/store addresses and branch targets, and $\swcontract{\cdot}{pub}{}$ records a trace of data values stored in the public memory region.
We then define public noninterference, the software property guaranteed by \sys, 
where $\simeq_\text{pub}$ constrains that the two states have equal values in the public memory region\revision{\bothver{ (see Definition~\ref{def:pub-eq})}{~\cite{secsep-proof}}}.
\begin{definition}[Software Public Noninterference]
A program $P$ satisfies \emph{software public noninterference} for a specific public region if for all initial configurations $S$ and $S'$, if $S\simeq_\text{pub} S'$, then $\swcontract{P}{ct}{}(S)=\swcontract{P}{ct}{}(S')$ and $\swcontract{P}{pub}{}(S)=\swcontract{P}{pub}{}(S')$.
\label{def:pub-noninterference}
\end{definition}

We use $\hwcontract{P}{}{}(S)$ to denote the microarchitectural observation trace of running program $P$ on our out-of-order processor with initial state $S$~\cite{schwarz2020context}. 
The hardware must obey the following condition.
\begin{definition}[Hardware Public Noninterference]
A processor satisfies \emph{hardware public noninterference} if for all programs $P$ and all initial states $S$, $S'$, if $\swcontract{P}{pub}{}(S)=\swcontract{P}{pub}{}(S')$ and $\swcontract{P}{ct}{}(S)=\swcontract{P}{ct}{}(S')$, then $\hwcontract{P}{}{}(S)=\hwcontract{P}{}{}(S')$.
\end{definition}

In summary, \sys achieves secure speculation by ensuring that the software component satisfies the public noninterference contract. 
For the hardware component, we refer readers to the ProSpeCT paper~\cite{daniel2023prospect}, which provides formal proof that the taint-tracking hardware mechanism described above satisfies hardware public noninterference.
Together, the software-hardware contract ensures end-to-end security of the overall system, where secret data do not influence microarchitectural side channels, even in the presence of speculative execution.

\section{Motivation and Overview}

\subsection{Limitations of Source-Level Annotation}
\label{sec:src-anno-limitation}

A fundamental limitation of source-level code transformation used by prior works~\cite{daniel2023prospect,schwarz2020context} lies in its heavy reliance on assumptions about compiler internals.
\revision{For example, to relocate secret stack variables, programmers must manually annotate declarations with the section label ``sec.''
Then, the compiler is expected to allocate the variables in that section, which will be marked as secret by the hardware defense.} 
However, this strategy fails to guarantee public noninterference, due to lack of control over register allocation, spilling, and compiler optimizations.

\begin{figure*}[t]
    \centering
\begin{subfigure}{0.3\linewidth}
\centering
\begin{minted}{c}
void salsa20_words 
(uint32_t *d, uint32_t s[16]) {
  __attribute__
  ((section("sec"))) static 
  uint32_t x[4][4];
  int i;
  for (i=0; i<16; ++i) 
    x[i/4][i%4] = s[i];
// Omit calculation on x
...
  for (i=0; i<16; ++i) 
    d[i] = x[i/4][i%4] + s[i];
}
\end{minted}
    \caption{Annotated source code}
    \label{fig:salsa20-source}
\end{subfigure}%
\begin{subfigure}{0.3\linewidth}
\centering
\begin{minted}{asm}
salsa20_words:

# Public spill
  movq %rdi, -64(%rsp)
# Source: line 7-8
  movl (%rsi), %r12d
  movl 4(%rsi), %r11d
  movl 8(%rsi), %r9d
  movl 12(%rsi), %eax
# Secret spill
  movq %rax, -88(%rsp)
...
\end{minted}
    \caption{Original assembly code}
    \label{fig:salsa20-asm}
\end{subfigure}%
\begin{subfigure}{0.3\linewidth}
\centering
\begin{minted}{text}
salsa20_words:
# rdi: |$d$|, rsi: |$s$|, rsp: |$p$|
# Public spill
  movq %rdi, -64(%rsp)|$^{\slotAnno{[p-64, p-56)}, \taintAnno{0}}$|
# Source: line 7-8
  movl (%rsi)|$^{\slotAnno{[s, s+64)}, \taintAnno{1}}$|, %r12d
  movl 4(%rsi)|$^{\slotAnno{[s, s+64)}, \taintAnno{1}}$|, %r11d
  movl 8(%rsi)|$^{\slotAnno{[s, s+64)}, \taintAnno{1}}$|, %r9d
  movl 12(%rsi)|$^{\slotAnno{[s, s+64)}, \taintAnno{1}}$|, %eax
# Secret spill
  movq %rax, |\textcolor{red}{$\mathbf{\delta}$}|-88(%rsp)|$^{\slotAnno{[p-88, p-80)}, \taintAnno{1}}$|
...
\end{minted}
    \caption{Transformed assembly code}
    \label{fig:salsa20-asm-trans}
\end{subfigure}%
    \caption{Program transformation: source-code annotation v.s. assembly rewriting}
    \label{fig:salsa20}
\end{figure*}

To illustrate how such a strategy can go wrong in practice, we present a case study from a cryptographic function \texttt{salsa20\_words}.
Figure~\ref{fig:salsa20} shows both the annotated C code (Figure~\ref{fig:salsa20-source}) and the corresponding assembly code generated by clang-16 (Figure~\ref{fig:salsa20-asm}).
In the C code, the function takes a pointer \texttt{d} and a secret array \texttt{s[16]} as input.
It then declares a local array \texttt{x[4][4]}, which is used to hold secret data from array \texttt{s} (lines 7-8) and is further used for computation in lines 11-12.

Given that array \texttt{x} holds secret data, a programmer can annotate it with the \revision{section label \texttt{sec},} 
expecting the compiler to allocate it in the secret-marked global region.
However, the generated assembly code shown in Figure~\ref{fig:salsa20-asm} deviates significantly from this expectation.

First, the compiler notices the array size is small enough to be stored in registers and decides to skip memory allocation for \texttt{x} completely.
Specifically, in lines 6-9 of the assembly code, multiple elements inside the secret input array (base pointer \texttt{rsi}) are loaded into distinct registers, with no redirection to a secret region.

Second, we observe that in line 11, a secret register is spilled onto the stack, mixing secret data with many other public stack values.
Since the source-level annotation has no control over register spilling, the programmer is forced to mark the whole stack as secret, resulting in overtainting and serious performance degradation.
For example, according to our experiment on the \texttt{salsa20} application, this conservative approach results in $70\%$ performance degradation.

\subsection{Overview of \sys}
\label{sec:overview-relocate-stack}

We propose \sys, a framework to perform the secret-public memory separation at the assembly level.
By operating after compilation, we have full control over memory layout. 
Our approach addresses the following two challenges.

First, we lack high-level semantic and pointer information.
As high-level semantic information is lost during compilation, we need to recover it to identify which instructions operate on secret data and need to be transformed.
A further complication is use of weakly typed pointers and potential pointer aliasing, which is particularly difficult to resolve without explicit type information.
To deal with this challenge, we design a variant of typed assembly language called \systal and an inference algorithm to deduce type information for off-the-shelf \sysisa programs.

Second, we face the challenge of performing assembly transformation under architectural constraints.
Specifically, may not use extra registers, which would require complex register management and register spilling.
To deal with the challenge, we arrange our memory layout to have the secret region (i.e. secret stack) and the original stack (i.e. public stack) maintain a constant distance ($\delta$) from each other.
As a result, redirecting memory accesses between the stacks only requires pointer offsetting by $\delta$.

To illustrate the effectiveness of our mechanism, we revisit the example in Figure~\ref{fig:salsa20}.
In Figure~\ref{fig:salsa20-asm-trans}, we show the type annotations derived by our inference tool. 
For brevity, we only show the memory-related annotations.
Each memory operand is annotated with a \textit{dependent type} that constrains its memory-access range and a \textit{taint type} indicating secrecy.
For example, in line 6, the array base pointer \texttt{rsi}, which references the secret input, is inferred to access the range of $[s, s+64)$ with the taint type as $1$, indicating secrecy.
In line 11, another stack access is annotated with the access range as $[p-88, p-80)$ and is similarly marked as secret.

Transformation should relocate any pointer with taint type of $1$.
For example, in line 11, the stack offset is incremented by $\delta$ to move the write to the secret stack. 
Additionally, the parent function (not shown) adds $\delta$ to the base pointer \texttt{rsi} before passing it as an argument to the callee, ensuring all the accesses within the callee are redirected to the secret stack.

The following sections now go into detail on the main components of our approach: type system (\autoref{sec:tal}), type inference (\autoref{sec:infer}), and transformation (\autoref{sec:trans}).

\section{\systal}
\label{sec:tal}
We propose \systal, a variant of typed assembly language~\cite{morrisett1999system} that helps reason about information flow statically.
The abstract ISA machine for \systal applies taint tracking on registers and memory at the byte level. This machine tracks the secret flow and does not allow executing instructions that transmit tainted values through side channels. 
For example, it gets stuck when executing load/store with tainted addresses or branches with tainted conditions. 

The goal of \systal's type system is to ensure that a well-typed program and the program generated from it by our transformation are
constant-time, thereby never getting stuck on this abstract machine.
It is challenging to reason statically about the program's taint flow, since high-level abstractions such as pointers and array indices are missing in original \sysisa assembly programs.
\systal enriches programs with types that not only constrain the taint status but also bound the values of registers and memory slots.

Furthermore, \systal splits memory into nonoverlapping slots according to the memory layout of the source program, associating a type to each memory slot. In this work, we only consider assembly programs compiled from constant-time C/C++ programs, so each memory slot contains either a scalar, pointer, or array, the last of which can have lengths not known at compile time, thanks to the use of symbolic descriptions of address ranges.

This design choice offers an additional benefit for information-flow tracking. 
Specifically, in cryptographic programs, although each static instruction may access different memory bytes during dynamic execution, it idiomatically only accesses data within the address range corresponding to a specific data object in the source program.
Therefore, each static instruction in \systal programs has fixed registers/memory slots acting as its taint source and destinations, allowing easy regulation of taint flow statically with types.

In the remainder of this section, we first define \systal syntax in Section~\ref{sec:tal-syntax} and typing rules in Section~\ref{sec:tal-typing-rules}. We also formalize type soundness in Section~\ref{sec:tal-soundness}.

\subsection{\systal Syntax}
\label{sec:tal-syntax}

\begin{figure}
    \centering
\begin{equation*}
\begin{array}{rcll}
    \opVar & \Coloneqq & r \mid i \mid \ell \mid \memOp{r_{b}}{r_{i}}{i_{s}}{i_{d}}\opMode{s, \tau} & \emph{Operand} \\
    \mathit{inst} & \Coloneqq & \inst{movq}{\opVar_1,\opVar_0} \mid \inst{leaq}{\opVar_1,\opVar_0} & \emph{Instruction} \\
    & & \mid \inst{addq}{\opVar_1,\opVar_0} \mid \inst{cmpq}{\opVar_1,\opVar_0} \\
    & & \mid \inst{jne}{\ell}\opMode{\sigma} \mid \inst{jmp}{\ell}\opMode{\sigma} \mid \inst{callq}{\ell}\opMode{\sigma_\mathit{call}, \sigma_\mathit{ret}} \\
    & & \mid \inst{retq}{} \mid \inst{halt}{}
    \\
    I & \Coloneqq & \inst{jmp}{\ell}\opMode{\sigma} \mid \inst{retq}{} \mid \inst{halt}{}  & \emph{Instruction}\\
    & & \mid \instVar;I & \emph{sequence}\\
    F & \Coloneqq & \qty{\ell_1: I_1, \dots, \ell_n: I_n, f_\text{ret}: \inst{retq}{}} & \emph{Function}\\
    P & \Coloneqq & \qty{f_1: F_1, \dots, f_n: F_n} & \emph{Program}  \\
    & \\
    R & \Coloneqq & \qty{r_1: (v_1, t_1), \dots} & \emph{Register file}\\
    M & \Coloneqq & \qty{\mathit{addr}_1: (v_1, t_1) \dots} & \emph{Memory}\\
    S & \Coloneqq & (R, M, \pcVar) & \emph{State}\\
    & \\
    e & \Coloneqq & x \mid v \mid \top \mid e_1\oplus e_2 \mid \ominus e & \emph{Dependent type}\\
    \tau & \Coloneqq & x \mid 0 \mid 1 \mid \tau_1\lor \tau_2 & \emph{Taint type}\\
    \beta & \Coloneqq & (e, \tau) & \emph{Basic type}\\
    \RType & \Coloneqq & \qty{r_1: \beta_1, \dots} & \emph{Register type}\\
    \MType & \Coloneqq & \qty{s_1: (\sval_1, \beta_1), \dots} & \emph{Memory type}\\
    \SType & \Coloneqq & (\Delta, \RType, \MType) & \emph{State type}\\
    
    \Gamma & \Coloneqq & \qty{\ell_1: (\Delta_1, \RType_1, \MType_1), \dots} & \emph{Function type}\\
    \PType & \Coloneqq & \qty{f_1: \Gamma_1, \dots} & \emph{Program type}\\
\end{array}
\end{equation*}
    \caption{\systal syntax}
    \label{fig:tal-syntax}
\end{figure}

\pgheading{Program Syntax}
\systal (selected syntax in Figure~\ref{fig:tal-syntax}) is built based on \sysisa.
We require that each basic block end with $\instr{halt}$, $\instr{retq}$, or an unconditional branch.
\systal also requires that each function $f$ (except for the top-level one) has a basic block $f_\text{ret}$ that only contains one return instruction to serve as the unique exit point for the function, which simplifies the typing rules.

\systal also introduces type annotations on load/store operands, branch instructions, and function calls (highlighted in \textcolor{blue}{blue} in Figure~\ref{fig:tal-syntax}). These annotations help to constrain well-typed programs, which will be detailed in Section~\ref{sec:tal-typing-rules}.

\pgheading{Type Syntax}
As mentioned before, an \systal abstract machine, with its machine state denoted as $S=(R, M, \pcVar)$, applies byte-level taint tracking on the registers $R$ and memory $M$
and gets stuck on insecure operations (e.g., load/store with tainted addresses). 

\systal's program type $\PType$ is a map from function names to function types, and a function type $\Gamma$ is a map from the function's basic-block labels to state types $\SType$, which serve as block preconditions.

A type $\SType$ contains three parts: the type context $\Delta$, register-file type $\RType$, and memory type $\MType$.
Specifically, $\Delta$ is a set of constraints that must be satisfied by type variables in $\RType$ and $\MType$.
Partial map $\RType$ assigns register names to their dependent and taint types.
A well-formed program can only read from registers that appear in $\RType$.
Partial map $\MType$ assigns disjoint memory slots ($s$) each to a region of addresses whose contents are initialized ($\sval$) and a type of data found therein. Each slot corresponds to a data object in the source program or a register spill.
\systal tracks pointers in registers and memory using dependent types, and both memory slots $s$ and valid regions $\sval$ are sets of addresses represented by dependent types. Hence, with the dependent types of load/store addresses, \systal can easily track which memory slot is accessed by each instruction.

\subsection{Typing Rules}
\label{sec:tal-typing-rules}
In \systal, program type-correctness is determined by the type-correctness of each function in the program, in turn determined by the type-correctness of each block in the function.
Intuitively, the state type of a basic block (or more generally, an instruction sequence) ensures that the abstract machine whose state satisfies the type constraints can execute the block (instruction sequence) without getting stuck.
When the machine is about to jump to the other block at a branch instruction, its state should also satisfy the target's state type.
Figures~\ref{fig:tal-inst-typing}-\ref{fig:tal-state-subtype} elaborate with typing rules.

\begin{figure}
    \centering
\begin{mathpar}
\inferrule[Typing-Movq-m-r]{
\hasType{\Delta, \RType, \MType}{\text{load}(\memOp{r_b}{r_i}{i_s}{i_d}\opMode{s, \tau}, 8)}{\beta}\\
\RType'=\RType[r_0\mapsto \beta]\\
\hasType{\PType, \Gamma}{I}{{\revision{(\Delta, \RType', \MType)}}}
}{
\hasType{\PType, \Gamma}{\inst{movq}{\memOp{r_b}{r_i}{i_s}{i_d}\opMode{s, \tau}, r_0};I}{(\Delta, \RType, \MType)}
}\and
\inferrule[Typing-Movq-r-m]{
\hasType{\Delta, \RType, \MType}{\text{store}(\memOp{r_b}{r_i}{i_s}{i_d}\opMode{s, \tau}, 8, \RType[r_1])}{(\sval, \beta)}\\
\MType'=\MType[s\mapsto (\sval, \beta)]\\
\hasType{\PType, \Gamma}{I}{(\Delta, \RType, \MType')}
}{
\hasType{\PType, \Gamma}{\inst{movq}{r_1, \memOp{r_b}{r_i}{i_s}{i_d}\opMode{s, \tau}}; I}{(\Delta, \RType, \MType)}
}\and
\inferrule[Typing-Cmpq-r-r]{
\RType[r_1]=(e_1, \tau_1)\\
\RType[r_0]=(e_0, \tau_0)\\
\hasType{\PType, \Gamma}{I}{(\Delta, \text{setFlag}(\RType, (e_0-e_1, \tau_0\lor \tau_1), \inst{cmpq}{}), \MType)}
}{
\hasType{\PType, \Gamma}{\inst{cmpq}{r_1, r_0};I}{(\Delta, \RType, \MType)}
}\and
\inferrule[Typing-Jne]{
\RType[\texttt{ZF}]=(e=0, 0)\\
\revision{\hasVal{\Delta}{\text{isNonChangeExp}(e)}}\\
\revision{\text{getInputVar}(\dom{\sigma})=\emptyset}\\
\revision{\text{getTaintVar}(\dom{\sigma})=\emptyset}\\
\revision{\forall x\in \dom{\sigma}.\;\sigma(x)\neq \top}\\
\hasType{\PType, \Gamma}{I}{(\Delta\cup \qty{e=0}, \RType, \MType)}\\
\Gamma(\ell)=(\Delta', \RType', \MType')\\
\dom{\MType'}=\dom{\MType}\\
(\Delta \cup \qty{e\neq 0}, \RType, \MType) \sqsubseteq \sigma(\Delta', \RType', \MType')
}{
\hasType{\PType, \Gamma}{\inst{jne}{\ell\opMode{\sigma}};I}{(\Delta, \RType, \MType)}
}\and
\inferrule[Typing-Callq]{
\revision{e=\spVar+c}\\
\RType[r_\texttt{rsp}]=(e, 0)\\
\revision{\MType[e-8, e]=(\emptyset, (\_, 0))}\\
\revision{\forall x, \sigmaCall(x)=e, \text{isPtr}(x).\; \hasVal{\Delta}{\text{isNonChangeExp}(e - \text{getPtr}(e))}}\\
\revision{\forall x\in \dom{\sigmaCall}.\;\sigmaCall(x)\neq \top}\\
\sigma_\text{ret}=\overrightarrow{x_1}\rightarrow\overrightarrow{x_2}\\
\RType_{p_0}=\RType[r_\texttt{rsp}\mapsto (e- 8, 0)]\\
(\Delta, \RType_{p_0}, \MType) \sqsubseteq \sigma_\text{call}(\PType(f)(f))\\
\overrightarrow{x_2}\not\in(\Delta, \RType, \MType)\\
(\Delta_{p_1}, \RType_{p_1}, \MType_{p_1})=(\sigma_\text{call}\cup \sigma_\text{ret})(\PType(f)(f_\text{ret}))\\
\hasType{\PType, \Gamma}{I}{(\Delta\cup \Delta_{p_1}, \RType_{p_1}[r_\texttt{rsp}\mapsto (e, 0)], \text{updateMem}(\MType, \MType_{p_1}))}
}{
\hasType{\PType, \Gamma}{\inst{callq}{f\opMode{\sigma_\text{call},\sigma_\text{ret}}};I}{(\Delta, \RType, \MType)}
}
\end{mathpar}
    \caption{Instruction-sequence typing}
    \label{fig:tal-inst-typing}
\end{figure}

In general, each of the instruction-sequence typing rules is structured as follows.
First, the instruction sequence's state type should provide enough constraints so that the abstract machine can execute the first instruction in the sequence safely without getting stuck.
Second, the rule derives new type constraints for the machine state after executing the first instruction.
It requires that the next instruction sequence to be executed is well-typed with respect to the derived state type.
Some examples illustrate the pattern.

\textbf{\typingMovqMR} constrains a load instruction to be memory-safe and constant-time, via a type annotation tracking taint status of the load data.
It invokes \textsc{Typing-Load} in Figure~\ref{fig:tal-operand-typing},
which requires that the load range fall in the initialized region in $s$ ($s_\text{addr}\subseteq \sval$) for memory safety, and the taint type of data in the slot must satisfy $\tau$.
\systal also requires the load address to be untainted. 

\revision{
\textbf{\typingMovqRM} also constrains store instructions to be memory-safe and constant-time.
It invokes \textsc{Typing-StoreOp-Spill} or \textsc{Typing-StoreOp-Non-Spill} depending on whether the store slot holds a spill or a data object in the source code.
The difference arises from the different lifetimes of the two types of slots.
}

When storing to a spill slot, as shown in \textsc{Typing-StoreOp-Spill}, \systal always derives the type for the next state by overwriting the slot's valid region and type with the store range and store data type.
Even for a partial store, the type system ``forgets'' the type for the data in the part of the slot that is not overwritten by the operand. 

When storing to a nonspill slot, as shown in \textsc{Typing-StoreOp-Non-Spill}, \systal requires the store data's taint status to satisfy the original taint type of the target memory slot. It also updates the valid region and dependent type by combining the store data and the existing data in the slot. For a partial store, \systal may only consider the updated dependent type as $\top$ for simplicity.
Since each nonspill slot corresponds to a data object in the source file, we impose uniformity on the slot's taint type during its lifetime, i.e., the whole function. Then, we can use its taint type as a hint for its target placement during transformation, to change all load/store operands accessing it accordingly.
We do not require the unified taint for register spills since we consider the register spill lifetime ends after the next register spill (or store) to the same slot.

\begin{figure}
    \centering
\begin{mathpar}
\inferrule[Typing-Addr]{
\RType[r_b]=(e_b, \tau_b)\\
\RType[r_i]=(e_i, \tau_i)\\
}{
\hasType{\RType}{\memOp{r_b}{r_i}{i_s}{i_d}}{(
e_b+e_i\times i_s+i_d
, \tau_b \lor \tau_i)}
}\and
\inferrule[Typing-Load]{
\hasType{\RType}{\memOp{r_b}{r_i}{i_s}{i_d}}{(e_\text{addr}, 0)}\\
\revision{\hasVal{\Delta}{\text{isNonChangeExp}(e_\text{addr}-\text{getPtr}(s))}}\\
s_\text{addr}=[e_\text{addr}, e_\text{addr}+c)\\
\MType[s]=(\sval, (e, \tau))\\
\hasVal{\Delta}{s_\text{addr}\subseteq \sval}\\
e'=(\hasVal{\Delta}{s_\text{addr}=\sval})\;?\; e : \top\\
\revision{e'=\top \Rightarrow (\hasVal{\Delta}{\text{isNonChangeExp}(e)})}
}{
\hasType{\Delta, \RType, \MType}{\text{load}(\memOp{r_b}{r_i}{i_s}{i_d}\opMode{s, \tau}, c)}{(e', \tau)}
}\and
\inferrule[Typing-StoreOp-Spill]{
\hasType{\RType}{\memOp{r_b}{r_i}{i_s}{i_d}}{(e_\text{addr}, 0)}\\
\revision{\hasVal{\Delta}{\text{isNonChangeExp}(e_\text{addr}-\text{getPtr}(s))}}\\
s_\text{addr} = [e_\text{addr}, e_\text{addr}+c)\\
s\in\dom{\MType}\\
\text{isSpill}(s)\\
\hasVal{\Delta}{s_\text{addr} \subseteq s}\\
\hasVal{\Delta}{\tau_1 \Rightarrow \tau}\\
}{
\hasType{\Delta, \RType, \MType}{\text{store}(\memOp{r_b}{r_i}{i_s}{i_d}\opMode{s, \tau}, c, (e, \tau_1))}{(s_\text{addr}, (e, \tau))}
}\and
\inferrule[Typing-StoreOp-Non-Spill]{
\hasType{\RType}{\memOp{r_b}{r_i}{i_s}{i_d}}{(e_{addr}, 0)}\\
s_\text{addr} = [e_\text{addr}, e_\text{addr}+c)\\
\revision{\hasVal{\Delta}{\text{isNonChangeExp}(e_\text{addr}-\text{getPtr}(s))}}\\
\MType[s]=(\sval, (e_0, \tau))\\
\neg \text{isSpill}(s)\\
\hasVal{\Delta}{s_\text{addr} \subseteq s}\\
\hasVal{\Delta}{\tau_1 \Rightarrow \tau}\\
e'=(\hasVal{\Delta}{\sval \subseteq s_\text{addr}})\;?\; e_1 : \top\\
\revision{e'=\top \Rightarrow (\hasVal{\Delta}{\text{isNonChangeExp}(e_0)\land \text{isNonChangeExp}(e_1)})}
}{
\hasType{\Delta, \RType, \MType}{\text{store}(\memOp{r_b}{r_i}{i_s}{i_d}\opMode{s, \tau}, c, (e_1, \tau_1))}{(s_\text{addr}\cup \sval, (e', \tau))}
}
\end{mathpar}
    \caption{Memory-operation typing \revision{(note that, via predicate isSpill, we rely on debug tables already generated by Clang)}}
    \todo{Important note: If checking \textsc{State-Subtype} and combining with subtype relation at the call site, one may notice that in each function type, the slot reserved for its callee's stack is also denoted/considered as a ``spill'' slot (or a slot that may contain spill). This is a detailed decision that I have to make to ensure functional correctness. Consider improving the name isSpill later.}
    \label{fig:tal-operand-typing}
\end{figure}

\textbf{\typingJne} specifies the rule for a conditional-branch instruction.
\shixin{removed flag explanation, is it good?}
\systal requires the flag holding the branch condition to be untainted 
so that the program is constant-time.
Furthermore, \systal also tracks whether each dependent type refers to pointer values that might be changed by our transformation (see Section~\ref{sec:trans-strategies} for details).
To guarantee functional correctness of the transformation, \systal requires that the branch condition is independent from these pointer values, denoted as $\text{isNonChangeExp}(e)$.
Then, \systal derives the next state types after executing the branches, including both cases where the branch is taken and not taken.

For the not-taken side, similar to previous cases for non-branch instructions, \systal derives the next state type by adding the negation of the branch condition (i.e., $e=0$) to the type constraints.

For the taken side, \systal derives the next state type by asserting the branch condition (i.e., $e\neq 0$). The primary goal is to ensure that the machine state at the branch instruction is well-formed to jump to the target block. We define the subtype judgment for state types as shown in Figure~\ref{fig:tal-state-subtype}. Intuitively, this judgment ensures that for any machine state $S$ that satisfies a state type $\SType_1$, if $\SType_1$ is a subtype of $\SType_2$, then $S$ must also satisfy $\SType_2$.
\typingJne requires that the state type for the branch's taken side is a subtype of the target block's type $\Gamma(\ell)$.
Note that in this rule, we are checking the subtype relation against $\sigma(\Gamma(\ell))$, where the branch annotation $\sigma$ is a substitution that instantiates type variables in $\Gamma(\ell)$ using expressions over variables in the current block's type context.
We use $\sigma(\cdot)$ as syntax sugar for applying the substitution $\sigma$ to a variety of syntactic objects.
Our type checker implements every entailment check $\Delta \vdash \ldots$ as a call to an SMT solver.
In this rule, \systal also has some extra constraints on $\sigma$ to guarantee type safety and transformation correctness, which will be detailed in \bothver{Appendix~\ref{sec:supp-tal}}{\cite{secsep-proof}}.

\textbf{\typingCallq} specifies type constraints and changes of each step of calling a function.
It first derives the state type $(\Delta, \RType_{p_0}, \MType)$ after pushing the return address, checking that the state type is a subtype of the callee function's first block type $\PType(f)(f)$ with respect to the function call's annotation $\sigma_\text{call}$.
Here, $\sigma_\text{call}$ represents the type-variable substitution between the callee and the caller.

Next, \systal derives the state type after returning from the callee, using the type of the callee's exit block.
There are several details to note.
First, we need to convert the return-state type represented under the callee's type context to the caller's context.
Compared to the callee's first block type $\PType(f)(f)$, its return-state type $\PType(f)(f_\text{ret})$ may introduce new type variables.
The type annotation $\sigma_\text{ret}$ maps these new variables to the caller's context.
Hence, we perform type-variable substitution using both substitutions to represent the return-state type for the caller, i.e., $(\Delta_{p_1}, \RType_{p_1}, \MType_{p_1})=(\sigma_\text{call}\cup \sigma_\text{ret})(\PType(f)(f_\text{ret}))$.
We then add the return state's type constraints $\Delta_{p_1}$ to the next state type's context.
Second, the callee's return-state type only specifies how it updates the memory region covered by its memory type, which is a subset of the memory region covered by the caller's memory type.
On the other hand, according to our typing rules for load and store operations, the memory regions that do not belong to the callee's memory type remain unchanged across the function call. 
Following this philosophy, we apply the callee's changes to memory slots to the parent's memory type to get the final memory type after return ($\text{updateMem}(\MType, \MType_{p_1})$).
We also pop the return address to get the final return-state type.

\begin{figure}
    \centering
\begin{mathpar}
\inferrule[Reg-Subtype]{
\hasVal{\Delta}{e_1=e_2 \lor (\revision{\text{isNonChangeExp}(e_1)} \land e_2=\top)}\\
\hasVal{\Delta}{\tau_1 \Rightarrow \tau_2}\\
}{
\hasVal{\Delta}{(e_1, \tau_1)\sqsubseteq(e_2, \tau_2)}
}\and
\inferrule[Mem-Slot-Subtype]{
\hasVal{\Delta}{s_2\subseteq s_1}\\
\revision{\hasVal{\Delta}{\text{isSpill}(s_2)\Rightarrow \text{isSpill}(s_1)}}\\
\hasVal{\Delta}{e_1=e_2\lor (\revision{\text{isNonChangeExp}(e_1)}\land e_2=\top) \lor s_2=\emptyset}\\
\hasVal{\Delta}{\tau_1=\tau_2 \lor (\text{isSpill}(s_1) \land s_2=\emptyset)}\\
\revision{\text{getPtr}(s_1)=\text{getPtr}(s_2)}
}{
\hasVal{\Delta}{(s_1, (e_1, \tau_1))\sqsubseteq (s_2, (e_2, \tau_2))}
}\and
\inferrule[State-Subtype]{
\hasVal{\Delta_1}{\Delta_2}\\
\forall r\in\dom{\RType_2}.\; \hasVal{\Delta_1}{\RType_1[r]\sqsubseteq \RType_2[r]}\\
\forall s_2\in \dom{\MType_2}.\; \exists s_1.\; \hasVal{\Delta_1}{(s_2\subseteq s_1\land \MType_1[s_1]\sqsubseteq \MType_2[s_2])}
}{
\hasVal{}{(\Delta_1, \RType_1, \MType_1)\sqsubseteq (\Delta_2, \RType_2, \MType_2)}
}
\end{mathpar}
    \caption{State subtyping}
    \label{fig:tal-state-subtype}
\end{figure}

\subsection{Type Soundness}
\label{sec:tal-soundness}

In this section, we justify the type safety of \systal programs, that is, \systal guarantees well-typed programs to be executed on an \systal abstract machine without getting stuck. 
First, we briefly explain well-formedness of \systal abstract machine states (see details in \bothver{Appendix~\ref{sec:tal-machine}}{\cite{secsep-proof}}).
A state $S$ is well-formed, i.e. $\typeChecked{P, \PType}{S}$, if all its registers and memory values satisfy constraints specified by the state type of the instruction sequence to be executed next.

\begin{theorem}[Type Safety]
\label{thm:type-safety}
If $\typeChecked{P, \PType}{S}$, then for some $S'$, $S\rightarrow S'$ and $\typeChecked{P, \PType}{S'}$; or $S$ is a termination state.
\end{theorem}

\section{Type Inference}
\label{sec:infer}
In this section, we introduce our type-inference algorithm that generates types for assembly programs.
Note that the inference algorithm is heuristic and does not guarantee type correctness. Instead, the correctness is checked separately by applying typing rules introduced in Section~\ref{sec:tal}.
Recalling the definitions from Section~\ref{sec:tal-syntax}, we need to generate state types of basic blocks and type annotations on instructions.
Our type-inference algorithm consists of three parts. First, we introduce unification type variables to represent state types and type annotations.
Second, we plug the type expressions into \systal's typing rules to collect type constraints.
Then, the third step is to solve for arithmetic predicates on type variables, which will be used to enrich the $\Delta$ of each block's state type so that it satisfies the typing constraints.
We iterate a process of learning new typing information and exploring its implications.

\subsection{Type initialization}
We begin type inference by using type-unification variables to represent state types (i.e., $(\Delta, \RType, \MType)$) and type annotations (i.e., load/store's destination-slot taint annotations; type-variable substitutions for branches and calls).
Our goal is to add appropriate constraints on these type variables to the type context $\Delta$ so that the state types and annotations satisfy the typing rules in Section~\ref{sec:tal-typing-rules}.

For register types, we simply assign a unification variable to each register; and for type annotations, we follow a similar strategy.
To initialize memory typing $\MType$, we first need to figure out $\dom{\MType}$ for each function.

\pgheading{Determine Memory Layout}
Core cryptographic routines usually do not allocate memory on the heap dynamically for reasons of performance,
so we consider the following three kinds of memory slots to determine each function's memory layout:
(1) data objects referenced by pointers in the function arguments; (2) local stack referenced by the stack pointer; (3) global variables referenced by global pointers.

We require simple type annotations in C source code, implemented through our custom annotation system, to explain the relationships among function parameters, e.g., one gives the size of an array that another points to \revision{(see example in Figure~\ref{fig:source-code-annotation})}.
These annotations are compiled down to assembly and serve as provided specifications for functions.
However, we must infer type information for other basic blocks within each function.
Other information can be determined with simple compiler support: address ranges of function stack frames using a Clang pass, and locations of global variables indicated directly in assembly code.

\begin{figure}
    \centering
\begin{minted}{c}
/**
 * @anno message : @size(mlen), @valid(0, mlen);
 * @anno mlen    : @taint[0];
 */
void foo(uint8_t *message, uint64_t mlen) { ... }
\end{minted}
    \caption{\revision{Example type annotations in C source code, which means that
    (1) the pointer \texttt{message} points to an array with size \texttt{mlen}, and the whole array is initialized;
    (2) \texttt{mlen} is a public variable whose taint type is $0$.
    }}
    \label{fig:source-code-annotation}
\end{figure}

\subsection{Type-Constraint Generation}
\label{sec:infer-constraint-gen}

In this section, we describe rules to generate constraints on the initialized block-state types. 
We provide several example rules in Figure~\ref{fig:infer-constraints}, where the generated constraints are \highlight{highlighted}.
Given a state type and the corresponding instruction sequence, the constraint-generation rule consists of two parts, following a similar structure to \systal typing rules.

First, a rule generates constraints on the state type so that the current instruction executes safely.
For example, the first two rules in Figure~\ref{fig:infer-constraints} constrain that a load/store operation must access a memory slot from the memory type, and the address must be untainted.
The third rule requires that the branch condition is untainted.

Second, a rule derives the next state type after executing the first instruction in the sequence and generates constraints for the next type.
Note that the state types are initialized using unification variables not constrained by predicates,
so we may not be able to derive the next state type deterministically.
For example, as shown in \textsc{Constraint-Movq-r-m-Unknown}, we cannot determine the target slot of the store operand and thereby are not able to update the memory type correspondingly.
In this case, we use the unmodified memory type to generate constraints for the next instruction sequence.
Note that these heuristic rules cannot generate all proper type constraints. We rely on these partially correct constraints to derive predicates and use the newly solved predicates to improve constraint generation in the next round.

\begin{figure}
    \centering
\begin{mathpar}
\inferrule[Constraint-Movq-m-r-Unknown]{
\hasType{\RType}{\memOp{r_b}{r_i}{i_s}{i_d}}{(e_a, \tau_a)}\\
\RType'=\RType[r_0\mapsto (\textcolor{red}{\top}, \tau)]\\
\hasType{\PType, \Gamma}{I}{(\RType', \MType, \Delta)}
\constraint{C}
}{
\hasType{\PType, \Gamma}{\inst{movq}{\memOp{r_b}{r_i}{i_s}{i_d}\opMode{s, \tau}, r_0};I}{(\Delta, \RType, \MType)}\\
\constraint{[e_a, e_a+8]\subseteq s; s\in \dom{\MType}; \tau_a=0; C}
}\and
\inferrule[Constraint-Movq-r-m-Unknown]{
\hasType{\RType}{\memOp{r_b}{r_i}{i_s}{i_d}}{(e_a, \tau_a)}\\
\hasType{\PType, \Gamma}{I}{(\Delta, \RType, \textcolor{red}{\MType})}
\constraint{C}
}{
\hasType{\PType, \Gamma}{\inst{movq}{r_1, \memOp{r_b}{r_i}{i_s}{i_d}\opMode{s, \tau}}; I}{(\Delta, \RType, \MType)}\\
\constraint{[e_a, e_a+8]\subseteq s; s\in \dom{\MType}; \tau_a=0; C}
}\and
\inferrule[Constraint-Jne]{
\RType[\texttt{ZF}]=(e=0, \tau)\\
\hasType{\PType, \Gamma}{I}{(\Delta\cup \qty{e=0}, \RType, \MType)}
\constraint{C}\\
}{
\hasType{\PType, \Gamma}{\inst{jne}{\ell\opMode{\sigma}};I}{(\Delta, \RType, \MType)}\\
\constraint{
\tau=0; (\Delta \cup \qty{e\neq 0}, \RType, \MType) \sqsubseteq \sigma(\Gamma(l)); C}
}
\end{mathpar}
    \caption{Typing-constraints generation}
    \label{fig:infer-constraints}
\end{figure}

\subsection{Dependent-Type Inference}

In this section, we show how we derive arithmetic predicates of dependent type variables from type constraints. As shown in Figure~\ref{fig:infer-constraints}, we generate two kinds of constraints for dependent types: (1) state-subtype and (2) load/store-address constraints.

\subsubsection{Solving Subtype Constraints}
\systal requires that the state type at each branch should be a subtype of the target block's state type (e.g., ($\Delta\cup \qty{e\neq 0}, \RType, \MType)\sqsubseteq \sigma(\Gamma(\ell))$ in \textsc{Constraint-Jne}).
Intuitively, the subtype relation requires that the range of each register/memory slot's value at the target block, represented by dependent type variables, should be a superset of the range of its value at the branch that jumps to the target block.
By unfolding all subtype constraints, we can get concrete constraints on the range of each dependent type variable. 
We propose a set of inference rules that syntactically applies to the range constraints with certain patterns and solves the predicates of each variable heuristically. 

Our heuristic rules for solving subtype constraints are built upon several common patterns we observed from the range constraints on each type variable.
Recall we use dependent types to reason about load/store operations, so we mainly focus on type variables used for pointer arithmetic.
Luckily, we target type inference for cryptographic programs, whose dependent-type range constraints share simple and intuitive patterns corresponding to two basic code patterns in Figure~\ref{fig:possible-value-code-pattern}.
In both examples, we demonstrate applying our rules to figure out the range for type variable $a$ that represents \reg{rax}'s dependent type at block \texttt{.L0}. We denote the range of $a$ as $S_a$ and derive constraints on $S_a$ by unfolding all state subtype constraints.

\pgheading{Infer set of values}
The first example (Figure~\ref{fig:code-pattern-set}) shows the case where \reg{rax} contains different values when entering basic block \texttt{.L0} from different branches. Specifically, the range of \reg{rax} is $\qty{e_1}$ when jumping from \texttt{.L1} and is $\qty{e_2}$ when jumping from \texttt{.L2}.
The subset constraint and the derived solution can be formulated as follows:
\begin{equation*}
    \left.
    \begin{array}{@{}l@{}}
        S_a \supseteq \qty{e_1} \\
        S_a  \supseteq \qty{e_2} \\      
    \end{array}
    \right\} \Rightarrow S_a = \qty{e_1, e_2}.
\end{equation*}

\pgheading{Infer range of loop counter}
The second example (Figure~\ref{fig:code-pattern-loop}) shows the case where \reg{rax} acts as a loop counter.
$\reg{rax}$ is initialized to $e_0$ when entering the loop body from block \texttt{.L1} and increased by a constant step $c_0$ in each iteration. The loop ends when \reg{rax} is equal to the boundary value $e_n$. Without loss of generality, we discuss the case where $c_0>0$.
According to our constraint-generation rules, when jumping back to the loop head \texttt{.L0}, the state type satisfies $\Delta=\qty{a\in S_a, a+c_0-e_n\neq 0}$ and $\RType=\qty{\reg{rax}:a+c_0}$.
So we can use the set $\qty{a+c_0: a\in S_a\land a+c_0-e_n\neq 0}$ to represent the range of \reg{rax} before jumping back. 
Thus, the subtype constraint and the corresponding heuristic rule can be formulated as follows:
\begin{equation*}
    \left.
    \begin{array}{@{}l@{}}
        S_a \supseteq \qty{e_0} 
        \qquad \qquad c_0>0\\
        S_a \supseteq \qty{a+c_0: a\in S_a\land a+c_0\neq e_n} \\
    \end{array}
    \right\} \Rightarrow S_a = \typeRange{e_0}{e_n-c_0}{c_0}.
\end{equation*}
This rule extracts three key features from the constraints:
\begin{itemize}
    \item $e_0$: the loop counter's base value at the loop's entrance;
    \item $c_0$: the per-iteration step value for the counter;
    \item $e_n$: the loop boundary in the branch condition.
\end{itemize}
Then, the rule heuristically determines that the range of $a$ is $S_a=\typeRange{e_0}{e_n-c_0}{c_0}$.
Here, we use $\typeRange{a}{b}{c}$ to represent a set of values in range $[a, b]$ with stride $c$.
In our implementation, we apply the above strategy to infer a range of loop counters. 

\pgheading{Infer implicit relation between variables}
Another challenge is that assembly programs do not explicitly keep semantic relations between type variables.
However, these relations are crucial to deriving accurate range constraints for variables.
For example, in Figure~\ref{fig:code-pattern-loop}, \reg{rax} and \reg{rbx} are increased consistently during each loop iteration, but the loop condition only constrains the boundary of \reg{rax} when jumping back to the loop header.
By unfolding the subtype constraints, we can only get the following constraints related to $b$: $S_b \supseteq \qty{e_1}$ and $S_b\supseteq \qty{b+c_1: b\in S_b}$, which implies that $S_b$ is infinite, thereby not accurately constraining the range of $b$.

As a solution, we introduce another rule that infers the linear relation between type variables that share similar constraint patterns.
In our example, \reg{rax} and \reg{rbx} are increased synchronously following the same loop structure, so we can use the range of \reg{rax} to constrain the range of \reg{rbx}, as shown in the following formula.
\begin{equation*}
    \left.
    \begin{array}{@{}l@{}}
        S_b \supseteq \qty{e_1} \\
        S_b \supseteq \qty{b+c_1: b\in S_b} \\
        S_a \supseteq \qty{e_0} \\
        S_a \supseteq \qty{a+c_0: a\in S_a\land a+c_0\neq e_n} \\
    \end{array}
    \right\} \Rightarrow 
    \begin{array}{@{}l@{}}
        S_b = \\
        \qty{\frac{(a-e_0) c_1}{c_0}+e_1: a\in S_a}.
    \end{array}
\end{equation*}

\newsavebox{\possiblevaluesetbzero}
\begin{lrbox}{\possiblevaluesetbzero}
\begin{minipage}{5cm}
\begin{minted}[linenos=false,xleftmargin=0pt]{asm}
.L0:
# |$\Delta=\qty{a \in S_a}$|
# |$\RType=\qty{\reg{rax}: a}$|
\end{minted}
\end{minipage}
\end{lrbox}

\newsavebox{\possiblevaluesetbone}
\begin{lrbox}{\possiblevaluesetbone}
\begin{minipage}{5cm}
\begin{minted}[linenos=false,xleftmargin=0pt]{asm}
.L1:
  movq $|$e_1$|, %rax
  jmp .L0
\end{minted}
\end{minipage}
\end{lrbox}

\newsavebox{\possiblevaluesetbtwo}
\begin{lrbox}{\possiblevaluesetbtwo}
\begin{minipage}{5cm}
\begin{minted}[linenos=false,xleftmargin=0pt]{asm}
.L2:
  movq $|$e_2$|, %rax
  jmp .L0
\end{minted}
\end{minipage}
\end{lrbox}

\newsavebox{\possiblevalueloopbzero}
\begin{lrbox}{\possiblevalueloopbzero}
\begin{minipage}{5cm}
\begin{minted}[linenos=false,xleftmargin=0pt]{asm}
.L0:
# |$\Delta=\qty{a \in S_a, b\in S_b}$|
# |$\RType=\qty{\reg{rax}: a, \reg{rbx}: b}$|
  addq $|$c_0$|, %rax
  addq $|$c_1$|, %rbx
  cmpq $|$e_n$|, %rax
  jne .L0
\end{minted}
\end{minipage}
\end{lrbox}

\newsavebox{\possiblevalueloopbone}
\begin{lrbox}{\possiblevalueloopbone}
\begin{minipage}{5cm}
\begin{minted}[linenos=false,xleftmargin=0pt]{asm}
.L1:
  movq $|$e_0$|, %rax
  movq $|$e_1$|, %rbx
  jmp .L0
\end{minted}
\end{minipage}
\end{lrbox}

\begin{figure}
\centering
\begin{subfigure}{0.5\linewidth}
\centering
    \begin{tikzpicture}[
        square/.style={draw=black,minimum width=1cm,minimum height=0.6cm},
        node distance=0.2cm,>={latex},semithick
    ]
        \node[square] (b1) {
            \begin{minipage}{1.9cm}
                \usebox{\possiblevaluesetbone}
            \end{minipage}
        };
        \node[square,right=of b1] (b2) {
            \begin{minipage}{1.9cm}
                \usebox{\possiblevaluesetbtwo}
            \end{minipage}
        };
        \node[square,anchor=north] (b0) at ($0.5*(b1.south)+0.5*(b2.south)+(0,-0.6)$) {
            \begin{minipage}{2.1cm}
                \usebox{\possiblevaluesetbzero}
            \end{minipage}
        };
        \draw[->] (b1.south) -- ++(0,-0.3) -| (b0.north -| b1.east);
        \draw[->] (b2.south) -- ++(0,-0.3) -| (b0.north -| b2.west);
        \draw[opacity=0] (b0.south) -- ++(0,-0.3);
    \end{tikzpicture}
    \caption{Infer set of values}
    \label{fig:code-pattern-set}
\end{subfigure}%
\begin{subfigure}{0.5\linewidth}
\centering
    \begin{tikzpicture}[
        square/.style={draw=black,minimum width=1cm,minimum height=0.6cm},
        node distance=0.2cm,>={latex},semithick
    ]
        \node[square] (b1) {
            \begin{minipage}{2.45cm}
                \usebox{\possiblevalueloopbone}
            \end{minipage}
        };
        \node[square,anchor=north] (b0) at ($(b1.south)+(0,-0.6)$) {
            \begin{minipage}{2.45cm}
                \usebox{\possiblevalueloopbzero}
            \end{minipage}
        };
        \draw[->] (b1) -- (b0);
        \draw[->] (b0.south) -- ++(0,-0.3) -- ($(b0.south east)+(0.3,-0.3)$) -- ($(b0.north east)+(0.3,0.3)$) -| ($(b0.north)+(0.3, 0)$);
    \end{tikzpicture}
    \caption{Infer loop counter}
    \label{fig:code-pattern-loop}
\end{subfigure}
\caption{Examples for dependent-type inference}
\label{fig:possible-value-code-pattern}
\end{figure}
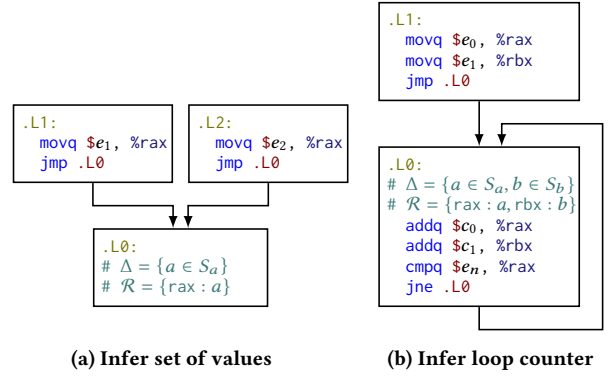

\subsubsection{Solving load/store constraints}
\systal also requires that the dependent type of each load/store address belong to a specific memory slot.
These constraints can be satisfied automatically with the predicates derived from the subtype relation when the load/store address and the memory slot have simple formulas, e.g., shifted from the base pointer by a constant offset.
However, when accessing an array with a variable length or a variable index, we need extra predicates for bounds checks regarding the length/index type variables,
inferred by the following two methods.

\begin{figure}
\centering
\begin{subfigure}{0.4\linewidth}
\centering
\begin{minted}{asm}
foo:
# |$\RType = \qty{\reg{rdi}: p, \reg{rsi}: n}$|
# |$\MType = \qty{[p, p+n]: \_}$|
  cmpq $8, %rsi
  jae .L0
  retq
.L0:
# Missing predicate:
# |$\Delta = \qty{n \geq 8}$|
  movq (%rdi), %rax
\end{minted}
    \caption{Boundary check}
    \label{fig:infer-extra-boundary-check}
\end{subfigure}%
\begin{subfigure}{0.48\linewidth}
\centering
\begin{minted}{asm}
foo:
# foo (uint64_t p[8], 
#      uint64_t k);
# |$\RType = \qty{\reg{rdi}: p, \reg{rsi}: k}$|
# |$\MType = \qty{[p, p+64]: \_}$|
# Missing predicate:
# |$\Delta = \qty{k\in [0, 7]}$|
  movq (%rdi, %rsi, 8), %rax

| |
\end{minted}
    \caption{Implicit assumption}
    \label{fig:infer-extra-boundary-no-check}
\end{subfigure}
\caption{Missing predicates to validate memory accesses}
\shixin{Any suggestions for the titles?}
\label{fig:infer-extra}
\end{figure}

\pgheading{Propagate branch conditions}
First, the function may already include proper bounds checks to guarantee memory safety.
For example, as shown in Figure~\ref{fig:infer-extra-boundary-check}, line 10 loads from $[p, p+8]$, and we lack the predicate $n\geq 8$ to validate it.
On the other hand, the program checks the branch condition $n\geq 8$ before jumping to \texttt{.L0}, which implies this missing predicate.
Motivated by this common pattern, we propose the following rule to deduce predicates by propagating branch conditions across basic blocks.
\begin{equation*}
    \left.
    \begin{array}{@{}l@{}}
        \sigma_1(\Delta, \RType, \MType) \sqsupseteq (\qty{\sigma_1(e)}\cup \Delta_1, \RType_1, \MType_1)\\
        \sigma_2(\Delta, \RType, \MType) \sqsupseteq (\qty{\sigma_2(e)}\cup \Delta_2, \RType_2, \MType_2)\\
        \dots \\       
    \end{array}
    \right\} \Rightarrow e \in \Delta
\end{equation*}
Each subtype constraint listed here corresponds to one branch that jumps to the specific block with state type $(\Delta, \RType, \MType)$.
This rule states that for all branches that jump to this block, if a predicate is always satisfied before branching, then it can be added to the block's type context.

\pgheading{Reverse-engineer load/store operations}
However, not all missing predicates can be deduced from branch conditions in the function.
For example, as shown in Figure~\ref{fig:infer-extra-boundary-no-check}, the function takes two inputs: pointer $p$ to an array with $8$ entries and index $k$.
The programmer implicitly assumes that the function is only called with $k\in [0, 7]$ and loads from the $k$th entry of $p$ without performing any boundary checks.
We propose a two-step method to infer these implicit assumptions by reverse-engineering the necessary conditions to validate the memory safety of load/store operations.

First, for each load/store address without any known target memory slot, we heuristically guess which slot it belongs to based on its address pattern. For example, it is expected to belong to a memory slot that shares the same base pointer.

Second, we constrain the load/store operation to fall in the slot by adding the corresponding predicates to the current block's state type.
As required by subtype constraints, this newly generated predicate must also be satisfied by every previous basic block that jumps to the current one. Hence, we apply the following rule to propagate each newly generated predicate to the previous blocks.
\begin{equation*}
    \left.
    \begin{array}{@{}l@{}}
        \sigma_1(\qty{e}\cup \Delta, \RType, \MType) \sqsupseteq (\Delta_1, \RType_1, \MType_1)\\
        \sigma_2(\qty{e}\cup \Delta, \RType, \MType) \sqsupseteq (\Delta_2, \RType_2, \MType_2)\\
        \dots \\       
    \end{array}
    \right\} \Rightarrow 
    \left\{
    \begin{array}{@{}l@{}}
        \sigma_1(e)\in \Delta_1\\
        \sigma_2(e)\in \Delta_2\\
        \dots \\
    \end{array}
    \right.
\end{equation*}

Note that both inference strategies require us to substitute local type variables properly for each basic block, i.e., to know the branch annotation $\sigma$. This type-variable substitution can be built by unifying each register and memory slot's type from the target block's state type and the state type before branching.

\subsection{Valid-Region Inference}
In this section, we explain how to infer valid regions of each basic block's state type, which is constrained by two aspects:
(1) each load instruction can only read from valid regions (\textsc{Typing-Load});
(2) each memory slot's valid region at a branch instruction must be a superset of its valid region at the destination block (\textsc{Mem-Slot-Subtype}).
Our overall inference strategy is to constrain the valid region of each memory slot using (2) and find the most accurate solution that covers the largest valid region to satisfy (1).
Specifically, the second constraint can be derived from subtype constraints. For example, for $\sigma_1(\Delta, \RType, \MType)\sqsupseteq (\Delta_1, \RType_1, \MType_1)$ and memory slot $s$ where $\MType[s]=(\sval, \_)$, $\MType_1[s]=(\sval_1, \_)$, the constraint on $\sval$ is $\sigma_1(\sval)\subseteq \sval_1$.

For a memory slot that is fully initialized at the beginning of the function, its valid region is always equal to its address range.
It is also straightforward to infer the valid region for a memory slot that holds a primitive type of data (e.g., \texttt{int}) or a register spill since the program usually writes to the full slot or leaves the full slot uninitialized.
Hence, its valid region is usually the slot address range or the empty set.
The major challenge is to infer the valid region for an array, where the program writes to part of it at a time, steadily increasing its valid region. We provide heuristic rules to represent the valid region accurately using dependent type variables. 

For example, as shown in Figure~\ref{fig:infer-range-loop}, the program fills an array by looping over all its entries.
We can derive the following constraints on the array's valid region at \texttt{.L0}.
\begin{equation*}
    \left.
    \begin{array}{@{}l@{}}
        \qty[0/a]\sval \subseteq \emptyset \qquad \qquad a\in [0, 63]\\
        \qty[\textcolor{red}{a+1}/\textcolor{blue}{a}](\sval) \subseteq \sval \cup [p+\textcolor{blue}{a}, p+\textcolor{red}{a+1}]\\
    \end{array}
    \right\} \Rightarrow \sval=[p, p+\textcolor{blue}{a}]
\end{equation*}
Our inference algorithm extracts the valid region's boundary from the pattern $\sval\cup [p+a, p+a+1]$. The key insight is that the next array write is always to the next uninitialized slot.

\begin{figure}
    \centering
\begin{minted}{asm}
.L1: # |$\RType=\qty{\reg{rdi}: p}$|, |$\MType=\qty{[p, p+64]: (\emptyset, \_)}$|
  movq $0, %rax
  jmp .L0
.L0: # |$\RType=\qty{\reg{rdi}: p, \reg{rax}: a}$|, |$\MType=\qty{[p, p+64]: (\sval, \_)}$|
     # |$\Delta=\qty{a\in [0, 63]}$|
  movb %rsi, (%rdi, %rax)
  addq $1, %rax
  cmpq $64, %rax
  jne .L0|$\opMode{\sigma}$| # |$\sigma(a)=a+1$|
\end{minted}
    \caption{Infer the valid region of an array}
    \label{fig:infer-range-loop}
\end{figure}

\subsection{Taint-Type Inference}
In this section, we demonstrate how to unify local taint variables at each block with each function's input taint variables and generate necessary predicates to satisfy all constraints.

According to Section~\ref{sec:tal-typing-rules} and Section~\ref{sec:infer-constraint-gen}, \systal constrains taint types via the following three aspects:
\begin{enumerate}
[leftmargin=*]
    \item Load/store addresses and branch conditions are untainted. Denote each of their taint types as $\tau=x_1\lor x_2\lor \dots x_n$. We can rewrite the constraint as $x_1 \Rightarrow 0\land x_2 \Rightarrow 0\land \dots x_n \Rightarrow 0$.
    \item The taint type of the accessed memory slot is equal to the taint annotation of a load/store operand (under some scenarios).
    According to type-constraint generation, both the memory slot's taint type and the load/store operand's taint annotation, denoted as $x_\text{slot}$ and $x_\text{op}$, are only represented by taint variables or constant taint values (instead of complex taint expressions).
    Therefore, the taint constraint can be written as $x_\text{slot}\Rightarrow x_\text{op}\land x_\text{op}\Rightarrow x_\text{slot}$,
    \item If store data is tainted, then the store operand's taint annotation is also tainted.
    Denote the store data's taint type as $x_1\lor x_2\lor \dots \lor x_n$ and the store operand's taint type as $x_\text{op}$.
    We can write the constraint as $x_1\Rightarrow x_\text{op} \land x_2\Rightarrow x_\text{op} \land \dots \land x_n\Rightarrow x_\text{op}$.
\end{enumerate}
In short, the taint constraints can be summarized in the form $E_1\land E_2\land \dots \land E_n$.
Here, each $E_k$ has the form $x_1\Rightarrow x_2$ where $x_1$ and $x_2$ are either taint variables or constant taint values.
This formula clearly constrains the taint flow among all taint variables.

Here is how we derive taint predicates.
For each local taint variable, we can identify its taint source represented by input taint variables. If there is no taint source, we set it to $0$. Otherwise, we set it to the logical OR of all its taint sources.
We can also collect predicates for input taint variables in a similar form $x_1\Rightarrow x_2$ and add them to the state type of the function's input block.

\section{Transformation}
\label{sec:trans}
We define a transformation that takes a well-typed \systal program as input and generates another program that satisfies our software contract, public noninterference (defined in Section~\ref{sec:sw-hw-contract}).
As discussed in Section~\ref{sec:overview-relocate-stack}, the overall strategy of our transformation is to maintain a secret stack that is shifted from the original stack by $\delta$ bytes. If a stack slot contains secrets, we shift its location by $\delta$ to move it onto the secret stack. If a stack slot contains public data, we do not change its location.
Note that our transformation does not affect heap/global variables, while we do use \systal's type system to ensure that they are used properly from an information-flow perspective.
We first present two basic transformation strategies to relocate memory data. Then, we illustrate how we apply these two strategies to transform the program.
We also formally prove that the transformation maintains the original program's functionality while guaranteeing public noninterference.

\subsection{Two Memory-Relocation Strategies}
\label{sec:trans-strategies}
Load/store instructions in \sysisa (and other ISAs) support the following addressing mode: taking a precalculated base pointer and adding an offset to the pointer to derive the target address.
There are thus two basic applicable strategies to transform memory accesses in assembly programs:
\begin{description}
    \item[TransPtr] We can modify the base pointer before it is used in the memory operand so that the memory operand automatically switches to accessing the relocated object.
    \item[TransOp] When we want to shift the target address by a constant offset, we can directly modify the memory operand to add this offset.
\end{description}

However, each strategy has limited applicability.
First, \transOp is a context-insensitive change, which uniformly shifts the memory-access address by the same offset.
As a result, \transOp is only suitable for the case where we statically know how the data object accessed by the instruction should be relocated (e.g., whether to move it to the secret stack).
However, the program may reuse the same instruction to operate on public and secret data in different situations.
For example, the \texttt{memset} function might be called to set either public or secret data objects, where we want to relocate them with different offsets.
Note that on the caller side, we may know more context information such as whether the data object is secret or not.
Hence, we choose \transPtr rather than \transOp to transform those store instructions in \texttt{memset} by modifying the pointer argument passed to \texttt{memset}, so that all the store instructions can automatically access the designated region.
\shixin{I think ``cannot'' is too absolute here, so try to avoid that.
}

On the other hand, when applying \transPtr to shift a base pointer, all load/store operands using the same base pointer will shift their target addresses by the same offset.
In other words, all memory slots referenced by the same base pointer will be relocated together by \transPtr, so it is only suitable for the case where those referenced slots share the same taint type.
For example, a function may access a struct that contains both secret and public fields (slots) through the same base pointer of the struct.
In this case, we use \transOp to avoid relocating the public slots.

One important case worth discussing is about translating memory accesses to slots referenced by the stack pointer.
On the one hand, the stack pointer, stored in \reg{rsp}, is used to reference different memory slots on the function's local stack, including both tainted and untainted ones. So, we should not apply \transPtr to transform the stack pointer.
On the other hand, the program may pass base pointers of stack objects to functions such as \texttt{memset}.
These pointers are stored in registers such as \reg{rdi} according to \sysisa's calling convention, and they are only used to access the corresponding data objects instead of arbitrary slots on the stack.
Hence, although the pointer points to the stack, we can still apply \transPtr as long as all slots within the corresponding object share the same taint.
\shixin{I am mixing plural and singular terms and have no idea about how to uniform them for these two sentence :(((}

\todo{Benefit: TransOp does not introduce extra instruction}
\todo{Add memset example (maybe for camera ready) (two examples for two cases are better than one giant example)}

\subsection{Transformation Details}
\label{sec:tf-details}
\pgheading{Determine transformation strategy}
The first step of our transformation is to decide which transformation strategy to use for each memory access.
We first determine the strategy for each memory slot and transform all memory accesses to that slot with the slot's strategy.
Given a function with input memory type $\MType$, we generate a map $\omega:\dom{\MType} \rightarrow \qty{\transOp, \transPtr}$ that maps each memory slot to its transformation strategy.
\shixin{Consider to change this to ``Given a memory type $\MType$'' since in Appendix $\omega$ depends on each instruction's memory type}

Following the discussion of the pros and cons of \transPtr and \transOp in Section~\ref{sec:trans-strategies}, we propose the the following approach to decide which strategy to use: $\omega(s)=\transPtr$ if and only if (1) $s$ is referenced by a pointer passed through a function argument, and (2) all slots in $\MType$ referenced by the same pointer have the same taint type.%
\footnote{We also require that for slot $s$ where its base pointer is a function argument and $\omega(s)=\transOp$, its taint type is $0$ or $1$.}
We formalize the generation of $\omega$ in \bothver{Appendix~\ref{sec:base-trans-pass}}{\cite{secsep-proof}}.

\pgheading{Transform load/store operands}
\label{tfop:memop}
Next, our transformation uses a pass $\CompilerOp$ to transform all load/store operands that access memory slots with transformation strategy \transOp.
For each memory operand, denoted as $\memOp{r_b}{r_i}{i_s}{i_d}\opMode{s, \tau}$, $s$ is the memory slot accessed by the operand, and $\tau$ is the slot taint type.
If $\omega(s)=\transOp$ and $\tau \neq 0$ (i.e., the slot might be tainted),
$\CompilerOp$ will rewrite the operand to $\memOp{r_b}{r_i}{i_s}{\delta+i_d}$ so that its target address is shifted by $\delta$ and relocated to the secret region.
\shixin{I am not sure whether it's good to use ``secret stack'' or simply ``secret region'' here.}
\shixin{I omit the annotation change since it will be detailed in the appendix.}

For instructions that perform load/store without explicit load/store operands in their ISA representations (e.g., $\instr{pushq}$, $\instr{popq}$), $\CompilerOp$ also synthesizes the transformed behavior accordingly.
Specifically, for simplicity, we will use $\instr{pushsecq}$ and $\instr{popsecq}$ to represent push/pop on the secret stack, which will be synthesized to valid \sysisa instructions in the final transformed program.

\pgheading{Transform pointer arguments}
\label{tfop:ptr-arg}
We define another pass $\CompilerPtr$ to perform \transPtr, which transforms pointer arguments passed to each callee function accordingly so that they use the transformed pointer to access designated regions.
Specifically, for each base pointer argument that references tainted slots with transformation strategy \transPtr, the callee function expects that the pointer is already shifted by $\delta$ and points to the address after relocation.

From the caller function's perspective, if the transformation strategy of the slots referenced by the same pointer is also \transPtr, then the pointer must be already shifted, and we do not need to make any changes.
On the other hand, if the transformation strategy of those slots is \transOp, then the pointer is not transformed yet, so we need to add $\delta$ to the pointer argument when passing it to the callee.
We formalize this transformation in \bothver{Appendix~\ref{sec:base-trans-pass}}{\cite{secsep-proof}}.

\pgheading{Restore callee-saved registers' taint}
\label{tfop:csr}
With $\CompilerOp$ and $\CompilerPtr$, our transformation can ensure that all memory operands accessing secret data on the stack are redirected to accessing the secret stack.
However, recall that \transOp shifts a memory operand's target address to the secret stack as long as the corresponding memory slot is not constantly untainted.
In other words, our transformation may conservatively redirect memory accesses to the secret stack even though they may operate on the public data under some circumstances, causing performance loss.

Specifically, each function usually saves callee-saved registers to its stack if needed, restoring them before returning to the call site.
Since the callee-saved registers can be tainted or untainted depending on the call site, we transform the function to always push them to the secret stack to avoid potential leakage.
For example, in Figure~\ref{fig:trans-callee}, \texttt{fparent\_sec} calls \texttt{fchild} with one callee-saved register \reg{r12} containing secret data, while \texttt{fparent\_pub} calls \texttt{fchild} with \reg{r12} containing a public pointer.
We then transform \texttt{fchild} to save \reg{r12} to the secret stack.
As a result, when returning to \texttt{fparent\_pub} after calling \texttt{fchild}, \reg{r12} is marked as tainted by a processor with coarse-grained memory taint tracking and secure speculation.
Since \reg{r12} is used as the store address on line 7, and the processor delays speculative memory accesses with tainted addresses to avoid leaking secrets, this store will be delayed until the commit stage, hurting performance.

We introduce another pass $\CompilerCallee$ that saves public callee-saved registers to the public stack before each call and retrieves them afterward.
$\CompilerCallee$ guarantees that when running the transformed program on a machine that does coarse-grained taint tracking, the callee-saved registers are never overtainted after function calls.

Extra stack space need not be allocated. As shown in Figure~\ref{fig:trans-callee-parent-pub}, \texttt{fparent\_pub} pushes \reg{r12} to the secret stack slot $s+\delta$ before using it, while the corresponding slot $s$ on the public stack is unused.
\todo{TODO for camera-ready: maybe add a diagram for this}
\shixin{No space for it :(}
Thus, we can use $s$ to save and restore the public value in \reg{r12} before and after calling \texttt{fchild} (highlighted in Figure~\ref{fig:trans-callee-parent-pub}).
\revision{
We formalize $\CompilerCallee$ in \bothver{Appendix~\ref{sec:opt-trans-pass}}{\cite{secsep-proof}}.
}

Note that although $\CompilerCallee$ helps avoid unnecessary delays on speculation, it inserts extra instructions into the program and may cause performance overhead.
We will evaluate this tradeoff by measuring the performance with and without $\CompilerCallee$ in Section~\ref{sec:performance-overhead}.

\begin{figure}
\centering
\begin{subfigure}{0.27\linewidth}
\centering
\begin{minted}{asm}
fchild:
  pushsecq %r12
  ...
  popsecq %r12
# r12 is
# from tainted 
# stack
| |
\end{minted}
    \caption{\reg{r12} $\rightarrow$ tainted}
    \label{fig:trans-callee-child}
\end{subfigure}%
\begin{subfigure}{0.35\linewidth}
\centering
\begin{minted}{asm}
fparent_sec:
  pushsecq %r12|$\opMode{s+\delta, \tau}$|
  movq $sec, %r12

  callq child


  popsecq %r12|$\opMode{s+\delta, \tau}$|
\end{minted}
    \caption{Secret \reg{r12}}
    \label{fig:trans-callee-parent-sec}
\end{subfigure}
\begin{subfigure}{0.37\linewidth}
\centering
\begin{minted}{asm}
fparent_pub:
  pushsecq %r12|$\opMode{s+\delta, \tau}$|
  movq $ptr, %r12
  movq %r12, (%rsp)|$\opMode{s, 0}$|
  callq child
  movq (%rsp)|$\opMode{s, 0}$|, %r12
  movq %rax, (%r12)
  popsecq %r12|$\opMode{s+\delta, \tau}$|
\end{minted}
    \caption{Public \reg{r12}}
    \label{fig:trans-callee-parent-pub}
\end{subfigure}
\caption{Restore callee-saved registers' taint}
\label{fig:trans-callee}
\end{figure}

\subsection{Transformation Soundness}
\pgheading{Functional Correctness}
First, we define a simulation relation $\hasVal{P, \PType}{S'\prec S}$, which correlates abstract machine states running the transformed program and the original program $P$.
Intuitively, the simulation relation maps the relocated memory data objects in the transformed program's state to those in the original program's.
It also requires that paired objects and registers in the two states have matching values as long as they are not pointers that might be changed by \transPtr.
\shixin{I eventually decided also to mention registers here. Please check!}

Denoting our overall transformation as $\Compiler$, we can formalize the functional correctness of our transformation using the following theorem.
\begin{theorem}[Functional Correctness]
If $\typeChecked{P, \PType}{S}$ and $\hasVal{P, \PType}{S'\prec S}$, then there exists $S_1$ and $S'_1$ such that $S\xrightarrow{\instVar} S_1$, $S'\xrightarrow{\Compiler(\instVar)}^* S_1'$, and $\hasVal{P, \PType}{S_1'\prec S_1}$; or $S$ and $S'$ are termination states.
\end{theorem}
Detailed formalization of the simulation relation and proof for the theorem can be found in \bothver{Appendix~\ref{sec:sim-rel}-\ref{sec:func-correctness}}{\cite{secsep-proof}}, where we focus on the proof for major passes $\CompilerOp$ and $\CompilerPtr$ and omit details for the optional pass $\CompilerCallee$ whose functional correctness is relatively more straightforward.

\pgheading{Public Noninterference}
Next, we briefly justify that the transformed program satisfies software public noninterference. The detailed proof can be found in \bothver{Appendix~\ref{sec:public-noninterference}}{\cite{secsep-proof}}.
In our transformation, we pick the address shift $\delta$ so that $\abs{\delta}$ is larger than the input program's maximum stack size. Then, we can denote the original stack region as $s_\text{pub}=[\spVar_\text{init}+\delta, \spVar_\text{init})$ and the new secret stack region as $s_\text{sec}=[\spVar_\text{init}+2\delta, \spVar_\text{init}+\delta)$, where $\spVar_\text{init}$ is the stack base.
Intuitively, $\Compiler$ will apply either \transOp or \transPtr to ensure that every instruction that accesses the original stack $s_\text{pub}$ and operates on tainted data will have its target address shifted by $\delta$ and access the secret stack $s_\text{sec}$ instead.
Hence, $\Compiler$ successfully ensures that the transformed program never stores secrets to the public stack region, thereby satisfying software public noninterference.
\shixin{I omit the discussion on non-stack region here. Is it fine?}

\section{Evaluation}

\subsection{Implementation and Experiment Setup}
\label{sec:eval-setup}
\pgheading{\sys{} Toolchain}
We implement a prototype toolchain in OCaml, using Z3~\cite{z3} as the SMT solver.
It includes
\revision{
(1) a parser for \sys{}'s C source code annotations,
(2) a parser for compiled \sysisa{} assembly programs,
(3) type-inference rules and algorithms (\autoref{sec:infer}),
(4) a checker that validates inferred types against typing rules (\autoref{sec:tal-typing-rules}), and
(5) transformation based on inferred types (\autoref{sec:trans}).
}
This prototype only covers the instructions in the benchmarks used for evaluation.
The toolchain incorporates LLVM/Clang to compile both the original and the transformed benchmark.

\pgheading{Hardware Defense}
We implement our hardware-defense part in gem5 simulator v22.1~\cite{Binkert:2011:gem5,Lowe-Power:2020:gem5-20}
replicating the defense idea from ProSpeCT~\cite{daniel2023prospect}.
Specifically, modules including ROB, register file, scheduler (\texttt{InstructionQueue}),
load/store queue, and branch squash logic (\texttt{Commit}, \texttt{IEW})
are modified to support taint tracking and to delay transmitter instructions that leak secrets.
We apply a microarchitecture configuration similar to that used in prior work on secure speculation~\cite{choudhary2021speculative}.
We model an 8-issue out-of-order superscalar processor with 32 load-queue entries, 32 store-queue entries, and 192 ROB entries.
We use a tournament branch-prediction policy with 4096 BTB entries and 16 RAS entries.
The memory system models a $\SI{32}{KB}$ $4$-way L1 I-cache, a $\SI{64}{KB}$ $8$-way L1 D-cache, and a $\SI{2}{MB}$ $16$-way L2 cache, with $\SI{64}{B}$ cache lines.

\pgheading{Experiment Setup}
We evaluate \sys{} on six cryptographic benchmarks:
our own implementation of \texttt{salsa20}
and five other benchmarks from BoringSSL~\cite{boringssl} (\texttt{sha512}, \texttt{chacha20}, \texttt{poly1305}, \texttt{x25519} and \texttt{ed25519\_sign}).
\texttt{ed25519\_verify} is excluded due to the current lack of declassification support, which can be implemented with minor extensions (see Section~\ref{sec:limitation-future}).
To minimize the instability due to cold caches,
each benchmark is modified to repeat its main routine 100 times.
In addition, we apply slight changes to some of the benchmarks, the details of which can be found in \bothver{Appendix~\ref{sec:eval:benchmark-changes}}{\cite{secsep-proof}}.

We conduct our experiments on a test platform equipped with
an Intel® Core™ i9-14900K CPU.
Benchmarks are transformed using $\delta=\SI{-8}{MB}$ and simulated in gem5 under syscall-emulation mode.


\subsection{Performance of Transformed Programs}
\label{sec:performance-overhead}
We evaluate and compare the performance overhead of the following four transformation schemes:
\shixin{@roger: ProSpeCT annotates stack variables instead of function variables, right? (I saw function variables in the old version.)}
\begin{enumerate}
[leftmargin=*]
    \item ProSpeCT (public stack): Manually annotate and relocate secret \revision{stack} variables,
    while treating the original stack as public.
    \revision{Note that it cannot relocate secret stack spills and thus is insecure.}
    \item ProSpeCT (secret stack): Manually \revision{relocate} public \revision{stack} variables
    while treating the original stack as secret.
    This approach conservatively protects any register spills and thus is secure.
    \item \sys{} (no $\CompilerCallee$): Perform transformation passes $\CompilerOp$ and $\CompilerPtr$.
    \item \sys{}: Perform all transformation passes, $\CompilerOp$, $\CompilerPtr$, and $\CompilerCallee$.
\end{enumerate}

\begin{figure*}
    \centering
    \includegraphics[width=\textwidth]{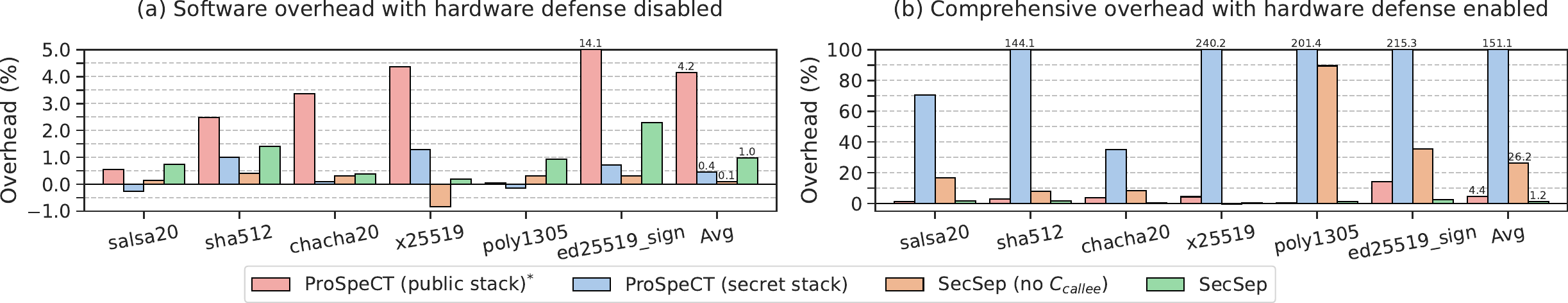}
    \captionsetup{width=\textwidth}
    \caption{Execution-time overhead of transformed programs relative to original programs.
    * means the scheme is not secure.
    }
    \label{fig:overhead}
\end{figure*}


\pgheading{Software Overhead}
We compare the execution time of transformed programs running on unmodified hardware, scaled to the execution time of the original program before transformation, shown in Figure~\ref{fig:overhead}(a).
The goal is to understand the software overhead introduced by the additional or transformed instructions and their microarchitectural impacts.

On average, all the schemes have relatively low overhead below 4.2\%, with ProSpeCT (public stack) having the highest overhead and \sys{} (no $\CompilerCallee$) having a close-to-zero overhead. 
The performance difference mainly comes from the number of extra instructions introduced during transformation
and whether the transformed program accesses regions far from the stack, which can result in worse cache performance.

\pgheading{Comprehensive Overhead}
We now examine the execution time of the transformed programs on hardware equipped with security defenses, shown in Figure~\ref{fig:overhead}(b).
\revision{The goal is to understand how precisely each transformation separates secret/public data and the combined overhead introduced by software and hardware.}

On average, \sys{} achieves the lowest overhead of $1.2\%$ among all schemes, while ProSpeCT (secret stack) can incur as high as $151.1\%$ overhead.
Compared to other secure schemes, \sys{} benefits from a more precise secret/public data separation,
thereby minimizing overtainting and enabling efficient execution when the hardware defense is enabled.
ProSpeCT (public stack) can also achieve relatively low overhead.
However, its performance gains stem from undertainting, which compromises security.
Nevertheless, it remains slower than \sys{}, likely due to the notable software overhead illustrated in \ref{fig:overhead}(a).

\pgheading{Effect of $\CompilerCallee$}
To assess the role of $\CompilerCallee$, we compare \sys{} (no $\CompilerCallee$) with \sys{}. 
When $\CompilerCallee$ is absent, the software overhead is lowered by $0.9\%$.
However, a significant overhead increase of $25\%$ occurs when the hardware defense is enabled,
highlighting the severity of overtainting caused by called routines.
Therefore, despite adding extra instructions, $\CompilerCallee$ is essential for efficient secure speculation.
Hence, we include it as a key component in \sys{}'s standard transformation scheme.

\subsection{Manual Effort}
\label{sec:eval:manual}
\newcommand{\allfuncs}{\textbf{F}}
\newcommand{\funcs}{\textbf{F'}}
\newcommand{\numvars}{\textbf{\#Var}}
\newcommand{\numargs}{\textbf{\#Arg}}

We also compare the manual effort required by ProSpeCT and \sys{} to annotate the source code for transformation as shown in \autoref{table:benchmark-metrics}.
We focus on ProSpeCT's public-stack scheme, as it exhibits efficient execution and \revision{can be secure for certain benchmarks~\cite{daniel2023prospect}.}
\shixin{I do not remember the patch stuff is for this purpose, removed.}
ProSpeCT requires examining stack variables in all functions (denoted as \allfuncs{}),
while \sys{} only requires examining arguments of functions in the binary's call graph (denoted as \funcs{}).
For fair comparison, we also count the stack variables and ProSpeCT annotations when examining only \funcs{},
with these numbers enclosed in parentheses.
\shixin{@roger: check whether my change from ``local'' variables to ``stack'' variables make sense to you.}

\begin{table}[!htbp]
\centering
\small
\resizebox{\linewidth}{!}{
\begin{tabular}{|lccc|cc|cc|}
\hline
\multicolumn{4}{|c|}{\textbf{Benchmark}} & \multicolumn{2}{c|}{\textbf{ProSpeCT (pub stack)}} & \multicolumn{2}{c|}{\textbf{\sys{}}} \\ \hline
\textbf{Name} & \textbf{LOC} & \textbf{\#F} & \textbf{\#F'} & \textbf{\#Var} & \textbf{\#Anno} & \textbf{\#Arg} & \textbf{\#Anno} \\
\hline


\texttt{salsa20        } &   72 &  5 &  3 &   9 (5) &   4 (2) &   6 &  5 \\
\texttt{sha512         } &  290 & 16 &  3 &  24 (2) &  14 (1) &   6 &  5 \\
\texttt{chacha20       } &  100 &  7 &  3 &   8 (7) &   4 (3) &   8 &  7 \\
\texttt{x25519         } & 1034 & 41 &  5 & 361 (11) & 335 (4) &  10 &  7 \\
\texttt{poly1305       } &  314 & 11 &  7 &  33 (32) &  26 (25) &  11 & 74 \\
\texttt{ed25519\_sign  } & 2314 & 72 & 11 & 512 (77) & 476 (69) &  25 & 24 \\

\hline
\end{tabular}
}

\captionsetup{width=\columnwidth}
\caption{\revision{Comparison of annotation efforts between ProSpeCT and \sys{}. \#F denotes \#functions present in the benchmark, while \#F' denotes \#functions called during the execution of the benchmark. 
\#Var denotes \#stack variables that need to be examined for correct annotation in ProSpeCT, and 
\#Arg denotes \#function arguments that need to be examined for correct annotation in \sys{}.
\#Anno reports \#lines of annotation in the corresponding technique.}}
\label{table:benchmark-metrics}
\end{table}

Note the number of \sys{} annotations grows linearly with the number of function arguments,
\revision{because \sys{} only requires identifying}
the memory layout and taint types of all function arguments for functions in \funcs{}.
An annotation burst is incurred at \texttt{poly1305} due to the frequent use of a 17-field structure as function arguments, while 14 of them share identical attributes and thus identical annotations.
Writing \sys{} annotations takes low effort
since cryptographic functions usually have clear interfaces and seldom use complex data structures with unpredictable sizes (e.g., linked lists).

We also observe that
the number of variables to examine using ProSpeCT (\numvars{})
is generally higher than the number of function arguments to examine using \sys{} (\numargs{}),
even when the scope is restricted to \funcs{} when counting \numvars{}.
This pattern indicates that \sys{} places less burden on the user by requesting only function-interface-level annotations to harden cryptographic programs.

\ifextendedversion
\subsection{Benchmark Changes}
\label{sec:eval:benchmark-changes}

In \texttt{salsa20}, due to the limited heuristic rules in our valid-region-inference prototype, we explicitly initialize one stack-allocated array to zero.
In \texttt{chacha20}, we add \sys{} annotations derived from program semantics to facilitate inference by
(1) identifying the trivial outcome of a complex initialization procedure and
(2) encoding a memory nonoverlap assumption presented in the source code.
We also modify \texttt{poly1305} to (1) eliminate the unsupported \texttt{goto}-based irreducible control flow and (2) avoid casting a structure to and from a byte array, which obscures the type
and hinders type inference.

One can expect that as more heuristic rules are added and the type system becomes fully implemented, these manual efforts will also be eliminated.

\subsection{Toolchain Efficiency}
\label{sec:eval:efficiency}
\revision{
We profile our \sys{} prototype by running the toolchain on all \sys{}-annotated benchmarks using the test platform.
The corresponding statistics are presented in Figure~\ref{fig:efficiency}.

The total time required for type inference and checking ranges from a few seconds to around 30 minutes.
The program transformation phase -- omitted from the figure --
relies exclusively on the inferred types and completes trivially fast,
costing under one second for all benchmarks.
While the toolchain may exhibit slower performance on large benchmarks,
it yields a valuable and reliable model of binary-level information flow,
enabling efficient execution and secure speculation even for complex cryptographic functions.

Both inference and checking spend over 90\% of their time making numerous SMT solver calls,
indicating a heavy reliance on the solver to guide heuristics and enforce type checking.
\sys{}'s performance can be further improved through optimizations such as tailoring SMT-solver invocations.
}

\begin{figure}
    \centering
    \includegraphics[width=0.9\linewidth]{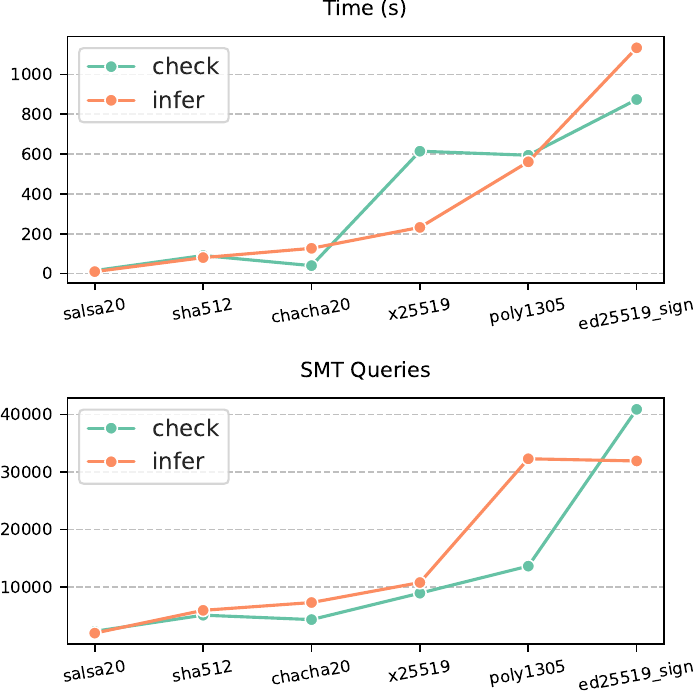}
    \caption{\revision{Efficiency of \sys{} toolchain.
    }}
    \label{fig:efficiency}
\end{figure}

\fi

\section{Limitations and Future Work}
\label{sec:limitation-future}
While \sys offers appealing features to rewrite assembly programs and separate secret/public data automatically, we acknowledge that it has several limitations worth discussing.
First, we use heuristic-based type inference to transform assembly programs compiled by the off-the-shelf compiler LLVM, where the heuristics are developed through an empirical review of these assembly programs.
The limitation is that it is not guaranteed to handle all possible assembly code patterns, and we need new heuristics for new compiler optimizations. This limitation can be alleviated by more engineering effort to derive better heuristics.

Second, our inference algorithm relies on extra information (e.g., memory layout, valid regions, and taint) provided by source-code annotations to generate types.
On the one hand, it is relatively straightforward for the programmer to provide these annotations since they only need to emphasize the high-level meanings of function arguments.
On the other hand, we acknowledge that extra manual effort is required to go through each function argument, thereby increasing the barrier to using our tools (evaluated in Section~\ref{sec:eval:manual}).
Future work might be conducted to improve the inference algorithm to lift the need for these manual annotations.

Third, we provided a prototype to demonstrate the overall idea, while more features could be supported to improve the usability of our tool.
For example, declassification is an essential notion in cryptographic programs supported by prior works such as ProSpeCT~\cite{daniel2023prospect}
but not included in our type system.
To extend our prototype to support declassification, one can define a special function that takes a secret input and writes it to a given address in the public region (passed as a parameter of the function). The function is excluded from type inference and checking, so the program can call this function to write secrets to public regions for declassification.
\todo{Double check whether this is correct.}
\sys can also be improved to be compatible with more programs by supporting dynamically linked libraries and handling analysis with dynamically allocated heap data and pointer type casting.

\section{Related Work}
We first discuss prior works on typed assembly language to justify the novelty and contribution of \systal.
Next, we discuss prior mitigations, including software and hardware approaches, that aim to protect cryptographic programs against speculative-execution attacks.
We also discuss prior work that transforms programs to separate secret and public data via a compiler approach.

\pgheading{Typed assembly language}
Prior works~\cite{morrisett1999system,grossman2000scalable,crary1999talx86,glew1999type,morrisett1998stack} propose typed assembly language (TAL) and type-preserving compilation from high-level programs to TAL, where the types help guarantee the security of assembly programs.
Instead of compiling high-level programs to generate assembly programs, \sys rewrites assembly programs generated by an off-the-shelf compiler (LLVM) while using inferred types to guarantee security and functional correctness. Hence, the transformed programs still benefit from the optimizations of realistic compilers.

Jiang et al.~\cite{jiang2022cache} also propose a type system and corresponding type-inference algorithm for assembly programs.
Their type system introduces more accurate information-flow tracking at bit granularity and helps detect side-channel vulnerabilities in cryptographic libraries.
Note that the soundness of their type system relies on an assumption of memory safety, while our type system accurately tracks possible address ranges of each memory access and guarantees memory safety.
\shixin{I am very confused about their typing rules. They have a vague assumption on ``memory safety''.
My current understanding is that they assume there is no out of bound access and the load data's taint can be inferred from the based pointer used by the instruction.
}

Other prior works~\cite{barthe2021high,shivakumar2023typing,arranz2025protecting} design information-flow type systems to guarantee that well-typed programs satisfy speculative constant time and implement their approaches in the Jasmin framework~\cite{almeida2017jasmin}.
In this framework, the developers directly program in Jasmin, an assembly-like programming language, which requires more manual effort compared with programming in higher-level languages such as C and is not compatible with some off-the-shelf cryptographic libraries such as BoringSSL~\cite{boringssl}.

\pgheading{Software mitigations against Spectre}
Several prior works~\cite{mosier2024serberus,vassena2021automatically,narayan2021swivel,choudhary2023declassiflow,patrignani2021exorcising,zhang2023ultimate,barthe2021high,shivakumar2023typing,shivakumar2023spectre,arranz2025protecting} harden cryptographic programs against Spectre attacks by analyzing the programs' speculative control flow and blocking insecure speculation at the software level, e.g., by memory-fence insertion or speculative load hardening (SLH)~\cite{specloadhardening}.
Many of these approaches~\cite{mosier2024serberus,vassena2021automatically,narayan2021swivel,choudhary2023declassiflow,patrignani2021exorcising,zhang2023ultimate} introduce large performance overhead since they unavoidably block safe speculation conservatively when blocking insecure speculation.
Other works~\cite{barthe2021high,shivakumar2023typing,shivakumar2023spectre,arranz2025protecting} managed to achieve speculative noninterference with marginal overhead by applying SLH intelligently, but they are built upon research-prototype source languages such as FaCT~\cite{cauligi2019fact} or Jasmin~\cite{almeida2017jasmin,almeida2020last}, thereby not compatible with off-the-shelf cryptographic libraries such as BoringSSL~\cite{boringssl}.
Furthermore, many of them~\cite{vassena2021automatically,choudhary2023declassiflow,patrignani2021exorcising,zhang2023ultimate,barthe2021high,shivakumar2023typing,shivakumar2023spectre} only consider speculation at conditional branches in their speculative control-flow analysis, so they cannot prevent leakage introduced by other speculation primitives~\cite{Kocher2018spectre,koruyeh2018spectre,horn2018speculative,cauligi2022sok,guanciale2020inspectre,ponce2022cats}.
\todo{Another limitation shared by these software mitigations is that they must be applied to full programs. However, cryptographic libraries are co-located with non-cryptographic code, while making non-cryptographic code satisfy speculative constant time will likely introduce large performance overhead.}
\todo{TODO for camera ready: the above statement is too vague. We may discuss more about it.}
\shixin{Should I mention the motivation that sw-only approach can only work for whole programs here? We also only support whole program analysis in our prototype though...}

\pgheading{Hardware mitigations against Spectre}
Prior works~\cite{yu2019speculative,weisse2019nda,barber2019specshield,yu2020speculative,choudhary2021speculative,loughlin2021dolma} 
adopt hardware taint tracking and delay speculative operations that transmit secrets.
These pure hardware solutions require no software changes but face challenges in identifying secret data in memory.
STT~\cite{yu2019speculative} and others~\cite{weisse2019nda,barber2019specshield,yu2020speculative} only consider speculatively loaded data as secrets, while leaving nonspeculatively loaded data unprotected, thereby not guaranteeing constant time~\cite{guarnieri2021hardware}.
SPT~\cite{choudhary2021speculative} solves this problem by considering all data loaded from memory as tainted and only marking data as public if it is transmitted by the program nonspeculatively, but it introduces complex hardware changes to achieve good performance.

The authors of ProSpeCT~\cite{daniel2023prospect} and others~\cite{schwarz2020context,fustos2019spectreguard} propose to make the program separate secret and public data into coarse-grained regions so that the hardware can easily identify secrets.
This paper complements their solution by providing an automatic approach to generate a program satisfying this requirement.

\hyphenation{Conf-LLVM}

\pgheading{Compiler approach to separating secret and public data}
ConfLLVM~\cite{brahmakshatriya2019confllvm} also uses a transformation technique that separates secret/public stack data into different regions. 
However, its static analysis cannot guarantee a program never writes secret data to the public region, so it relies on inserted run-time checks to ensure the correctness of the separation, which introduces extra overhead.

\section{Conclusion}
This paper proposed \systal, a new variant typed assembly language that helps rewrite cryptographic programs so that they split their secret and public data across coarse-grained memory regions.
We provide a heuristic inference algorithm to infer the types of off-the-shelf cryptographic programs and automate the transformation process.
The transformed programs enable hardware with fine-grained taint tracking at the register level and coarse-grained taint tracking at the memory level to achieve secure speculation with low performance overhead.

\begin{acks}
The authors thank the Matcha Group (MIT) for their help and the anonymous CCS reviewers for their feedback.
This work was supported in part by a gift from Amazon; by the Air Force Office of Scientific Research (AFOSR) under grant FA9550-22-1-0511;  by ACE, one of the seven centers in JUMP 2.0, a Semiconductor Research Corporation (SRC) program sponsored by DARPA.
\todo{@mengjia @adam: please check whether any funding sources are missing or incorrect here.}
\end{acks}

\bibliographystyle{ACM-Reference-Format}
\bibliography{main}


\ifextendedversion
\appendix
\section{Type Safety}
\subsection{Supplementary Notes on \systal}
\label{sec:supp-tal}
\shixin{@adam @mengjia: please suggest name for this section}
\shixin{@adam: maybe check this section again :)}
In this section, we explain some details about \systal's typing rules omitted in the main sections, and provide extra assumptions on well-typed programs.
\todo{Consider also explaining isNonChangeExp in this section. For constraints omitted in the main text typing rules, maybe mark them as gray and explain here.}
\todo{In the paper, I said we only supported non-overlapped memory slots. However, in our implementation, we support it following Andres's paper's idea. Add that to this section too!!!}

\subsubsection{Input and Block Type Variables}
In \systal, we call the type variables of each function's input block state type ``input variables'' (since they serve as the function's arguments) and other variables only used in each non-input block's state types ``block variables.'' 
Input variables are shared among all state type contexts within the same function.
Hence, \textsc{Typing-Jne} requires that for branch instruction $\inst{jne}{\ell\opMode{\sigma}}$, its branch annotation satisfies $\text{getInputVar}(\dom{\sigma})=\emptyset$, so that when $\sigma$ is used to instantiate type variables in the target block's context, it only instantiates block variables and leave input variables unchanged.

Furthermore, \systal also expects that all taint variables are input variables and shared among all blocks in the same function, as reflected by $\text{getTaintVar}(\dom{\sigma})=\emptyset$ in \textsc{Typing-Jne}.
This requirement is necessary to prove the functional correctness of our transformation.

\revision{
\subsubsection{Explanation of $\text{isSpill}$}
\label{sec:def-is-spill}
\shixin{@adam, @mengjia: please check this!}
In Section~\ref{sec:tal-typing-rules}, we introduce a predicate $\text{isSpill}$ to check whether a given memory slot is a stack spill and invoke different typing rules (e.g., \textsc{Typing-StoreOp-Spill} or \textsc{Typing-StoreOp-Non-Spill}) accordingly.
Specifically, we rely on debug tables generated by Clang to retrieve this information.

However, there is one special case worth discussing.
Typing rules in \systal (\textsc{Typing-Callq} and \textsc{State-Subtype}) require that when calling a callee function, each of its memory slots (including its local stack slots) must belong to one of the caller's memory slots.
Since functions usually do not directly access their callee's stack \todo{find some more convincing words for this}, we expect each function to have a special slot (denoted as $\sCalleeStack$) that marks the memory region reserved for its callee's local stack so that it satisfies the above constraint.

Note that $\sCalleeStack$ is not a typical stack spill slot, but we still make $\text{isSpill}(\sCalleeStack)=\mathit{true}$ for the simplicity of typing constraints.
Recall that we impose uniformity on the taint type of other non-spill slots so that the unified taint can be used as a hint for transformation.
On the other hand, $\sCalleeStack$ is not used in the caller's scope, so we do not need to assign it a uniform taint for transformation purposes.
Imposing $\text{isSpill}(\sCalleeStack)=\mathit{true}$ is a trick to avoid this constraint while still making the type system sound and the transformation functionally correct.


\subsubsection{Definition of $\text{updateMem}$}
\label{sec:def-update-mem}
When calling a function, part of the caller's memory slots will be accessed and modified by the callee, including the memory slots passed to the callee through pointers, and the special slot $\sCalleeStack$ reserved for the callee's local stack.
\textsc{Typing-Callq} invokes a predicate $\text{updateMem}$ to generate the caller's memory type after the function call, which is defined in Algorithm~\ref{alg:update-mem}.
Note that $\text{updateMem}$ has several extra assertions which further constrain well-typed programs.
\todo{Do I need to explain each assertion line-by-line?}

First, if a memory slot $s$ is accessed by the callee, $s$ must either be $\sCalleeStack$, the slot reserved for the callee's local stack slot, or a non-spill slot.
This constraint is reasonable since stack spill slots are used to hold register values temporarily and do not refer to any specific data objects that might be passed to a callee function.
\todo{Maybe say that it is satisfied by all 6 benchmarks.}

Second, all local stack slots of the callee (including its spill slots) should belong to $\sCalleeStack$, and we expect them to be uninitialized (i.e., have empty valid regions) before and after the call.
\todo{Maybe just say ``invalid'' to unify the term I use to the next sec.}

Third, for ease of illustrating the proof idea, we assume that each memory slot (except for $\sCalleeStack$) is updated by at most one slot of the callee's function, while the general case can be defined and proved similarly.
Furthermore, the callee slot's train type should be the same as the caller slot's, so that the slot is relocated by the same offset in the caller and callee during transformation.



\begin{algorithm}
\begin{algorithmic}
\Require Caller's memory type before call $\MType_0$, callee's memory type at return (represented under caller's type context) $\MType_1$
\State $S\gets \emptyset$
\State $S[s\mapsto \qty{}]$ for each $s\in \dom{\MType_0}$
\For{$s_1\in \dom{\MType_1}$}
    \State Find $s\in\dom{\MType_0}$ such that $s_1\subseteq s$
    \State $S[s\mapsto S[s]\cup \qty{s_1}]$
\EndFor
\State $\MType\gets \MType_0$
\For{$s\in\dom{\MType}$}
    \State $(\sval_0, (e_0, \tau_0)) \gets \MType_0[s]$
    \If{$S[s]=\emptyset$}
        \State \textbf{continue}
    \ElsIf{$\text{isSpill}(s)$}
        \Comment{$s$ should be $\sCalleeStack$, the slot reserved for callee's local stack slots, so $s$ and all slots in $S[s]$ should be uninitialized.}
        \State \textbf{assert} $\sval_0=\emptyset$
        \For{$s_1\in S[s]$}
            \State \textbf{assert} $\MType_1[s_1]=(\emptyset, \_)$
        \EndFor
    \Else
        \State \textbf{assert} $S[s]=\qty{s_1}$ for some $s_1$
        \Comment{We only discuss the case where each slot is updated by one callee's slot}
        \State $(\sval_1, (e_1, \tau_1))\gets \MType_1[s_1]$
        \State \textbf{assert} $\neg \text{isSpill}(s_1) \land \tau_0=\tau_1$
        \Comment{We expect $s_1$ is not a spill slot from the callee's view.
        \todo{Double check this claim against our typing rules,}}
        \State $\sval \gets (\sval_0 \backslash s_1) \cup \sval_1$
        \If{$\sval=\sval_1$}
            \State $\MType[s \mapsto (\sval, (e_1, \tau_0))]$
        \Else
            \State \textbf{assert} $\text{isNonChangeExp}(e_0)\land \text{isNonChangeExp}(e_1)$
            \State $\MType[s \mapsto (\sval, (\top, \tau_0))]$
        \EndIf
    \EndIf
\EndFor
\end{algorithmic}
\caption{Update memory type after function call ($\text{updateMem}$)}
\label{alg:update-mem}
\end{algorithm}
}

\subsubsection{Assumptions on Well-typed Programs}
We provide extra assumptions on well-typed programs by constraining function types in Figure~\ref{fig:tal-prog-typing-revision}.
Specifically, rule \textsc{Typing-Func} states that for a well-typed function:
\begin{itemize}
    \item The main function directly halts, and other functions only return at the $f_\text{ret}$ block.
    \item The dependent type of $\reg{rsp}$ at the beginning of the function is always represented by a special type variable $\spVar$.
    \item The memory slot used to store the return address does not belong to the function's memory type, which implies that the function will never overwrite the return address during its lifetime.
    \item Memory slots on the function's local stack are marked as invalid (i.e., uninitialized) at the beginning and end of the function.
    \item Callee-saved registers (including $\reg{rsp}$) are restored to their initial values when entering the function before return.
    \item Values of caller-saved registers and memory slots that are valid in the return state should not depend on pointer values, so they will not be affected by our transformation on pointers.
    \item Type constraints in the type context at the beginning of the function must also be constrained by the type context at other basic blocks.
\end{itemize}

Furthermore, we assume that each function only jumps to other functions through function calls and returns.
\todo{Consider to move all other constraints here!!!}

\begin{figure}
    \centering
\begin{mathpar}
\inferrule[Typing-Prog]{
\forall f \in \dom{P}.\; \hasType{P, \PType}{P(f)}{\PType(f)}\\
}{
\hasType{}{P}{\PType}
}\and
\inferrule[Typing-Func]{
P(f)=F\\
\forall \ell\in \dom{F}. \;
\hasType{\PType, \Gamma}{F(\ell)}{\Gamma(\ell)}\\
\forall \ell \in \dom{F}, \ell \neq f_\text{ret}. \; \instr{halt}, \instr{retq} \not \in F(\ell)\\
F(f_\text{ret})= (f=\mathit{main}) \;?\; \instr{halt} : \instr{retq}\\
\Gamma(f)=(\Delta_0, \RType_0, \MType_0)\\
\Gamma(f_\text{ret})=(\Delta_1, \RType_1, \MType_1)\\
\RType_0[r_\texttt{rsp}]=\RType_1[r_\texttt{rsp}]=(\mathit{sp}, 0)\\
\forall s \in\dom{\MType_{0,1}}.\; \hasVal{\Delta_{0, 1}}{s\cap [\mathit{sp}, \mathit{sp}+8)=\emptyset}\\
\forall s, \text{isLocalStack}(s).\; \MType_0[s]=(\emptyset,\_) \land \MType_1[s]=(\emptyset, \_)\\
\forall r, \text{isCalleeSaved}(r).\; (\exists x\neq \top.\; \RType_0[r]=\RType_1[r]=(x, \_))\\
\forall s, \MType_1[s]=(\sval,\beta).\; \sval=\emptyset \lor \text{isNonChangeExp}(\beta)\\
\forall r\in \dom{\RType_1}, \neg \text{isCalleeSaved}(r).\; \text{isNonChangeExp}(\RType_1[r])\\
\revision{\forall \ell \in \dom{F}, (\Delta, \RType, \MType)=\Gamma(\ell).\; \Delta_0\subseteq \Delta}\\
}{
\hasType{P, \PType}{F}{\Gamma}
}
\end{mathpar}
    \caption{Program typing}
    \todo{I removed constraints on getFuncInterface. I decide to make it a simple function that directly retrieves $\Gamma(f)\rightarrow \Gamma(f_\text{ret})$. Need to implement this in ocaml}
    \todo{I need to modify the implementation of infer tool to remove the return address slot from $\MType$!}
    \todo{Check other changes in implementation!!!}
    \label{fig:tal-prog-typing-revision}
\end{figure}

\subsubsection{Memory Configuration}
For ease of illustration, we provide a fixed configuration for public/secret and stack/non-stack memory regions as follows:
\begin{itemize}
    \item Public stack: $\stackPub=[\spVar_\text{init}+\delta, \spVar_\text{init})$
    \item Secret stack: $\stackSec=[\spVar_\text{init} + 2\delta, \spVar_\text{init}+\delta)$
    \item Public non-stack: $\otherPub$
    \item Secret non-stack: $\otherSec$
\end{itemize}
With this configuration, we provide extra constraints on the memory layout of the well-typed program (to be transformed later):
First, the program should only put its stack in $\stackPub$ and keep $\stackSec$ untouched so that our transformation can relocate data objects to this region without overwriting other valuable data.
Second, as we focus on separating secret/public data in the stack region, we expect the original program to separate non-stack data properly, so we require that all data in $\otherPub$ be untainted.

\shixin{Self note: IMPROTANT!!! Whenever we try to compare with the four concrete memory regions, we have to use concrete addresses or memory slots (with all variables instantiated).}

\subsection{\systal's Abstract Machine}
\label{sec:tal-machine}
In Section~\ref{sec:tal-syntax}, we defined \systal's abstract machine state as $(R, M, \pcVar)$, which only covers the register file, memory, and the current PC.
For the convenience of describing the relation between states and state types, we extend the state to $(R, M, \Phi)$, where $\Phi$ records PC and type substitution to instantiate type variables for the current state and all call sites on the call stack using constants:
\begin{equation*}
\begin{array}{rcll}
    S & \Coloneqq & (R, M, \Phi) & \emph{State}\\
    R & \Coloneqq & \qty{r_1: (v_1, t_1), \dots, r_n: (v_n, t_n)} & \emph{Register File}\\
    M & \Coloneqq & \qty{\mathit{addr}_1: (v_1, t_1), \dots, \mathit{addr}_n: (v_n, t_n)} & \emph{Memory}\\
    \Phi & \Coloneqq & [] | (\pcVar, (\sigma_f, \sigma_b)); \Phi  & \emph{PC/type}\\
    & & & \emph{substitution stack}\\
    \sigma & \Coloneqq & \overrightarrow{x} \rightarrow \overrightarrow{e} & \emph{Type subsitution}\\
\end{array}
\end{equation*}
Specifically, $\sigma_f$ is used to instantiate type variables in state types of all instructions in the current function, and $\sigma_b$ is for type variables in state types of all instructions in the current block.

We then define the operational semantics for this extended abstract machine in Figure~\ref{fig:dyn-operatinoal-semantics-revision} and constrain well-formedness in Figure~\ref{fig:well-formed-states-revision}.
Specifically, we use the helper function $\text{getStateType}$ to denote the operation of applying typing rules in Figure~\ref{fig:tal-inst-typing} and deriving the state type at each instruction from the state type of its basic block (recorded in $\PType$).
We also define the helper function $\text{getDom}$ as
\begin{align*}
    \text{getDom}(\MType, \sigma) & = \bigcup_{s\in\dom{\MType}}\sigma(s).
\end{align*}

\begin{figure}
    \centering
\begin{mathpar}
\inferrule[Dyn-Movq-m-r]{
\text{getInst}(P, \pcVar) = \inst{movq}{i_d(r_b, r_i, i_s)\opMode{s, \tau}, r}\\
\pcVar' = \text{nextPc}(P, \pcVar)\\
R[r_b]=(v_b, 0)\\
R[r_i]=(v_i, 0)\\
e = i_d+v_b+v_i\times i_s\\
M[e, e+8]=(v_m, t_m)\\
t_m \Rightarrow (\sigma_f\cup \sigma_b)(\tau)\\
R'=R[r \mapsto (v_m, t_m)]
}{
\dynNext{P}{(R, M, (\pcVar, (\sigma_f, \sigma_b));\Phi)}{(R', M, (\pcVar', (\sigma_f, \sigma_b));\Phi)}{}
}\and
\inferrule[Dyn-Movq-r-m]{
\text{getInst}(P, \pcVar) = \inst{movq}{r, i_d(r_b, r_i, i_s)\opMode{s, \tau}}\\
\pcVar' = \text{nextPc}(P, \pcVar)\\
R[r_b]=(v_b, 0)\\
R[r_i]=(v_i, 0)\\
e = i_d+v_b+v_i\times i_s\\
R[r]=(v_r, t_r)\\
t_r \Rightarrow (\sigma_f\cup \sigma_b)(\tau)\\
M'=M[[e, e+8)\mapsto (v_r, t_r)]
}{
\dynNext{P}{(R, M, (\pcVar, (\sigma_f, \sigma_b));\Phi)}{(R, M', (\pcVar', (\sigma_f, \sigma_b));\Phi)}{}
}\and
\inferrule[Dyn-Jne]{
\text{getInst}(P, \pcVar) = \inst{jne}{\ell\opMode{\sigma}}\\
\text{getFunc}(P, \pcVar)=\text{getFunc}(P, \pcVar')\\
R[\reg{ZF}]=(b, 0)\\
\pcVar' = b\;?\; \text{nextPc}(P, \pcVar) : \text{getBlockPc}(P, \ell)\\
\revision{\sigma_b' = b\;?\; \sigma_b : (\sigma_f\cup \sigma_b)\circ \sigma}\\
}{
\dynNext{P}{(R, M, (\pcVar, (\sigma_f, \sigma_b));\Phi)}{(R, M, (\pcVar', (\sigma_f, \sigma_b'));\Phi)}{}
}\and
\inferrule[Dyn-Callq]{
\text{getInst}(P, \pcVar)=\inst{callq}{f\opMode{\sigmaCall,\sigmaRet}}\\
\text{getBlockPc}(P, f)=\pcVar'\\
R[r_\texttt{rsp}]=(v, 0)\\
R'=R[r_\texttt{rsp}\mapsto(v-8, 0)]\\
M'=M[[v-8, v)\mapsto (\text{nextPc}(P, \pcVar), 0)]\\
\Phi'=(\pcVar', ((\sigma_f\cup\sigma_b)\circ\sigmaCall , []\rightarrow[])); (\pcVar, (\sigma_f, \sigma_b));\Phi
}{
\dynNext{P}{(R, M, (\pcVar, (\sigma_f, \sigma_b));\Phi)}{(R', M', \Phi')}{}
}\and
\inferrule[Dyn-Retq]{
\text{getInst}(P, \pcVar)=\inst{retq}{}\\
\text{getInst}(P, \pcVar_p) = \inst{callq}{f\opMode{\sigmaCall, \sigmaRet}}\\
R[r_\texttt{rsp}]=(v, 0)\\
R'=R[r_\texttt{rsp}\mapsto(v+8, 0)]\\
\pcVar' = \text{nextPc}(P, \pcVar_p)\\
M[v, v+8]=(\pcVar', 0)\\
\Phi'=(\pcVar',(\sigmaFp, (\sigma_{b}\circ \sigmaRet^{-1})\cup\sigmaBp));\Phi
}{
\dynNext{P}{(R, M, (\pcVar, (\sigma_{f}, \sigma_{b})); (\pcVar_p, (\sigmaFp, \sigmaBp));\Phi)}{(R', M, \Phi')}{}
}
\end{mathpar}
    \caption{Dynamic operational semantics for ISA machine with taint tracking and type substitution}
    \label{fig:dyn-operatinoal-semantics-revision}
\end{figure}

\begin{figure}
    \centering
\begin{mathpar}
\inferrule[Value-Type]{
v=e\lor e=\top\\\\
t\Rightarrow\tau\\
}{
\hasType{}{(v, t)}{(e, \tau)}
}\and
\inferrule[Reg-Type]{
\forall r\in\dom{\RType}.\;\hasType{}{R[r]}{\RType[r]}
}{
\hasType{}{R}{\RType}
}\and
\inferrule[Mem-Slot-Range-Type]{
\hasType{}{M[v_1, v_2]}{\beta}\\
}{
\hasType{}{M}{([v_1, v_2), \beta)}
}\and
\inferrule[Mem-Slot-Empty-Type]{
}{
\hasType{}{M}{(\emptyset, \beta)}
}\and
\inferrule[Mem-Slot-Top-Type]{
\forall x\in \sval. \hasType{}{M[x]}{(\top, \tau)}
}{
\hasType{}{M}{(\sval, (\top, \tau))}
}\and
\inferrule[Mem-Type]{
\forall s, \MType[s]=(\sval, (\_, \tau)).\;\\
((s\subseteq \stackPub \lor s \subseteq \otherPub \lor s\subseteq \otherSec)\\
\land (\text{isSpill}(s) \Rightarrow s\subseteq \stackPub)
\land (s\subseteq \otherPub \Rightarrow \tau = 0)\\
\land (\sval \subseteq s)\land (\hasType{}{M}{\MType[s]}))
}{
\hasType{}{M}{\MType}
}\and
\inferrule[Reg-Mem-Type]{
\revision{\text{getInputVar}(\dom{\sigma_f})=\dom{\sigma_f}}\\
\revision{\text{getInputVar}(\dom{\sigma_b})=\emptyset}\\
\revision{\sigma=\sigma_f\cup \sigma_b}\\
\sigma(\Delta)=\qty{\mathit{true}}\\
\hasType{}{R}{\sigma(\RType)}\\
\hasType{}{M}{\sigma(\MType)}\\
\revision{\forall s\in \dom{\MType}.\; \text{getVars}(s)= \sigma_f(s)}\\
}{
\revision{}{\hasType{}{(R, M, (\sigma_f, \sigma_b))}{(\Delta, \RType, \MType)}}
}\and
\inferrule[Well-Formed-State-Base]{
\hasType{}{P}{\PType}\\
\text{getFunc}(P, \pcVar)=\mathit{main}\\
\revision{\hasType{}{(R, M, (\sigma_f, \sigma_b))}{\text{getStateType}(\PType, \pcVar)}}\\
}{
\typeChecked{P, \PType}{(R, M, (\pcVar, (\sigma_f, \sigma_b)))}
}\and
\inferrule[Well-Formed-State-Inductive]{
\hasType{}{P}{\PType}\\
\text{getFunc}(P, \pcVar) = f \neq \mathit{main}\\
(\Delta, \RType, \MType)=\text{getStateType}(\PType, \pcVar)\\
\revision{\hasType{}{(R, M, (\sigma_f, \sigma_b))}{(\Delta, \RType, \MType)}}\\
\sRet=[\sigma_f(\spVar), \sigma_f(\spVar)+8)\\
M[\sRet]=(\text{nextPc}(P, \pcVar_p), 0)\\
\revision{\forall x\not\in \text{getDom}(\MType, \sigma_f)\cup \sRet.\; M[x]=M_p[x]}\\
\text{getInst}(P, \pcVar_p)= \inst{callq}{f\opMode{\sigmaCall,\sigmaRet}}\\
\sigma_f=(\sigmaFp\cup \sigmaBp)\circ \sigmaCall\\
\typeChecked{P, \PType}{(R_p, M_p, (\pcVar_p, (\sigmaFp, \sigmaBp)); \Phi)}
}{
\typeChecked{P, \PType}{(R, M, (\pcVar, (\sigma_f, \sigma_b)); (\pcVar_p, (\sigmaFp, \sigmaBp)); \Phi)}
}
\end{mathpar}
    \caption{Well-formed machine states, where $\spVar$ represents \reg{rsp}'s dependent type at the beginning of a function and appears in $\sigma_f$}
    \todo{Explain that the last constraint is only to constrain memory state (that's why we can use ``exist'' instead of forall.}
    \todo{Explain the meaning of $\spVar$ somewhere else}
    \label{fig:well-formed-states-revision}
\end{figure}

\subsection{Type Safety Proof}
\label{sec:type-safety-proof}
\begin{theorem}[Type Safety]
If $\typeChecked{P, \PType}{S}$, then there exists $S_1$ such that $S\rightarrow S_1$ and $\typeChecked{P, \PType}{S_1}$, or $S$ is a termination state.
\label{thm:type-safety-revision}
\end{theorem}
\begin{proof}
Denote $S=(R, M, (\pcVar, (\sigma_f, \sigma_b)); \Phi)$.
According to \textsc{Well-Formed-State-Base} and \textsc{Well-Formed-State-Inductive}, $\pcVar$ is a valid PC in $P$.
Let $(f, \instVar)=\text{getFuncInst}(P, \pcVar)$.
We also denote $(\Delta, \RType, \MType)=\text{getStateType}(\PType, \pcVar)$ as the state type for $S$, and $\sigma=\sigma_f\cup\sigma_b$ as the full type substitute to instantiate all type variables in the state type using constants.
Consider the following cases for $\instVar$:

\textbf{Case $\inst{movq}{r_1, \memOp{r_b}{r_i}{i_s}{i_d}\opMode{\sOp, \tauOp}}$.}
First, according to \textsc{Typing-Movq-r-m}, registers $r_b$, $r_i$ are untainted, and we can denote their type as $\RType[r_b]=(e_b, 0)$, $\RType[r_i]=(e_i, 0)$.
Then, according to \textsc{Reg-Type}, there should be $R[r_b]=(v_b, 0)$ and $R[r_i]=(v_i, 0)$, where $v_b=\sigma(e_b)$ and $v_i=\sigma(e_i)$.
We further denote $\RType[r_1]=(e_r, \tau_r)$ and $R[r_1]=(v_r, t_r)$, where $v_r=\sigma(e_r)$ and $t_r=\sigma(\tau_r)$.
According to both \textsc{Typing-StoreOp-Spill} and \textsc{Typing-StoreOp-Non-Spill}, there should always be $\hasVal{\Delta}{\tau_r \Rightarrow \tauOp}$, so $t_r \Rightarrow \sigma(\tauOp)$.
Therefore, $S$ satisfies constraints in \textsc{Dyn-Movq-r-m} and is allowed to execute $\instVar$.

Second, we construct the next state $S_1$ as follows:
\begin{align*}
    e_a & =i_d+e_b+e_i\times i_s \quad v_a=\sigma(e_a)=i_d+v_b+v_i\times i_s+i_d\\
    M_1 & =M[[v, v+8)\mapsto (v_r, t_r)] \quad \pcVar_1=\text{nextPc}(P, \pcVar)\\
    S_1 & = (R, M_1, (\text{nextPc}(P, \pcVar), (\sigma_f, \sigma_b));\Phi).
\end{align*}

Third, we prove that $S_1$ is well-formed, i.e., $\typeChecked{P, \PType}{S_1}$.
According to \textsc{Typing-Movq-r-m}, there should exists $\sval_o$, $e_o$, and $\tau_o$ such that
\begin{equation*}
    \hasType{\Delta, \RType, \MType}{\text{store}(\memOp{r_b}{r_i}{i_s}{i_d}\opMode{\sOp, \tauOp}, 8, \RType[r_1])}{(\sval_o, (e_o, \tau_o))},
\end{equation*}
and for $\MType_1\coloneq \MType[\sOp \mapsto (\sval_o, (e_o, \tau_o))]$, there should be $(\Delta, \RType, \MType_1)=\text{getStateType}(\PType, \pcVar_1)$.

We aim to prove that $\hasType{}{(R, M_1, (\sigma_f, \sigma_b))}{(\Delta, \RType, \MType_1)}$.
According to \textsc{Reg-Mem-Type}, as all other requirements can be easily derived by unfolding $\hasType{}{(R, M, (\sigma_f, \sigma_b))}{(\Delta, \RType, \MType)}$ (implied by $\typeChecked{P, \PType}{S}$),
we just need to prove $\hasType{}{M_1}{\sigma(\MType_1)}$.

For each $s\in \dom{\MType_1}$, denote $\MType_1[s]=(\sval_1, (e_1, \tau_1))$.
According to \textsc{Mem-Type}, we need to prove the following statements
\begin{equation}
\begin{aligned}
& (s\subseteq \stackPub \lor s \subseteq \otherPub \lor s\subseteq \otherSec)\\
& (\text{isSpill}(s) \Rightarrow s\subseteq \stackPub)
\quad (s\subseteq \otherPub \Rightarrow \tau_1 = 0)\\
& (\sval_1 \subseteq s)\quad (\hasType{}{M}{\MType[s]})
\end{aligned}
\label{eq:safety-store-mem-goal}
\end{equation}
If $s\neq \sOp$, then $\MType_1[s]=\MType[s]$ and
$\hasVal{\Delta}{s\cap \sOp=\emptyset}$ since all memory slots are non-overlapped. 
According to both \textsc{Typing-StoreOp-Spill} and \textsc{Typing-StoreOp-Non-Spill}, there should be $\hasVal{\Delta}{[e_a, e_a+8)\subseteq \sOp}$. 
Thus, there should be $\hasVal{\Delta}{[e_a, e_a+8)\cap s=\emptyset}$, so $[v_a, v_a+8)\cap \sigma(s)=\emptyset$.
This implies that for all $x\in \sigma(s)$, $M_1[x]=M[x]$.
Therefore, (\ref{eq:safety-store-mem-goal}) holds.

If $s=\sOp$, then $\sval_1=\sval_o$, $e_1=e_o$, $\tau_1=\tau_o$, and
we need to derive $\sval_o$, $e_o$, $\tau_o$ to prove (\ref{eq:safety-store-mem-goal}). 
Denote $\MType[\sOp]=(\sval_{o0}, (\_, \tau_{o0}))$.
We consider the following cases:
\begin{itemize}
    \item $\text{isSpill}(\sOp)$: According to \textsc{Tying-StoreOp-Spill}, there should be $\sval_o=[e_a, e_a+8)$, $e_o=e_r$, and $\tau_o=\tauOp$.
    \item $\neg \text{isSpill}(\sOp)$ and $\sval_{o0}\subseteq [e_a, e_a+8)$:
    According to \textsc{Typing-StoreOp-Non-Spill}, there should be $\sval_o=[e_a, e_a+8)$, $e_o=e_r$, and $\tau_o=\tauOp$.
    \item $\neg \text{isSpill}(\sOp)$ and $\sval_{o0}\not\subseteq [e_a, e_a+8)$:
    According to \textsc{Typing-StoreOp-Non-Spill}, there should be $\sval_o=[e_a, e_a+8)\cup \sval_{o0}$, $e_o=\top$, and $\tau_o=\tauOp$.
\end{itemize}
According to the following discussion, we can know that if $\sOp \subseteq \otherPub$, then $\tau_o=0$:
\begin{itemize}
    \item $\text{isSpill}(\sOp)$: according to \textsc{Mem-Type} and $\typeChecked{P, \PType}{S}$, we know that $\text{isSpill}(\sOp)\Rightarrow \sOp\subseteq \stackPub$. So $\sOp \not \subseteq \otherPub$ and we can ignore this case.
    \item $\neg \text{isSpill}(\sOp)$: According to \textsc{Typing-StoreOp-Spill}, $\tau_{o0}=\tauOp=\tau_o$.
    Furthermore, by unfolding $\typeChecked{P, \PType}{S}$ we can know that $\sOp \subseteq \otherPub \Rightarrow \tau_{o0}=0$. Therefore, $\tau_o=0$.
\end{itemize}
Furthermore, other statements in (\ref{eq:safety-store-mem-goal}) can be proved by unfolding $\typeChecked{P, \PType}{S}$ and performing a discussion on whether $\text{isSpill}(\sOp)$ is true or not.
Therefore, $\hasType{}{(R, M_1, (\sigma_f, \sigma_b))}{(\Delta, \RType, \MType_1)}$.

If $f\coloneq \text{getFunc}(P, \pcVar)=\mathit{main}$, \textsc{Well-Formed-State-Base} implies that $\typeChecked{P, \PType}{S_1}$.

We consider the case where $f\neq \mathit{main}$.
There should exists $\pcVar_p$, $\sigmaFp$, $\sigmaBp$, and $\Phi_0$ such that $\Phi=(\pcVar_p, (\sigmaFp, \sigmaBp));\Phi_0$.
Furthermore, $\typeChecked{P, \PType}{S}$ implies that there exists $R_p$, $M_p$ such that
$\typeChecked{P, \PType}{(R_p, M_p, \Phi)}$.

Note that $[v_a, v_a+8)\subseteq \sigma(\sOp)\subseteq \text{getDom}(\MType, \sigma_f)$, and $\sRet\cap \text{getDom}(\MType, \sigma_f)=\emptyset$ (implied by Lemma~\ref{lemma:no-ret-slot}), and $M_1[x]=M[x]$ for all $x\not\in [v_a, v_a+8)$, there should be 
\begin{align*}
    & M_1[\sRet] = M[\sRet] = (\text{nextPc}(P, \pcVar_p), 0)\\
    & \forall x \not \in \text{getDom}(\MType, \sigma_f) \cup \sRet.\; M_1[x]=M[x]=M_p[x].
\end{align*}
Therefore, the statement also holds for the case when $f\neq \mathit{main}$.

\todo{\textbf{Case $\inst{movq}{\memOp{r_d}{r_i}{i_s}{i_d}\opMode{\sOp, \tauOp}, r_0}$.}}

\textbf{Case $\inst{jne}{\ell}\opMode{\sigmaOp}$.}
First, according to \textsc{Typing-Jne}, flag $\reg{ZF}$ is untainted and we can denote its type as $\RType[\reg{ZF}]=(e=0, 0)$.
Then, according to \textsc{Reg-Type}, \textsc{Value-Type}, and well-formedness of $S$, there should be $R[\reg{ZF}]=(v, 0)$ where $v=\sigma(e=0)$.
Thus, $S$ satisfies constraints in \textsc{Dyn-Jne} and is allowed to execute $\instVar$.
\shixin{Note: I did not prove/constrain that the next PCs for both cases are in the same function.}

Second, we construct the next state $S_1$ as follows:
\begin{equation*}
    S_1 =
    \begin{cases}
        (R, M, (\text{nextPc}(P, \pcVar), (\sigma_f, \sigma_b); \Phi) & v = \mathit{true}\\
        (R, M, (\text{getBlockPc}(P, \ell), (\sigma_f, (\sigma_f\cup \sigma_b)\circ \sigmaOp); \Phi) & v = \mathit{false}.
    \end{cases}
\end{equation*}

Third, we prove that $S_1$ is well-formed under both cases (when the branch is taken and not taken).
When $v=\mathit{true}$, i.e., the branch is not taken, denote the next PC as $\pcVar_1=\text{nextPc}(P, \pcVar)$.
According to \textsc{Typing-Jne}, $\text{getStateType}(\PType, \pcVar_1)=(\Delta \cup \qty{e=0}, \RType, \MType)$.
Note that $\sigma(e=0)=v=\mathit{true}$ and $\hasType{}{(R, M, (\sigma_f, \sigma_b))}{(\Delta, \RType, \MType)}$, so there should also be $\hasType{}{(R, M, (\sigma_f, \sigma_b))}{(\Delta \cup \qty{e=0}, \RType, \MType)}$.
Thus, we can derive that $\typeChecked{P, \PType}{S_1}$ according to \textsc{Well-Formed-State-Base} and \textsc{Well-Formed-State-Inductive}.

When $v=\mathit{false}$, i.e., the branch is taken, denote next PC as $\pcVar_1=\text{getBlockPc}(P, \ell)$.
Similarly to the not-taken case, we can prove that
\begin{equation}
    \hasType{}{(R, M, (\sigma_f, \sigma_b))}{(\Delta\cup \qty{e \neq 0}, \RType, \MType)}.
    \label{eq:jne-taken-before}
\end{equation}
According to \textsc{Typing-Jne}, we can also know that
\begin{equation}
    (\Delta\cup \qty{e \neq 0}, \RType, \MType) \sqsubseteq \sigmaOp(\PType(f)(\ell)).
    \label{eq:jne-taken-subtype}
\end{equation}
Since $\pcVar_1$ is the PC of block $\ell$ in function $f$, the state type for $\pcVar_1$ should be $\text{getStateType}(\PType, \pcVar_1)=\PType(f)(\ell)$.
We aim to prove that $\hasType{}{(R, M, (\sigma_f, (\sigma_f\cup \sigma_b)\circ \sigmaOp))}{\PType(f)(\ell)}$.
Denote $(\Delta_1, \RType_1, \MType_1)=\PType(f)(\ell)$ and $\sigma_1=\sigma_f\cup ((\sigma_f\cup \sigma_b)\circ \sigmaOp)$. According to \textsc{Reg-Mem-Type}, we just need to prove the following statements:
\todo{Consider to replace $\sigma_f\cup \sigma_b$ with $\sigma$.}
\begin{itemize}
    \item $\sigma_1(\Delta_1)=\mathit{true}$:
    According to Lemma~\ref{lemma:br-substitute}, $\sigma_1(\Delta_1)=((\sigma_f\cup \sigma_b)\circ \sigmaOp)(\Delta_1)=(\sigma_f\cup \sigma_b)(\sigmaOp(\Delta_1))$.
    According to (\ref{eq:jne-taken-subtype}) and \textsc{State-Subtype}, $\hasVal{\Delta\cup \qty{e \neq 0}}{\sigmaOp(\Delta_1)}$.
    According to (\ref{eq:jne-taken-before}) and Lemma~\ref{lemma:stmt-instance}, we can know that $(\sigma_f\cup \sigma_b)(\sigmaOp(\Delta_1))=\mathit{true}$.
    Thus, the statement is true.
    \item $\hasType{}{R}{\sigma_1(\RType_1)}$:
    For each $r\in \dom{\RType_1}$, denote $R[r]=(v, t)$ and $\RType_1[r]=(e_1, \tau_1)$.
    According to (\ref{eq:jne-taken-subtype}) and \textsc{Reg-Subtype}, $r\in \RType[r]$, and when denoting $\RType[r]=(e, \tau)$, there should be
    \begin{align*}
        & \hasVal{\Delta}{e=\sigmaOp(e_1) \lor (\text{isNonChangeExp}(e)\land \sigmaOp(e_1)=\top)}\\
        & \hasVal{\Delta}{\tau \Rightarrow \sigmaOp(\tau_1)}.
    \end{align*}
    According to (\ref{eq:jne-taken-before}) and Lemma~\ref{lemma:stmt-instance}, the above statements still hold when we instantiate their type variables with $\sigma=(\sigma_f\cup \sigma_b)$ (and also simplify with Lemma~\ref{lemma:br-substitute}), so we can get
    \begin{align*}
        & \sigma(e)=\sigma_1(e_1) \lor (\text{isNonChangeExp)}(\sigma(e)) \land \sigma_1(e_1)=\top)\\
        & \sigma(\tau) \Rightarrow \sigma_1(\tau).
    \end{align*}
    Furthermore, (\ref{eq:jne-taken-before}) also implies that
    \begin{equation*}
        v=\sigma(e) \lor \sigma(e)=\top \qquad t\Rightarrow \sigma(\tau).
    \end{equation*}
    By combining the above statements and performing a simple case discussion on whether $\sigma_1(e_1)=\top$, we can finally get
    \begin{equation*}
        v=\sigma_1(e_1) \lor \sigma_1(e_1)=\top \qquad t\Rightarrow \sigma_1(\tau_1).
    \end{equation*}
    Thus, the statement is true.
    \item $\hasType{}{M}{\sigma_1(\MType_1)}$:
    For each $s_1\in \dom{\MType_1}$, denote $\MType_1[s_1]=(\sval_1, (e_1, \tau_1))$.
    According to (\ref{eq:jne-taken-subtype}) and \textsc{State-Subtype}, there exists $s\in \MType$ such that
    \begin{equation*}
        \hasVal{\Delta}{\sigmaOp(s_1)\subseteq s \land \MType[s] \sqsubseteq \sigmaOp(\MType_1[s_1])}.
    \end{equation*}
    Denote $\MType[s]=(\sval, (e, \tau))$. Then, according to \textsc{Mem-Slot-Subtype}, there should be
    \begin{align*}
        & \hasVal{\Delta}{\sigmaOp(\sval_1)\subseteq \sval}\\
        & \hasVal{\Delta}{\text{isSpill}(\sigmaOp(\sval_1)) \Rightarrow \text{isSpill}(\sval)}\\
        & \hasVal{\Delta}{e=\sigmaOp(e_1) \lor (\text{isNonChangeExp}(e) \land \sigmaOp(e_1)=\top) \lor \sigmaOp(\sval_1)=\emptyset}\\
        & \hasVal{\Delta}{\tau = \sigmaOp(\tau_1) \lor (\text{isSpill}(\sval)\land \sigmaOp(\sval_1)=\emptyset)}.
    \end{align*}
    According to (\ref{eq:jne-taken-before}), the above statements still hold when we instantiate their type variables with $\sigma=(\sigma_f \cup \sigma_b$) (and also simplify with Lemma~\ref{lemma:br-substitute}), so we can get 
    \begin{align*}
        & \sigma_1(s_1) \subseteq \sigma(s)\\
        & \sigma_1(\sval_1)\subseteq \sigma(\sval)\\
        & \text{isSpill}(\sigma_1(\sval_1)) \Rightarrow \sigma(\sval)\\
        & \sigma(e)=\sigma_1(e_1) \lor (\text{isNonChangeExp}(\sigma(e)) \land \sigma_1(e_1)=\top) \lor \sigma_1(\sval_1)=\emptyset\\
        & \sigma(\tau)=\sigma_1(\tau_1) \lor (\text{isSpill}(\sigma(\sval)) \land \sigma_1(\sval_1)=\emptyset).
    \end{align*}
    (\ref{eq:jne-taken-before}) implies that $\hasType{}{M}{\sigma(\MType)}$, so there should be
    \begin{align*}
        & \sigma(s) \subseteq \stackPub \lor \sigma(s)\subseteq \otherPub \lor \sigma(s) \subseteq \otherSec\\
        & \text{isSpill}(\sigma(s)) \Rightarrow \sigma(s)\subseteq \stackPub\\
        & \sigma(s) \subseteq \otherPub \Rightarrow \sigma(\tau)=0
        \quad\quad \sigma(\sval)\subseteq s\\
        & \hasType{}{M}{\sigma(\MType[s])}.
    \end{align*}
    \shixin{For now, I omit details here.}
    By combining the above statements and performing a simple discussion on $\sigma_1(\sval_1)$ and $\sigma_1(\MType_1[s_1])$, we can know that the statement is true.
    \item $\forall s\in \dom{\MType_1}.\; \text{getVars}(s)=\dom{\sigma_f}$:
    According to Lemma~\ref{lemma:same-dom-mtype}, $\dom{\MType_1}=\dom{\MType}$.
    According to (\ref{eq:jne-taken-before}), for each $s$ in $\dom{\MType}$, $\text{getVars}(s)=\dom{\sigma_f}$. Thus, the statement is true.
\end{itemize}

So far, we have successfully proved that $\hasType{}{(R, M, (\sigma_f, (\sigma_f\cup \sigma_b)\circ \sigmaOp))}{\PType(f)(\ell)}$.
Therefore, it is straightforward to derive that $\typeChecked{P, \PType}{S_1}$ from this statement and $\typeChecked{P, \PType}{S}$.

\textbf{Case $\inst{callq}{f_c}\opMode{\sigmaCall,\sigmaRet}$.}
First, according to \textsc{Typing-Callq}, $r_\reg{rsp}$ is untainted, and we can denote its type as $\RType[r_\reg{rsp}]=(e, 0)$.
Then, according to \textsc{Reg-Type}, \textsc{Value-Type}, and well-formedness of $S$, there should be $R[r_\reg{rsp}]=(v, 0)$ where $v=\sigma(e)$.
Thus, $S$ satisfies constraints in \textsc{Dyn-Callq} and is allowed to execute $\instVar$.

Second, we construct the next state $S_1$ as follows:
\begin{align*}
    \pcVar_c & = \text{getBlockPc}(P, f_c)\\
    R_1 & = R[r_\reg{rsp} \mapsto (v-8, 0)]\\
    M_1 & = M[[v-8, v) \mapsto (\text{nextPc}(P, \pcVar), 0)]\\
    \Phi_1 & = (\pcVar_c, (\sigma \circ \sigmaCall, []\rightarrow[]));(\pcVar, (\sigma_f, \sigma_b)); \Phi\\
    S_1 & = (R_1, M_1, \Phi_1).
\end{align*}
Furthermore, let $\RType_1=\RType[r_\reg{rsp}\mapsto (e-8, 0)]$.
$\typeChecked{P, \PType}{S}$ implies $\hasType{}{(R, M, (\sigma_f, \sigma_b))}{(\Delta, \RType, \MType)}$.
Thus, we can derive $\hasType{}{(R_1, M_1, (\sigma_f, \sigma_b))}{(\Delta, \RType_1, \MType)}$ by proving constraints in \textsc{Reg-Mem-Type}.
Specifically, the constraint $\hasType{}{M_1}{\sigma(\MType)}$ can be implied from $\hasType{}{M}{\sigma(\MType)}$, $\MType[e-8, e]=(\emptyset, \_)$, and $M_1[x]=M[x]$ for all $x\not\in[v-8, v]$.
According to \textsc{Typing-Callq}, there should also be $(\Delta, \RType_1, \MType)\sqsubseteq \sigmaCall(\PType(f_c)(f_c))$.

Third, we prove that $S_1$ is well-formed.
Since $\pcVar_c$ is the PC of $f_c$, the state type for $\pcVar_c$ should be $\text{getStateType}(\PType, \pcVar_c)=\PType(f_c)(f_c)$.
We aim to prove that $\hasType{}{(R_1, M_1, (\sigma \circ \sigmaCall, []\rightarrow[]))}{\PType(f_c)(f_c)}$.
Denote $\sigma_c=\sigma \circ \sigmaCall$ and $(\Delta_c, \RType_c, \MType_c)=\PType(f_c)(f_c)$.
According to \textsc{Reg-Mem-Type}, we just need to prove the following statements:
\begin{itemize}
    \item $\sigma_c(\Delta_c)=\mathit{true}$, $\hasType{}{R_1}{\sigma_c(\RType_c)}$, and $\hasType{}{M_1}{\sigma_c(\MType_c)}$:
    Note that $\hasType{}{(R_1, M_1, (\sigma_f, \sigma_b))}{(\Delta, \RType_1, \MType)}$ and $(\Delta, \RType_1, \MType)\sqsubseteq \sigmaCall(\PType(f_c)(f_c))$, we can prove the three statements following a similar strategy used to prove for the case of $\instr{jne}$.
    \todo{Convert this to a lemma for both cases.}
    \item $\forall s\in \dom{\MType_c}.\; \text{getVars}(s)=\dom{\sigma_c}$: this is true since $\sigmaCall$ is a substitute that instantiate all variables in $\PType(f_c)(f_c)$.
\end{itemize}
Therefore, $\hasType{}{(R_1, M_1, (\sigma \circ \sigmaCall, []\rightarrow[]))}{\PType(f_c)(f_c)}$.

Let $s_\text{retc}=[\sigma_c(\spVar), \sigma_c(\spVar)+8)$.
According to $(\Delta, \RType_1, \MType)\sqsubseteq \sigmaCall(\PType(f_c)(f_c))$ and \textsc{Reg-Subtype}, there should be $e-8=\sigmaCall(\spVar)$, so
$\sigma_c(\spVar)=\sigma(\sigmaCall(\spVar))=\sigma(e-8)=v-8$.
Thus, we have $M_1[s_\text{retc}]=M_1[v-8, v]=(\text{nextPc}(P, \pcVar), 0)$.

To prove $\typeChecked{P, \PType}{S_1}$, we still need to prove the following statements (here we intend to use $R$ and $M$ to prove the existence of $R_p$ and $M_p$ in \textsc{Well-Formed-State-Inductive}):
\begin{itemize}
    \item $\forall x\not\in \text{getDom}(\MType_c, \sigma_c)\cup s_\text{retc}.\; M_1[x]=M[x]$: this hold since $M_1[x]=M[x]$ for all $x\not\in s_\text{retc}$.
    \item $\text{getInst}(P,\pcVar)=\inst{callq}{f_c\opMode{\sigmaCall,\sigmaRet}}$: this is our assumption.
    \item $\sigma_c=(\sigma_f\cup \sigma_b)\circ\sigmaCall$: this is the definition of $\sigma_c$.
    \item $\typeChecked{P, \PType}{(R, M, (\pcVar, (\sigma_f, \sigma_b));\Phi}$: this is our assumption.
\end{itemize}

\textbf{Case $\inst{retq}{}$.}
First, according to \textsc{Typing-Func}, $f\neq \mathit{main}$.
Then, there should exists $\pcVar_p$, $\sigmaFp$, $\sigmaBp$ and $\Phi_0$ such that $\Phi = (\pcVar_p, (\sigmaFp, \sigmaBp)); \Phi_0$.

According to \textsc{Typing-Func}, the state type of $\pcVar$ is $(\Delta, \RType, \MType)\coloneq \text{getStateType}(f)(f_\text{ret})$.
It also implies that $\RType[r_\reg{rsp}]=(\spVar, 0)$.
Then, according to \textsc{Reg-Type}, \textsc{Value-Type}, and well-formedness of $S$, there should be $R[r_\reg{rsp}]=(v, 0)$, where $v=\sigma(\spVar)=\sigma_f(\spVar)$.
Furthermore, according to $\typeChecked{P, \PType}{S}$ and \textsc{Well-Formed-State-Inductive}, there should exist $R_p$ and $M_p$ such that
\begin{align*}
    & \sRet =[\sigma_f(\spVar), \sigma(\spVar)+8)=[v, v+8)\\
    & M[\sRet] = (\text{nextPc}(P, \pcVar_p), 0)\\
    & \text{getInst}(P, \pcVar_p)=\inst{callq}{f\opMode{\sigmaCall,\sigmaRet}}\\
    & \forall x\not\in \text{getDom}(\MType, \sigma_f)\cup \sRet.\; M[x]=M_p[x]\\
    & \sigma_f=(\sigmaFp\cup \sigmaBp)\circ \sigmaCall\\
    & \typeChecked{P, \PType}{(R_p, M_p, \Phi)}.
\end{align*}
Thus, $S$ satisfies all constraints in \textsc{Dyn-Retq} and is allowed to execute $\instVar$.

Second, we construct the next state $S_1$ as follows:
\begin{align*}
    R_1 & = R[r_\reg{rsp}\mapsto (v+8, 0)]\\
    \Phi_1 & = (\text{nextPc}(P, \pcVar_p), (\sigmaFp, (\sigma_b \circ \sigmaRet^{-1})\cup \sigmaBp);\Phi_0\\
    S_1 & = (R_1, M, \Phi_1).
\end{align*}

Third, we prove that $S_1$ is well-formed.
Let
\begin{align*}
    (\Delta_0, \RType_0, \MType_0) & = \text{getStateType}(\PType, \pcVar_p)\\
    (\Delta_1, \RType_1, \MType_1) & = (\sigmaCall\cup \sigmaRet)(\Delta, \RType, \MType)\\
    (\Delta_2, \RType_2, \MType_2) & =
    (\Delta_0\cup \Delta_1, \RType_1[r_\reg{rsp}\mapsto \RType_0[r_\reg{rsp}]], \\
    & \qquad \text{updateMem}(\MType_0, \MType_1)).
\end{align*}
According to \textsc{Well-Formed-State-Inductive}, there should be $\typeChecked{P, \PType}{(R_p, M_p, \Phi)}$, so 
\begin{equation}
    \hasType{}{(R_p, M_p, (\sigmaFp, \sigmaBp))}{(\Delta_0, \RType_0, \MType_0)}.
    \label{eq:call-state-type}
\end{equation}

According to rule \textsc{Typing-Callq}, there should be $(\Delta_2, \RType_2, \MType_2)=\text{getStateType}(\PType, \text{nextPc}(P, \pcVar_p))$.
We aim to prove
\begin{equation*}
\hasType{}{(R_1, M, (\sigmaFp, (\sigma_b\circ \sigmaRet^{-1})\cup \sigmaBp))}{(\Delta_2, \RType_2, \MType_2)}.
\end{equation*}
Denote $\sigma_1=\sigmaFp\cup (\sigma_b\circ \sigmaRet^{-1})\cup \sigmaBp$.
According to \textsc{Reg-Mem-Type}, we just need to prove the following:
\begin{itemize}
    \item $\sigma_1(\Delta_2)=\qty{\mathit{true}}$:
    For each statement $e$ in $\Delta_2$, there should be $e\in \Delta_0$ or $e\in \Delta_1$.
    If $e\in \Delta_0$, then $\sigma_1(e)=(\sigmaFp\cup \sigmaBp)(e)=\mathit{true}$ (implied by (\ref{eq:call-state-type})).
    If $e\in \Delta_1$, then there exists $e'\in \Delta$ such that $e=(\sigmaCall\cup \sigmaRet)(e')$ and $(\sigma_f\cup \sigma_b)(e')=\mathit{true}$. 
    According to \textsc{Well-Formed-State-Inductive}, we have $\sigma_f=(\sigmaFp\cup \sigmaBp)\circ \sigmaCall$.
    Then, according to Lemma~\ref{lemma:call-ret-substitute},
    \begin{align*}
        \sigma_1(e) & = (\sigmaFp\cup (\sigma_b\circ\sigmaRet^{-1})\cup \sigmaBp)(e)\\
        & = ((\sigmaFp\cup (\sigma_b\circ\sigmaRet^{-1})\cup \sigmaBp)\circ (\sigmaCall\cup \sigmaRet))(e')\\
        & = (((\sigmaFp\cup \sigmaBp)\circ \sigmaCall) \cup \sigma_b) (e')\\
        & = (\sigma_f\cup \sigma_b)(e') = \mathit{true}.
    \end{align*}
    Therefore, the statement holds.
    \item $\hasType{}{R_1[r_\reg{rsp}]}{\sigma_1(\RType_2[r_\reg{rsp}])}$:
    According to the definition of $(\Delta_0, \RType_0, \MType_0)$ and \textsc{Typing-Callq}, there should exists $e_0$ such that $\RType_0[r_\reg{rsp}]=(e_0, 0)$.
    Let $\RType_{p0}=\RType_0[r_\reg{rsp}\mapsto (e-8, 0)]$, we can also get $(\Delta_0, \RType_{p0}, \MType_0)\sqsubseteq \sigma_\text{call}(\PType(f)(f))$.
    Thus, we can derive that $e_0-8=\sigmaCall(\spVar)$.
    Then, according to Lemma~\ref{lemma:call-ret-substitute},
    \begin{align*}
        \sigma_1(\RType_2[r]) & = (\sigma_1(\sigmaCall(\spVar+8)), 0)\\
        & = ((\sigma_1\circ (\sigmaCall\cup \sigmaRet))(\spVar+8), 0)\\
        & = (\sigma(\spVar+8), 0) = (v+8, 0)=R_1[r_\reg{rsp}].
    \end{align*}
    Therefore, the statement holds.
    \item $\forall r\in \dom{\RType_2}, r\neq r_\reg{rsp}.\; \hasType{}{R'[r]}{\sigma_1(\RType_2[r])}$: According to Lemma~\ref{lemma:call-ret-substitute}, we have $\sigma_1(\RType_2[r])=(\sigma_1 \circ (\sigmaCall\cup \sigmaRet))(\RType[r])=(\sigma_f\cup \sigma_b)(\RType[r])$ and $R'[r]=R[r]$. Furthermore, we can derive $\hasType{}{R[r]}{(\sigma_f\cup \sigma_b)(\RType[r])}$ from the well-formedness of $S$. Thus, the original statement holds.
    \item $\hasType{}{M}{\sigma_1(\MType_2)}$: 
    For each $s\in \dom{\MType_2}$, denote $\MType_2[s]=(\sval_2, (e_2, \tau_2))$.
    According to \textsc{Mem-Type}, as other requirements of well-formedness are straightforward to prove, we focusing on proving that $\hasType{}{M}{\sigma_1(\MType_2[s])}$.

    According to Algorithm~\ref{alg:update-mem}, when executing $\text{updateMem}$ to generate $\MType_2$, we generate a map $S$ that maps each slot in $\dom{\MType_2}$ to slots in $\dom{\MType_1}$ that belong to $s$. Consider the following cases:
    \begin{itemize}
        \item $s=\sigmaCall([\spVar, \spVar+8))=[\spVar+c-8, \spVar+c)$ (where $\sigmaCall(\spVar)=\spVar+c-8$ and $\RType_2[r_\reg{rsp}]=\RType_0[r_\reg{rsp}]=(\spVar+c, \_)$): according to \textsc{Typing-Callq}, $\MType_0[s]=(\emptyset, \_)$. According to \textsc{Typing-Func}, for all $s_1\in\dom{\MType_1}$, $s_1\cap s= \emptyset$, so $S[s]=\emptyset$.
        These two facts imply that $\sval_2=\emptyset$, so $\sigma_1(\sval_2)=\emptyset$ and the statement holds according to \textsc{Mem-Slot-Empty-Type}.
        \item $S[s]=\emptyset$ and $s\neq [\spVar+c-8, \spVar+c)$: in this case, we have $\sval_2=\sval_0$ and $e_2=e_0$.
        According to the well-formedness of $R_p$, $M_p$, there should be $\hasType{}{M_p}{\sigma_1(\MType_0[s])}$.
        Recall that $\forall x\not\in \text{getDom}(\MType, \sigma_f)\cup \sRet.\; M[x]=M_p[x]$.
        By applying Lemma~\ref{lemma:possible-slot-non-overlap}, it is also straightforward to derive that $\sigma_1(\sval_2)\cap (\text{getDom}(\MType, \sigma_f)\cup \sRet) = \emptyset$,
        so $M[\sigma_1(\sval_2)]=M_p[\sigma_1(\sval_0)]$.
        Thus, we can derive $\hasType{}{M}{\sigma_1(\MType_2[s])}$ from $\hasType{}{M_p}{\sigma_1(\MType_0[s])}$.
        \item $\text{cardinal}(S[s])>1$: according to $\text{updateMem}$, $\sval_2=\sval_1=\emptyset$, so $\sigma_1(\sval_2)=\emptyset$ and the statement holds according to \textsc{Mem-Slot-Empty-Type}.
        \item $\text{cardinal}(S[s])=1$: denote $S[s]=\qty{s_1}$ and $\MType_1[s_1]=(\sval_1, (e_1, \tau_1))$.
        Then, $\sval_2=(\sval_0\backslash s_1)\cup \sval_1$ and $\tau_2=\tau_1$ according to $\text{updateMem}$.
        According to the definition of $\MType_1$, there exists $s_c\in \MType$ such that $s_c=(\sigmaCall \cup \sigmaRet)(s_1)$.
        Then, we have $\MType_1[s_1]=(\sigmaCall\cup \sigmaRet)(\MType_c[s_c])$.

        By unfolding the assumption $\typeChecked{P, \PType}{S}$, we can get $\hasType{}{M}{\sigma(\MType[s_c])}$.
        According to Lemma~\ref{lemma:call-ret-substitute}, for all $e$, $\sigma(e)=\sigma_1((\sigmaCall\cup \sigmaRet)(e))$.
        These two facts imply that $\hasType{}{M}{\sigma_1(\MType_1[s_1])}$.

        We then consider the following cases.

        If $\sval_2=\sval_1$, then $e_2=e_1$. We can derive $\hasType{}{M}{\sigma_1(\MType_2[s])}$ from $\hasType{}{M}{\sigma_1(\MType_1[s_1])}$.

        If $\sval_2\subsetneqq \sval_1$, then $e_2=\top$. According to \textsc{Mem-Slot-Top-Type}, $\hasType{}{M}{\sigma_1(\MType_2[s])}$ holds.
    \end{itemize}
\end{itemize}
\end{proof}

\begin{lemma}
For all $P$, $\PType$, $\pcVar_0$, $\pcVar_1$,
and for $i\in \qty{0, 1}$, $(\Delta_i, \RType_i, \MType_i)=\text{getStateType}(\PType, \pcVar_i)$,
if $\text{getFunc}(P, \pcVar_0)=\text{getFunc}(P, \pcVar_1)$,
then there should be $\dom{\MType_0}=\dom{\MType_1}$.
\label{lemma:same-dom-mtype}
\end{lemma}
\begin{proof}
This lemma means that memory state types of instructions in the same function have the same memory slots.

According to \textsc{Typing-Func} and \textsc{Typing-Prog}, all basic blocks in $P$ are well-typed.
By unfolding typing rules for each basic block, we know that $\dom{\MType_0}=\dom{\MType_1}$ holds if $\pcVar_0$ and $\pcVar_1$ are in the same basic block.
According to \textsc{Typing-Jne} (and similar typing rules for other branch instructions omitted here), $\dom{\MType_0}=\dom{\MType_1}$ holds if $\pcVar_0$ and $\pcVar_1$ refer to a branch instruction and its target.
Therefore, we can derive that $\dom{\MType_0}=\dom{\MType_1}$ holds as long as $\pcVar_0$ and $\pcVar_1$ are in the same function.
\end{proof}

\begin{lemma}
For all $P$, $\PType$, $\pcVar$, and $(\Delta, \RType, \MType)=\text{getStateType}(\PType, \pcVar)$,
there should be
\begin{equation*}
    \forall s\in \dom{\MType}.\; \hasVal{\Delta}{s \cap [\spVar, \spVar+8)=\emptyset}.
\end{equation*}
\label{lemma:no-ret-slot}
\end{lemma}
\begin{proof}
Denote $f=\text{getFunc}(P, \pcVar)$, $\ell=\text{getFunc}(P, \pcVar)$, $(\Delta_0, \RType_0, \MType_0)=\PType(f)(f)$, and $(\Delta_1, \_, \_)=\PType(f)(\ell)$.
\textsc{Typing-Func} and typing rules for instructions sequences implies that
\begin{align*}
    & \forall s\in \dom{\MType_0}.\; \hasVal{\Delta_0}{s \cap [\spVar, \spVar+8)=\emptyset},\\
    & \Delta_0 \subseteq \Delta_1, \qquad \Delta_1=\Delta.
\end{align*}
Furthermore, Lemma~\ref{lemma:same-dom-mtype} implies that $\dom{\MType}=\dom{\MType_0}$.
Therefore, the statement holds.
\end{proof}

\begin{lemma}
Suppose $\hasType{}{(R, M, (\sigma_f, \sigma_b))}{(\Delta, \RType, \MType)}$ where $\sigma_f$ and $\sigma_b$ are substitutes that instantiate type variables using constants.
For any arithmetic statement $e$, if $\hasVal{\Delta}{e=\mathit{true}}$, then $(\sigma_f\cup\sigma_b)(e)=\mathit{true}$.
\label{lemma:stmt-instance}
\end{lemma}
\begin{proof}
According to \textsc{Reg-Mem-Type}, $(\sigma_f\cup\sigma_b)(\Delta)=\qty{\mathit{true}}$, which implies that $(\sigma_f\cup\sigma_b)$ is a valid substitute map that instantiates all type variables in $\Delta$.
Therefore, the statement holds.
\todo{@adam: do you think this is too naive to be a lemma? Should I keep this?}
\end{proof}

\begin{lemma}
For all type substitute $\sigma_f$, $\sigma_b$, and $\sigma_0$, if $\dom{\sigma_f}\cap \dom{\sigma_b}=\emptyset$ and $\dom{\sigma_f}\cap \dom{\sigma_0}=\emptyset$, then for all $e$,
\begin{equation*}
    (\sigma_f\cup ((\sigma_f\cup \sigma_b) \circ \sigma_0))(e)=((\sigma_f\cup \sigma_b)\circ \sigma_0)(e).
\end{equation*}
\label{lemma:br-substitute}
\end{lemma}
\begin{proof}
    First, we prove that the statement holds when $e=x$ is a type variable.
    If $x\in \dom{\sigma_f}$, then $(\sigma_f\cup ((\sigma_f\cup \sigma_b) \circ \sigma_0))(x) = \sigma_f(x)$.
    Since $\dom{\sigma_f}\cap \dom{\sigma_0}=\emptyset$, then $x\not\in \dom{\sigma_0}$.
    In this case, applying $\sigma_0$ to instantiate $x$ will still get $x$, i.e., $\sigma_0(x)=x$. Then,
    \begin{equation*}
        ((\sigma_f\cup \sigma_b) \circ \sigma_0)(x)= (\sigma_f\cup \sigma_b)(\sigma_0(x))=(\sigma_f\cup \sigma_b)(x)=\sigma_f(x).
    \end{equation*}
    Thus, the statement holds for this case.

    If $x\not \in \dom{\sigma_f}$, then 
    \begin{equation*}
        (\sigma_f\cup ((\sigma_f\cup \sigma_b) \circ \sigma_0))(e) = ((\sigma_f\cup \sigma_b) \circ \sigma_0)(e).
    \end{equation*}
    Therefore, the statement holds when $e$ is a type variable.
    
    For a general type expression $e$, the statement can be proved by applying the statement to each type variable in $e$.
\end{proof}

\begin{lemma}
For all type substitute $\sigmaFp$, $\sigmaBp$, $\sigma_b$, $\sigmaCall$, and $\sigmaRet$, if 
\begin{itemize}
    \item $\sigmaCall$ is a substitute that instnatiates each type variable in its domain using expressions over variables in $\dom{\sigmaFp}\cup \dom{\sigmaBp}$.
    \item $\sigmaRet=\overrightarrow{x_1}\rightarrow \overrightarrow{x_2}$ is a substitute that instantiates each type variable in its domain using another unique type variable;
    \item $\dom{\sigmaRet}=\dom{\sigma_b}$ and $\dom{\sigmaRet}\cap \dom{\sigmaCall}=\emptyset$;
    \item $\dom{\sigmaRet^{-1}}\cap \dom{\sigmaFp}=\dom{\sigmaRet^{-1}}\cap \dom{\sigmaBp}=\dom{\sigmaFp}\cap \dom{\sigmaBp}=\emptyset$
\end{itemize}
then for all $e$,
\begin{equation*}
    ((\sigmaFp \cup (\sigma_b\circ \sigmaRet^{-1})\cup \sigmaBp)\circ (\sigmaCall\cup \sigmaRet))(e)=(((\sigmaFp\cup \sigmaBp)\circ \sigmaCall)\cup \sigma_b)(e).
\end{equation*}
\label{lemma:call-ret-substitute}
\end{lemma}
\begin{proof}
First, we prove that the statement holds when $e=x$ is a type variable.
If $x\in \dom{\sigmaCall}$, then $x\not\in \sigma_b$ and $\sigmaCall(x)\in \dom{\sigmaFp}\cup \dom{\sigmaBp}$. Then, we have
\begin{align*}
    & ((\sigmaFp \cup (\sigma_b\circ \sigmaRet^{-1})\cup \sigmaBp)\circ (\sigmaCall\cup \sigmaRet))(x)\\
    =\; & ((\sigmaFp \cup (\sigma_b\circ \sigmaRet^{-1})\cup \sigmaBp)\circ \sigmaCall)(x)\\
    =\; & ((\sigmaFp \cup \sigmaBp)\circ \sigmaCall)(x)\\
    =\; & (((\sigmaFp \cup \sigmaBp)\circ \sigmaCall)\cup \sigma_b) (x).
\end{align*}
If $x\in \dom{\sigmaRet}$, then $x\in\dom{\sigma_b}$. Then, we have
\begin{align*}
    & ((\sigmaFp \cup (\sigma_b\circ \sigmaRet^{-1})\cup \sigmaBp)\circ (\sigmaCall\cup \sigmaRet))(x)\\
    =\; & ((\sigmaFp \cup (\sigma_b\circ \sigmaRet^{-1})\cup \sigmaBp)\circ \sigmaRet)(x)\\
    =\; & ((\sigma_b\circ \sigmaRet^{-1})\circ \sigmaRet)(x)
    = \sigma_b(x) \\
    =\; & (((\sigmaFp \cup \sigmaBp)\circ \sigmaCall)\cup \sigma_b) (x).
\end{align*}
Therefore, the statement holds when $e$ is a type variable.

For a general type expression $e$, the statement can be proved by applying the statement to each type variable in $e$.
\end{proof}

\section{Transformation}
\subsection{Base Transformation Pass}
\label{sec:base-trans-pass}
We begin with formalizing our base transformation pass $\Compiler$, which takes a well-typed program and generates a program that satisfies public noninterference.
For this pass, besides well-formedness, we additionally assume that for all functions of the input program, pointer arguments are passed via registers for ease of illustration.

In Algorithm~\ref{alg:compiler-base}, we define how $\Compiler$ transforms each instruction. In the following part, we will overload the same symbol to represent transforming instruction sequences or the whole program by applying $\Compiler$ to each instruction (e.g., $\Compiler(P)$).

In this algorithm, the transformation strategy map $\omega_p$ is generated as described in Algorithm~\ref{alg:trans-strategy}. The two passes $\CompilerOp$ and $\CompilerPtr$ are formalized in Algorithm~\ref{alg:compiler-op}-\ref{alg:compiler-ptr}.
\revision{
When generating the transformation strategy for each slot $s$, we additionally require that if $s$ is not a local stack slot (i.e., both the current function and its caller will access it), and its transformation strategy is $\transOp$ (i.e., $\omega(s)=\transOp$), then its corresponding taint type must be constant (either $0$ or $1$) (see the assertion in Algorithm~\ref{alg:trans-strategy}).
}
This restriction is crucial to guarantee functional correctness since the function and its caller must agree on where the slot should be relocated and whether to add $\delta$ to the memory operand when accessing the slot.
\todo{Prove that all PC in the same func share the same $\omega$}
\todo{One thing I may need to discuss somewhere: if callee's slot is \transOp, then its corresponding slot in caller must also be \transOp.}

\begin{algorithm}
\caption{Transform Instruction ($\Compiler$)}
\begin{algorithmic}
\Require Instruction to transform $\instVar$ and its $\pcVar$, original program $P$ and program type $\PType$
\Ensure Transformed instructions $\iota$
\State $F_p \gets \text{getFunc}(P, \pcVar)$
\State $(\Delta_p, \RType_p, \MType_p) \gets \text{getStateType}(\PType, \pcVar)$
\State $\omega_p \gets \text{getTransStrategy}(\MType_p)$
\Switch{$\instVar$}
    \Case{$\inst{movq/addq/cmpq}{\opVar_1, \opVar_0}$}
        \State $\iota \gets \inst{movq/addq/cmpq}{\CompilerOp(\omega_p, \opVar_1), \CompilerOp(\omega_p, \opVar_0)}$
    \EndCase
    \Case{$\inst{pushq/popq}{r}$}
        \State $\iota \gets \inst{pushsecq/popsecq}{r}$
    \EndCase
    \Case{$\inst{callq}{f_c\opMode{\sigmaCall, \sigmaRet}}$}
        \State $\SType_c \gets \text{getStateType}(\PType, \text{getBlockPc}(P, f_c))$
        \State $\iota_\text{call}, \sigmaCall' \gets \CompilerPtr((\Delta_p, \RType_p, \MType_p), \sigmaCall(\SType_c),\sigmaCall)$
        \State $\iota \gets \iota_\text{call}; \inst{callq}{f_c\opMode{\sigmaCall', \sigmaRet}}$
    \EndCase
    \Default
        \State $\iota \gets \instVar$
    \EndDefault
\EndSwitch
\end{algorithmic}
\label{alg:compiler-base}
\end{algorithm}

\begin{algorithm}
\caption{Determine Transformation Strategy}
\begin{algorithmic}
\Require Memory type $\MType$
\Ensure Transformation strategy $\omega$
\State $T \gets \qty{}$
\State \Comment{$T$ is a map from each base pointers to a set of taint types of slots referenced by the pointer}
\For{$s\in \dom{\MType}$}
    \State $\ptrVar \gets \text{getPtr}(s)$
    \State $(\_, (\_, \tau)) \gets \MType[s]$
    \State $T[\ptrVar \mapsto T(\ptrVar)\cup \qty{\tau}]$
\EndFor
\State $\omega \gets \qty{}$
\For{$s\in \dom{\MType}$}
    \State $\ptrVar \gets \text{getPtr}(s)$
    \State $(\_, (\_, \tau)) \gets \MType[s]$
    \If{$\ptrVar \neq \spVar$ and $\text{cardinal}(T[\ptrVar])=1$}
        \State $\omega[s \mapsto \transPtr]$
    \Else
        \State $\omega[s \mapsto \transOp]$
        \State \textbf{assert} $\text{isLocalStack}(s) \lor \tau\in\qty{0, 1}$
    \EndIf
\EndFor
\end{algorithmic}
\label{alg:trans-strategy}
\end{algorithm}

\begin{algorithm}
\caption{Transform Operand ($\CompilerOp$)}
\begin{algorithmic}
\Require Operand to transform $\opVar$, transformation strategy map $\omega$
\Ensure Transformed operand $\opVar'$
\State $\opVar' \gets \opVar$
\Switch{$\opVar$}
    \Case{$\memOp{r_b}{r_i}{i_s}{i_d}\opMode{s, \tau}$}
        \If{$\omega(s)=\transOp \land \tau \neq 0$}
            \State $\opVar' \gets \memOp{r_b}{r_i}{i_s}{\delta + i_d}\opMode{s+\delta, \tau}$
        \EndIf
    \EndCase
\EndSwitch
\end{algorithmic}
\label{alg:compiler-op}
\end{algorithm}

\begin{algorithm}
\caption{Transform Pointer Arguments for Function Call ($\CompilerPtr$)}
\begin{algorithmic}
\Require State type at the call site $(\Delta_p, \RType_p, \MType_p)$, state type at the beginning of the callee function $(\Delta_c, \RType_c, \MType_c)$, original type substitution $\sigmaCall$ of the function call
\Ensure Instructions to transform pointer arguments $\iota$, transformed type substitution $\sigmaCall'$
\State $\omega_p \gets \text{getTransStrategy}(\MType_p)$
\State $\omega_c \gets \text{getTransStrategy}(\MType_c)$
\State $S_\transPtr \gets \qty{}$
\For{$s_c\in \MType_c$}
    \State Find $s_p\in \MType_p$ such that $\sigmaCall(s_c)\subseteq s_p$
    \State $(\sval_c, (\_, \tau_c)) \gets \MType_c[s_c]$
    \State \textbf{assert} $\text{isSpill}(s_p) \Rightarrow \sval_c=\emptyset$
    \todo{Explain this when explaining spill}
    \todo{Check why do I need this???}
    \State \textbf{assert} $\omega_c(s_c)=\transOp \Rightarrow \omega_p(s_p)=\transOp$
    \shixin{I assert this for the convenience of proof. Maybe I need to explain it somewhere.}
    \If{$\omega_c(s_c)=\transPtr \land \omega_p(s_p) = \transOp \land \sigmaCall(\tau_c) \neq 0$}
        \State \Comment{If the pointer is not transformed in the caller, and the slot is tainted (needs to be relocated), then add $\delta$ to the pointer before passing to the callee}
        \State \textbf{assert} $\neg \text{isNonChangeExp}(\text{getPtr}(s_c))$
        \State $S_\transPtr \gets S_\transPtr \cup \qty{\text{getPtr}(s_c)}$
    \EndIf
\EndFor
\State $\iota \gets []$
\State $\sigmaCall' \gets \sigmaCall$
\For{$\ptrVar \in S_\transPtr$}
    \State Find unique $r$ such that $\RType_c[r]=(\ptrVar, 0)$
    \State $\iota \gets \inst{addq}{\delta, r}; \iota$
    \State $\sigmaCall'[\ptrVar \mapsto \sigmaCall'(\ptrVar)+\delta]$
\EndFor
\State \textbf{assert} $\forall r, \RType_c[r]=(e, \_).\; e\in S_\transPtr \lor \sigmaCall'(e)=\sigmaCall(e)$
\State \textbf{assert} $\forall s, \MType_c[s]=(\_, (e, \_)).\; \sigmaCall'(e)=\sigmaCall(e)$
\State \Comment{Sanity check: each transformed pointer argument is passed through a unique register and is independent from other registers and memory slots. \todo{Consider adding constraints for $\Delta$.}
}
\end{algorithmic}
\label{alg:compiler-ptr}
\end{algorithm}

\subsection{Optimization Pass}
\label{sec:opt-trans-pass}
We then formalize our second transformation pass $\CompilerCallee$, which restores the callee-saved registers' taint after function calls.
Intuitively, $\CompilerCallee$ ensures that public callee-saved registers will not be conservatively marked as tainted after finishing the callee function,
which prevents the hardware defense from introducing extra pipeline delays.

For the ease of implementing and describing this pass, we additionally assume that each function of the input program will use $\instr{pushq}$ to push all callee-saved registers it needs to change to the stack in the beginning, and restore them before returning to the call site.
With this assumption, we can check the first block of the function to find in which slot the function saves the callee-saved registers.
In Algorithm~\ref{alg:compiler-callee}, we define how $\CompilerCallee$ transforms each instruction (i.e., $\text{getSavedMemSlot}(F_p, r)$).
Similarly, we can also use $\CompilerCallee(P)$ to represent the output program generated by applying $\CompilerCallee$ to each instruction.
\todo{Consider moving this to the end of the appendix, sitting after the proof.}

\begin{algorithm}
\caption{Restore Callee-saved Register Taint ($\CompilerCallee$)}
\begin{algorithmic}
\Require Instruction to transform $\instVar$ and its $\pcVar$, original program $P$, and program type $\PType$
\Ensure Transformed instructions $\iota$
\If{$\instVar = \inst{callq}{f_c\opMode{\sigmaCall, \sigmaRet}}$}
\State $F_p \gets \text{getFunc}(P, \pcVar)$
\State $(\Delta_p, \RType_p, \MType_p) \gets \text{getStateType}(\PType, \pcVar)$
\State $\iota_\text{before}, \iota_\text{after} \gets [], []$
\State Let $S$ be the set of callee-saved registers used in $F_p$
\For{$r \in S$}
    \If{$\RType[r] = (\_, 0)$}
        \Comment{$r$ is untainted}
        \State $x \gets \text{getSavedMemSlot}(F_p, r) - \RType_p[r_\reg{rsp}]$
        \State $\iota_\text{before} \gets \inst{movq}{r, x(r_\reg{rsp})}; \iota_\text{before}$
        \State $\iota_\text{after} \gets \inst{movq}{x(r_\reg{rsp}), r}; \iota_\text{after}$
    \EndIf
\EndFor
\State $\iota \gets \iota_\text{before};\instVar;\iota_\text{after}$
\Else
\State $\iota \gets \instVar$
\EndIf
\end{algorithmic}
\label{alg:compiler-callee}
\end{algorithm}

Note that we omit the discussion for proving functional correctness and public noninterference of the program generated by $\CompilerCallee$.
The overall idea is to build the simulation relation between states of running $\Compiler(P)$ and $\CompilerCallee(\Compiler(P))$.

\subsection{Simulation Relation}
\label{sec:sim-rel}
\shixin{Taint vars are also in $\text{getNoChangeVar}(\Delta)$.}
\begin{figure}
    \centering
\begin{mathpar}
\inferrule[Sim-PC]{
\text{getFuncInstSeq}(P, \pcVar) = (f, I)\\
\text{getFuncInstSeq}(\Compiler(P), \pcVar') = (f, \Compiler(I))\\
}{
\hasVal{P, \PType}{\pcVar' \prec \pcVar}
}\and
\inferrule[Sim-Var-Map]{
\dom{\sigma}=\dom{\sigma'}\\
\forall x\in \dom{\sigma}.\; \text{isNonChangeExp}(x) \Rightarrow \sigma'(x)=\sigma(x)
}{
\hasVal{\Delta}{\sigma' \prec \sigma}
}\and
\inferrule[Sim-Val]{
(e = \top \land v'=v) \lor
(e\neq \top \land v'=\sigma'(e) \land v=\sigma (e))
}{
\hasVal{}{((v', t), \sigma') \prec_{e, \tau} ((v, t), \sigma)}
}\and
\inferrule[Sim-Reg]{
\forall r\in \RType.\;
(R'[r], \sigma') \prec_{\RType[r]} (R[r], \sigma)
}{
\hasVal{}{(R', \sigma') \prec_\RType (R, \sigma)}
}\and
\inferrule[Sim-Mem]{
\omega = \text{getTransStrategy}(\MType)\\\\
\forall s, \MType[s]=(\sval, (e, \tau)), x=\text{getPtr}(s).\;\\
((\omega(s)=\transOp \Rightarrow \sigma(s)\subseteq \stackPub \land \\
(\text{isLocalStack}(s) \lor \tau\in \qty{0, 1})
)\land
\deltaOp=\text{getOpShift}(\MType, \omega, s)\land\\
\deltaPtr=\sigma'(x)-\sigma(x)\land \deltaPtr=\text{getPtrShift}(\MType, \omega, \sigma, s)\land\\
(\sigma(\sval)=\emptyset \lor (M'[\sigma(\sval)+\deltaPtr+\deltaOp], \sigma')\prec_{e, \tau} (M[\sigma(\sval)], \sigma)))\\
}{
\hasVal{}{(M', \sigma') \prec_{\MType} (M, \sigma)}
}\and
\inferrule[Simulation-Relation-Helper]{
\hasVal{P, \PType}{\pcVar' \prec \pcVar}\\
(\Delta, \RType, \MType) = \text{getStateType}(\PType, \pcVar)\\
\hasType{P, \PType}{(R, M, (\sigma_f, \sigma_b))}{(\Delta, \RType, \MType)}\\
\hasVal{\Delta}{\sigma_f' \prec \sigma_f}\\
\hasVal{\Delta}{\sigma_b' \prec \sigma_b}\\
\sigma=\sigma_f \cup \sigma_b\\
\sigma'=\sigma_f'\cup \sigma_b'\\
\hasVal{}{(R', \sigma') \prec_\RType (R, \sigma)}\\
\hasVal{}{(M', \sigma') \prec_{\MType} (M, \sigma)}
}{
\hasVal{P, \PType}{(R', M', (\pcVar', (\sigma_f', \sigma_b')))\prec_\text{helper} (R, M, (\pcVar, (\sigma_f, \sigma_b)))}
}\and
\inferrule[Simulation-Relation-Base]{
\hasVal{P, \PType}{(R', M', (\pcVar', (\sigma_f', \sigma_b')))\prec_\text{helper} (R, M, (\pcVar, (\sigma_f, \sigma_b)))}\\
\typeChecked{P, \PType}{(R, M, (\pcVar, (\sigma_f, \sigma_b)))}
}{
\hasVal{P, \PType}{(R', M', (\pcVar', (\sigma_f', \sigma_b')))\prec (R, M, (\pcVar, (\sigma_f, \sigma_b)))}
}\and
\inferrule[Simulation-Relation-Other]{
\hasVal{P, \PType}{(R', M', (\pcVar', (\sigma_f', \sigma_b')))\prec_\text{helper} (R, M, (\pcVar, (\sigma_f, \sigma_b)))}\\
\typeChecked{P, \PType}{(R, M, (\pcVar, (\sigma_f, \sigma_b));\Phi)}\\
\Phi = (\pcVar_p, (\sigmaFp, \sigmaBp));\Phi_0\\
\Phi' = (\pcVar_{p}', (\sigmaFp', \sigmaBp')); \Phi_0'\\
\text{getInst}(\Compiler(P), \pcVar_p')=\inst{callq}{f\opMode{\sigmaCall', \sigmaRet'}}\\
\sigma_f'=(\sigmaFp'\cup \sigmaBp')\circ \sigmaCall'\\
\pcVar_{p0}'=\text{getCallPrefixPc}(\Compiler(P), \pcVar_p')\\
\hasVal{P, \PType}{(R_p', M_p', (\pcVar_{p0}', (\sigmaFp', \sigmaBp')); \Phi_0')\prec (R_p, M_p, (\pcVar_p, (\sigmaFp, \sigmaBp));\Phi_0)}\\
(\Delta, \RType, \MType)=\text{getStateType}(\PType,\pcVar)\\
\omega = \text{getTransStrategy}(P, \pcVar)\\
\sRet=[\sigma_f(\spVar), \sigma_f(\spVar)+8)\\
M[\sRet]=(\text{nextPc}(P, \pcVar_p), 0)\\
M'[\sRet]=(\text{nextPc}(\Compiler(P), \pcVar_p'), 0)\\
\forall x\not\in \text{getDom}(\MType, \sigma_f)\cup \sRet.\; M[x]=M_p[x]\\
\forall x\not\in \text{getShiftedDom}(\MType, \omega, \sigma_f, \sigma_f')\cup \sRet.\; M'[x]=M_p'[x]\\
}{
\hasVal{P, \PType}{(R', M', (\pcVar', (\sigma_f', \sigma_b'));\Phi')\prec (R, M, (\pcVar, (\sigma_f, \sigma_b));\Phi)}
}
\end{mathpar}
    \caption{Simulation Relation}
    \todo{When defining getNonChangeVar, do not forget to add $\spVar$ to it.}
    \label{fig:sim-rel-revision}
\end{figure}

In this section, we define the simulation relation between states of running the original program and the transformed program $\Compiler(P)$ as shown in Figure~\ref{fig:sim-rel-revision}.

At the high level, $\Compiler$ relocates memory objects in memory to separate secret and public data and modifies pointer values accordingly.
To justify functional correctness, we constrain the non-pointer values of registers and memory slots to be the same as those in the original state.
\todo{I am still unsure whether I should mention and explain isNonChangeExp here.}
Another key point of the simulation relation is to build a memory address (slot) map between the original and the transformed memory.
Specifically, we define helper mapping functions used in the simulation relations as follows:
\begin{multline*}
\text{getOpShift}(\MType, \omega, s)\\
=
\begin{cases}
    0 & \omega(s)=\transPtr \lor \MType[s]=(\_, (\_, 0))\\
    \delta & \text{otherwise}\\
\end{cases}
\end{multline*}
\begin{multline*}
\text{getPtrShift}(\MType, \omega, \sigma, s) \\
=
\begin{cases}
    0 & \omega(s)=\transOp \lor \sigma(s)\not\subseteq \stackPub \cup \stackSec \\ 
    \delta & \omega(s)=\transPtr \land \sigma(\MType[s])=(\_, (\_, 1))\\ 
    0, \delta & \omega(s)=\transPtr \land \sigma(\MType[s])=(\_, (\_, 0))\\ 
\end{cases}
\end{multline*}
\revision{
\begin{multline*}
\text{getShift}(\MType, \omega, \sigma, \sigma', s) \\
=
\text{getOpShift}(\MType, \omega, s)+
(\sigma'(\text{getPtr}(s))-\sigma(\text{getPtr}(s))).
\end{multline*}
\begin{equation*}
\text{getShiftedSlot}(\MType, \omega, \sigma, \sigma', s)
= \sigma(s)+ \text{getShift}(\MType, \omega, \sigma, \sigma', s)
\end{equation*}
}
\begin{multline*}
\text{getPossibleSlot}(\MType, \omega, \sigma, \sigma', s)\\
=
\begin{cases}
    \sigma(s)\cup (\sigma(s)+\delta) & \omega(s)=\transOp\\
    \sigma(s)+\sigma'(\text{getPtr}(s))-\sigma(\text{getPtr}(s)) & \omega(s)=\transPtr.
\end{cases}
\end{multline*}
\begin{align*}
    \text{getDom}(\MType, \sigma) & = \bigcup_{s\in\dom{\MType}}\sigma(s)\\
    \text{getShiftedDom}(\MType, \omega, \sigma, \sigma') & =
    \bigcup_{s\in\dom{\MType}}\text{getPossibleSlot}(\MType, \omega, \sigma, \sigma', s)
\end{align*}
\shixin{Note that getPtrShift is not fully deterministic. Since this is used to constrain $\sigma'$ which is constructed by us during proof, so it should be fine.}
\todo{Explain it a bit.}

\todo{
New thing introduced in this version: getNonChangeVar.
I need to explain this in the type system part.
}

\subsection{Functional Correctness}
\label{sec:func-correctness}
With the simulation defined in the last section, we formalize and prove the functional correctness of the transformed program generated by the base pass $\Compiler$ as shown in Theorem~\ref{thm:functional-correctness-revision}.

\begin{theorem}
If $\typeChecked{P, \PType}{S}$ and $\hasVal{P, \PType}{S'\prec S}$, then $S\xrightarrow{\instVar} S_1$, $S'\xrightarrow{\Compiler(\instVar)}^* S_1'$, and $\hasVal{P, \PType}{S_1'\prec S_1}$, or $S$ and $S'$ are termination states.
\label{thm:functional-correctness-revision}
\end{theorem}
\begin{proof}
Denote
\begin{align*}
    S & =(R, M, (\pcVar, (\sigma_f, \sigma_b); \Phi)) & \sigma=\sigma_f\cup \sigma_b\\
    S' & =(R', M', (\pcVar', (\sigma_f', \sigma_b')); \Phi') & \sigma'=\sigma_f'\cup \sigma_b.
\end{align*}
According to the definitions of well-formedness and the simulation relation, $\pcVar$ and $\pcVar'$ are valid PCs in $P$ and $\Compiler(P)$, respectively. Then, there should exist $f$, $\instVar$, $I$ such that $\text{getFuncInstSeq}(P, \pcVar)=(f, \instVar;I)$ and $\text{getFuncInstSeq}(\Compiler(P), \pcVar')=(f, \Compiler(\instVar); \Compiler(I))$.
We also denote 
\begin{equation*}
    (\Delta, \RType, \MType)=\text{getStateType}(\PType, \pcVar) \qquad \omega=\text{getTransStrategy}(\MType).
\end{equation*}
Consider the following cases for $\instVar$:

\textbf{Case $\inst{movq}{r_1, \memOp{r_b}{r_i}{i_s}{i_d}\opMode{\sOp, \tauOp}}$.}
In this case, there should be $\Compiler(\instVar) =\inst{movq}{r_1, \memOp{r_b}{r_i}{i_s}{\delta_0+i_d}\opMode{\sOp+\delta_0, \tauOp}}$
where if $\omega(\sOp)=\transOp\land \tauOp\neq 0$, then $\delta_0=\delta$, else $\delta_0=0$.

First, according to \textsc{Typing-Movq-R-M}, $r_b$ and $r_i$ are untainted and we can denote their type as $\RType[r_b]=(e_b, 0)$, $\RType[r_i]=(e_i, 0)$.
Then, there should be
\begin{align*}
    R[r_b] & = (\sigma(e_b), 0) & R[r_i] & = (\sigma(e_i), 0)\\
    R'[r_b] & = (\sigma'(e_b), 0) & R'[r_i] & = (\sigma'(e_i), 0).
\end{align*}
We further denote $\RType[r_1]=(e_r, \tau_r)$, $R[r_1]=(v_r, t_r)$, and $R'[r_1]=(v_r', t_r')$, where $v_r=\sigma(e_r)$ and $v_r'=\sigma'(e_r)$. According to the simulation relation between $S$ and $S'$, there should be $t_r'=t_r$.
According to both \textsc{Typing-StoreOp-Spill} and \textsc{Typing-StoreOp-Non-Spill}, there should always be $\hasVal{\Delta}{\tau_r\Rightarrow \tauOp}$, so $t_r\Rightarrow \sigma(\tauOp)$.
Therefore, both $S$ and $S'$ satisfy constraints in \textsc{Dyn-Movq-r-m} and can execute the next store instruction.

Second, denote $e_a=i_d+e_b+e_i\times i_s$, $v_a=\sigma(e_a)$, and $v_a'=\sigma'(e_a)+\delta_0$.
We construct the next states $S_1$ and $S_1'$ as follows:
\begin{align*}
    \pcVar_1 & =\text{nextPc}(P, \pcVar) & \pcVar_1' & =\text{nextPc}(\Compiler(P), \pcVar')\\
    M_1 & = M[[v_a, v_a+8) \mapsto R[r_1]] & 
    M_1' & = M'[[v_a', v_a'+8) \mapsto R'[r_1]]\\
    \Phi_1 & = (\pcVar_1, (\sigma_f, \sigma_b)) &
    \Phi_1' & = (\pcVar_1', (\sigma_f', \sigma_b'))\\
    S_1 & = (R, M_1, \Phi_1) & 
    S_1' & = (R', M_1', \Phi_1').
\end{align*}
According to Theorem~\ref{thm:type-safety-revision}, $\typeChecked{P, \PType}{S_1}$.

Third, we prove that $\hasVal{P, \PType}{S_1'\prec S_1}$.
We begin with proving $\hasVal{P, \PType}{S_1'\prec_\text{helper}S_1}$.
According to \textsc{Typing-Movq-r-m}, there should exists $\sval_o$, $e_o$, and $\tau_o$ such that
\begin{equation*}
    \hasType{\Delta, \RType, \MType}{\text{store}(\memOp{r_b}{r_i}{i_s}{i_d}\opMode{\sOp, \tauOp}, 8, \RType[r_1])}{(\sval_o, (e_o, \tau_o))},
\end{equation*}
and for $\MType_1\coloneq \MType[\sOp \mapsto (\sval_o, (e_o, \tau_o))]$, we have $(\Delta, \RType, \MType_1)=\text{getStateType}(\PType, \pcVar_1)$.
Then, according to Lemma~\ref{lemma:same-trans-strategy-mtype}, there should be $\text{getTransStrategy}(\MType_1)=\text{getTransStrategy}(\MType)=\omega$.
Furthermore, both \textsc{Typing-StoreOp-Spill} and \textsc{Typing-StoreOp-Non-Spill} implies that $\tauOp=\tau_o$.
This implies that $\delta_0=\text{getOpShift}(\MType_1, \omega, \sOp)$.
According to \textsc{Simulation-Relation-Helper}, as all other requirements can be easily derived from $\hasVal{P, \PType}{S'\prec S}$, we just need to prove $\hasVal{}{(M_1', \sigma')\prec_{\MType_1}(M_1, \sigma)}$.
\todo{Fix here!}

Denote $\ptrVar=\text{getPtr}(\sOp)$, then according to the store typing rules, $\hasVal{\Delta}{\text{isNonChangeExp}(e_a-\ptrVar)}$.
Thus, there should be $\sigma'(e_a-\ptrVar)=\sigma(e_a-\ptrVar)$. Denote $\delta_p=\sigma'(\ptrVar)-\sigma(\ptrVar)$, then we have $v_a'=v_a+\delta_p+\delta_0$.
Furthermore, the store typing rules also imply that $\hasVal{\Delta}{[e_a, e_a+8)\subseteq \sOp}$, so $[v_a, v_a+8)\subseteq \sigma(\sOp)$.
Note that
\begin{align*}
    & \text{getShiftedSlot}(\MType_1, \omega, \sigma, \sigma', \sOp) \\
    =\; &  \sigma(\sOp)+\text{getOpShift}(\MType_1, \omega, \sOp) + \sigma'(\sOp)-\sigma(\sOp)\\
    =\; & \sigma(\sOp)+\delta_0+\delta_p,
\end{align*}
so $[v_a', v_a'+8) \subseteq \text{getShiftedSlot}(\MType_1, \omega, \sigma, \sigma', \sOp)$.

For each $s\in \dom{\MType_1}$, denote
\begin{align*}
    \MType_1[s] & =(\sval_1, (e_1, \tau_1)) & x & =\text{getPtr}(s)\\
    \deltaOp[1] & = \text{getOpShift}(\MType_1, \omega, s) & \deltaPtr & = \sigma'(x)-\sigma(x).
\end{align*}
According to \textsc{Sim-Mem}, we need to prove the following statements:
\begin{align*}
    & \omega(s)=\transOp \Rightarrow \sigma(s)\subseteq \stackPub \land (\text{isLocalStack}(s)\lor \tau_1\in \qty{0, 1})\\
    & \deltaPtr = \text{getPtrShift}(\MType_1, \omega, \sigma, s)\\
    & \sval_1=\emptyset \lor (M'[\sigma(\sval_1)+\deltaPtr+\deltaOp[1]], \sigma')\prec_{e_1, \tau_1}(M[\sigma(\sval_1)], \sigma).
\end{align*}

If $s\neq \sOp$, then $\MType_1[s]=\MType[s]$ and $\hasVal{\Delta}{s\cap \sOp}=\emptyset$.
According to Lemma~\ref{lemma:shifted-slot-non-overlap}, there should be
\begin{equation*}
\text{getShiftedSlot}(\MType_1, \omega, \sigma, \sigma', s)\cap \text{getShiftedSlot}(\MType_1, \omega, \sigma, \sigma', \sOp) =\emptyset
\end{equation*}
Note that $(\sigma(\sval)+\deltaPtr+\deltaOp[1])\subseteq \text{getShiftedSlot}(\MType_1, \omega, \sigma, \sigma', s)$.
Hence, we have
\begin{align*}
    M_1[\sigma(\sval)] & = M[\sigma(\sval)]\\
    M_1'[\sigma(\sval)+\deltaPtr+\deltaOp[1]] & = M[\sigma(\sval)+\deltaPtr+\deltaOp[1]].
\end{align*}
We can then prove that the statement is true by using the above relation and the assumption $\hasVal{}{(M', \sigma')\prec_{\MType}(M, \sigma)}$ (derived from $\hasVal{P, \PType}{S'\prec S}$).

If $s=\sOp$, we have
\begin{equation*}
    \MType_1[s]=\MType_1[\sOp]=(\sval_o, (e_o, \tauOp))=(\sval_1, (e_1. \tau_1)).
\end{equation*}
The statement can be derived by discussing whether $\text{isSpill}(\sOp)$ and whether $[v_a, v_a+8)\subseteq \sval_o$.
We omit details here since the reasoning is relatively straightforward.
\todo{Finish this.}

If $f=\mathit{main}$, then $\hasVal{P, \PType}{S_1'\prec S_1}$ according to \textsc{Simulation-Relation-Base}.

We consider the case where $f\neq \mathit{main}$.
The key idea is to prove that writing to $[v_a, v_a+8)$ (or $[v_a', v_a'+8)$) does not affect memory content not in $\text{getDom}(\MType, \sigma_f)$ (or $\text{getShiftedDom}(\MType, \omega, \sigma_f, \sigma_f')$).

There should exist $\pcVar_p$, $\sigmaFp$, $\sigmaBp$, $\Phi_0$, $\pcVar_p'$, $\sigmaFp'$, $\sigmaBp'$, and $\Phi_0'$ such that
\begin{equation*}
    \Phi = (\pcVar_p, (\sigmaFp, \sigmaBp));\Phi_0 \qquad
    \Phi' = (\pcVar_{p}', (\sigmaFp', \sigmaBp')); \Phi_0'.
\end{equation*}
Furthermore, $\hasVal{P, \PType}{S'\prec S}$ implies that there exists $R_p$, $M_p$, $R_p'$, $M_p'$ such that $\hasVal{P, \PType}{(R_p', M_p', \Phi')\prec (R_p, M_p, \Phi)}$.
According to Lemma~\ref{lemma:shifted-slot}, we have
\begin{align*}
    [v_a, v_a+8) & \subseteq \sigmaOp(\sOp) \subseteq \text{getDom}(\MType_1, \sigma_f)\\
    [v_a', v_a'+8) & \subseteq \text{getShiftedSlot}(\MType_1, \omega, \sigma, \sigma', \sOp) \\
    & \subseteq \text{getShiftedDom}(\MType_1, \omega, \sigma_f, \sigma_f')
\end{align*}
From Lemma~\ref{lemma:no-ret-slot}, we can also derive that $\sRet\cap \text{getDom}(\MType_1, \sigma_f)=\emptyset$, and $\sRet\cap \text{getShiftedDom}(\MType_1, \omega, \sigma_f, \sigma_f')=\emptyset$.
\todo{This needs to be further justified.}
Thus, there should be 
\begin{align*}
    & M_1[\sRet] = M[\sRet] = (\text{nextPc}(P, \pcVar_p), 0)\\
    & M_1'[\sRet] = M'[\sRet] = (\text{nextPc}(\Compiler(P), \pcVar_p'), 0)\\
    & \forall x \not \in \text{getDom}(\MType, \sigma_f) \cup \sRet.\; M_1[x]=M[x]=M_p[x]\\
    & \forall x \not \in \text{getShiftedDom}(\MType, \omega, \sigma_f, \sigma_f') \cup \sRet.\; M_1'[x]=M'[x]=M_p'[x].
\end{align*}
Therefore, the statement also holds for the case when $f\neq \mathit{main}$.

\todo{\textbf{Case $\inst{movq}{\memOp{r_b}{r_i}{i_s}{i_d}\opMode{\sOp, \tauOp}, r_0}$.}}

\textbf{Case $\inst{jne}{\ell\opMode{\sigmaOp}}$}
In this case, $\Compiler(\instVar)=\instVar$.
First, according to \textsc{Typing-Jne}, flag $\reg{ZF}$ is untainted and we can denote its type as $\RType[\reg{ZF}]=(e=0, 0)$, where $\hasVal{\Delta}{\text{isNonChangeExp}(e)}$.
Then, according to well-formedness of $S$ and simulation relation between $S'$ and $S$, there should be $R[\reg{ZF}]=(v, 0)$, $R'[\reg{ZF}]=(v', 0)$, and $\hasVal{}{((v', 0), \sigma')\prec_{e=0, 0}((v, 0), \sigma)}$.
So both $S$ and $S'$ satisfy constraints in \textsc{Dyn-Jne} and can execute the next instruction.
According to \textsc{Sim-Val}, it is straightforward to derive $v=v'$ by discussing on whether $e=\top$ and applying $\hasVal{\Delta}{\text{isNonChangeExp}(e)}$.
Thus, $\instVar$ and $\Compiler(\instVar)$ have the same branch direction (both taken or not taken).

Second, we construct the next states $S_1$ and $S_1'$ as follows:
\begin{align*}
    \pcVar_n & =\text{nextPc}(P, \pcVar) &
    \pcVar_t & =\text{getBlockPc}(P, \ell) \\
    \pcVar_n' & =\text{nextPc}(\Compiler(P), \pcVar') &
    \pcVar_t'& =\text{getBlockPc}(\Compiler(P), \ell).
\end{align*}
\begin{align*}
S_1 & =
\begin{cases}
    (R, M, (\pcVar_n, (\sigma_f, \sigma_b); \Phi) & v = \mathit{true}\\
    (R, M, (\pcVar_t, (\sigma_f, \sigma\circ \sigmaOp); \Phi) & v = \mathit{false}.
\end{cases}\\
S_1' & =
\begin{cases}
    (R', M', (\pcVar_n', (\sigma_f', \sigma_b'); \Phi') & v = \mathit{true}\\
    (R', M', (\pcVar_t', (\sigma_f', \sigma'\circ \sigmaOp); \Phi') & v = \mathit{false}.
\end{cases}
\end{align*}
According to Theorem~\ref{thm:type-safety-revision}, there should be $\typeChecked{P, \PType}{S_1}$.

Third, we prove that $\hasVal{P, \PType}{S_1'\prec S_1}$ under both cases (when the branch is taken and not taken).
Denote $(\Delta, \RType, \MType)=\text{getStateType}(\PType, \pcVar)$.
When $v=\mathit{true}$, i.e., the branch is not taken, all fields in $S_1$ (or $S_1'$) are the same as $S$ (or $S'$) except for its current PC.
The next state type is $\text{getStateType}(\PType, \pcVar_n)=(\Delta\cup \qty{e= 0}, \RType, \MType)$ according to \textsc{Typing-Jne}, where the extra 
type context constraint is satisfied in both machines, i.e., $\sigma(e=0)=\sigma'(e=0)=\mathit{true}$.
Furthermore, since $\Compiler(\instVar)=\instVar$, there should be $\hasVal{P, \PType}{\pcVar_n'\prec \pcVar_n}$.
Therefore, we can easily derive $\hasVal{P, \PType}{S_1'\prec S_1}$ by unfolding $\hasVal{P, \PType}{S'\prec S}$.

When $v=\mathit{false}$, i.e., the branch is taken, according to \textsc{Type-Jne}, the next state type is $(\Delta_1, \RType_1, \MType_1)\coloneq \PType(f)(\ell)$, where $\dom{\MType_1}=\dom{\MType}$, and $(\Delta\cup \qty{e\neq 0}, \RType, \MType)\sqsubseteq \sigmaOp(\Delta_1, \RType_1, \MType_1)$.
According to Lemma~\ref{lemma:same-trans-strategy-mtype},  $\text{getTransStrategy}(\MType_1)=\text{getTransStrategy}(\MType)=\omega$.

We focus on proving $\hasVal{P, \PType}{S_1'\prec_\text{helper} S_1}$, as the other requirements of $\hasVal{P, \PType}{S_1'\prec S_1}$ can be easily derived by unfolding $\hasVal{P, \PType}{S\prec S}$ and applying what we have illustrated above.
Denote 
\begin{equation*}
    \sigma_1=\sigma_f \cup (\sigma\circ \sigmaOp) \qquad
    \sigma_1'=\sigma_f'\cup (\sigma'\circ \sigmaOp).
\end{equation*}
According to \textsc{Simulation-Relation-Helper}, we just need to prove the following requirements:
\begin{itemize}
    \item $\hasVal{P, \PType}{\pcVar_t'\prec \pcVar_t}$: this hold since both $\pcVar_t$ and $\pcVar_t'$ are PCs for block $\ell$ in $P$ and $\Compiler(P)$.
    \item $\hasType{P, \PType}{(R, M, (\sigma_f, \sigma\circ \sigmaOp))}{(\Delta_1, \RType_1, \MType_1)}$: this is implied by $\typeChecked{P, \PType}{S_1}$.
    \item $\hasVal{\Delta_1}{\sigma_f'\prec \sigma_f}$: this is implied by $\hasVal{P, \PType}{S'\prec S}$.
    \item $\hasVal{\Delta_1}{\sigma'\circ \sigmaOp \prec \sigma \circ \sigmaOp}$:
    in this case, $\dom{\sigma'\circ\sigmaOp}=\dom{\sigma'\circ\sigmaOp}=\dom{\sigmaOp}$.
    According to \textsc{State-Subtype}, we have $\hasVal{\Delta}{\sigmaOp(\Delta_1)}$.
    Hence, for all $x\in \dom{\sigmaOp}$, if $\hasVal{\Delta_1}{\text{isNonChangeOp}(x)}$, then $\hasVal{\Delta}{\text{isNonChangeOp}(\sigmaOp(x))}$.
    Furthermore, according to the definition, it is easy to derive $\hasVal{\Delta}{\sigma'\prec \sigma}$ from $\hasVal{\Delta}{\sigma_f'\prec \sigma_f}$ and $\hasVal{\Delta}{\sigma_b'\prec \sigma_b}$.
    Thus, $\sigma'(\sigmaOp(x))=\sigma(\sigmaOp(x))$, and the original statement holds.
    \item $\hasVal{}{(R', \sigma_1')\prec_{\RType_1} (R, \sigma_1)}$:
    for all $r\in \RType$, according to \textsc{Sim-Reg}, there should be $(R'[r], \sigma')\prec_{\RType[r]}(R[r], \sigma)$.
    According to \textsc{Reg-Subtype}, we have $\hasVal{}{\RType[r]\sqsubseteq \sigmaOp(\RType_1[r])}$.
    Applying Lemma~\ref{lemma:sim-val-substitute-special}, we can get $\hasVal{}{(R'[r], \sigma'\circ\sigmaOp)\prec_{\RType_1[r]}(R[r], \sigma\circ \sigmaOp)}$.
    Note that $\sigma_1=\sigma\circ \sigmaOp$ and $\sigma_1'=\sigma'\circ \sigmaOp$ according to Lemma~\ref{lemma:br-substitute}. Then, there should be $\hasVal{}{(R'[r], \sigma_1')\prec_{\RType_1[r]}(R[r], \sigma_1)}$.
    Thus, the statement is true.
    \item $\hasVal{}{(M', \sigma_1')\prec_{\MType_1} (M, \sigma_1)}$:
    \todo{Check whether I can apply Lemma~\ref{lemma:sim-mem-substitute-general} here to simplify the proof (and avoid missing too many details).}
    for each $s_1\in \dom{\MType_1}$, denote $\MType_1[s_1]=(\sval_1, (e_1, \tau_1))$ and $x_1=\text{getPtr}(s_1)$.
    According to \textsc{Typing-Jne}, $\text{getTaintVar}(\dom{\sigmaOp})=\emptyset$, so $\sigmaOp(\tau_1)=\tau_1$.
    According to \textsc{State-Subtype} and \textsc{Mem-Slot-Subtype}, there must exists $s\in \MType_1$ and $\MType[s]=(\sval, (e, \tau))$ such that
    \begin{align*}
        & \hasVal{\Delta}{\sigmaOp(s_1)\subseteq s} \quad \hasVal{\Delta}{\sigmaOp(\sval_1)\subseteq \sval}\\
        & \hasVal{\Delta}{\text{isSpill}(\sigmaOp(\sval_1))\Rightarrow \text{isSpill}(\sval)}\\
        & \hasVal{\Delta}{e=\sigmaOp(e_1)\lor (\text{isNonChangeExp}(e)\land \sigmaOp(e_1)=\top)}\\
        & \qquad \lor \sigmaOp(\sval_1)=\emptyset\\
        & \hasVal{\Delta}{\tau=\tau_1\lor (\text{isSpill}(\sval)\land \sigmaOp(\sval_1)=\emptyset)}.
    \end{align*}
    Note that $\dom{\MType_1}=\dom{\MType}$, so there should be $s_1=s$ and $x_1=x$.
    \textsc{Reg-Mem-Type} implies that $\text{getVars}(s)=\sigma_f(s)$, 
    so $\sigma(s)=\sigma_f(s)=\sigma_1(s)$ and $\sigma(x)=\sigma_f(x)=\sigma_1(x)$.
    Denote $\deltaOp=\text{getOpShift}(\MType_1, \omega, s)$ and $\deltaPtr=\sigma_1'(x)-\sigma_1(x)$.
    Hence, we just need to prove the following statements:
    \begin{itemize}
        \item $\omega(s)=\transOp \Rightarrow \sigma_1(s)\subseteq \stackPub \land (\text{isLocalStack}(s) \lor \tau_1\in\qty{0, 1})$:
        Suppose $\omega(s)=\transOp$. Then $\sigma(s)\subseteq \stackPub$. \textsc{Reg-Mem-Type} implies that $\text{getVars}(s)=\sigma_f(s)$, so $\sigma(s)=\sigma_f(s)=\sigma_1(s)$.
        Thus, $\sigma_1(s)\subseteq \stackPub$.
        
        Note that if $\neg \text{isLocalStack}(s)$, then $\tau\in \qty{0, 1}$ and $\neg \text{isSpill}(s)$, which is equivalent to $\neg \text{isSpill}(\sval)$. 
        This together implies that $\tau=\tau_1$, so the statement is true.
        \item $\deltaPtr=\text{getPtrShift}(\MType_1, \omega, \sigma_1, s)$:
        Note that $\deltaPtr=\sigma_1'(x)-\sigma_1(x)=\sigma'(x)-\sigma(x)=\text{getPtrShift}(\MType, \omega, \sigma, s)$.
        Then, this can be proved by discussing whether $\omega(s)=\transPtr$. The key point is to note that when $\omega(s)=\transPtr$, there should be $\neg \text{isSpill}(s)$ and thereby $\tau=\tau_1$.
        \item $\sigma_1(\sval_1)=\emptyset$ or 
        $(M'[\sigma_1(\sval_1)+\deltaPtr+\deltaOp], \sigma_1')\prec_{e_1, \tau_1}(M[\sigma_1(\sval_1)], \sigma_1)$:
        We just need to prove that the latter statement holds when $\sigma_1(\sval_1)\neq\emptyset$.
        This implies that $\tau=\tau_1$, so we have
        \begin{equation*}
            \deltaOp =\text{getOpShift}(\MType_1, \omega, s)=\text{getOpShift}(\MType, \omega, s).
        \end{equation*}
        Then, from $\hasVal{}{(M', \sigma')\prec_{\MType}(M, \sigma)}$, we can derive that 
        \begin{equation*}
            (M'[\sigma(\sval)+\deltaPtr+\deltaOp], \sigma')\prec_{e, \tau}(M[\sigma(\sval)], \sigma).
        \end{equation*}
        We can then prove the original statement by discussing the format of $e$ and applying Lemma~\ref{lemma:sim-val-substitute-special}, similar to the process of proving the simulation relation for registers.
        \shixin{I omit many details here.}
    \end{itemize}
\end{itemize}

\textbf{Case $\inst{callq}{f_c\opMode{\sigmaCall, \sigmaRet}}$.} 
\revision{
Denote
\begin{align*}
    & \pcVar_c=\text{getBlockPc}(P, f_c)\\
    & (\Delta_c, \RType_c, \MType_c)=\text{getStateType}(\PType, \pcVar_c)\\
    & \omega_c=\text{getTransStrategy}(\MType_c).
\end{align*}
According to the definition of $\Compiler$, there should be
\begin{align*}
    & \Compiler(\instVar)=\iota_1; \iota_2;\dots;\iota_n;\inst{callq}{f_c\opMode{\sigmaCall', \sigmaRet}}\\
    & ((\iota_1; \iota_2;\dots;\iota_n), \sigmaCall')=\CompilerPtr((\Delta, \RType, \MType), \sigmaCall(\Delta_c, \RType_c, \MType_c), \sigmaCall),
\end{align*}
and for $i=1,\dots, n$,
\begin{align*}
    & \iota_i=\inst{addq}{\delta, r_i} && r_i\neq r_\reg{rsp}\\
    & \RType_c[r_i]=(\ptrVar_i, 0) &&
    \sigmaCall'(\ptrVar_i)=\sigmaCall(\ptrVar_i)+\delta
\end{align*}
}

First, according to \textsc{Typing-Callq}, there exists a constant $c$ such that $\RType[r_\reg{rsp}]=(\spVar+c, 0)$.
Then, according \textsc{Reg-Type}, \textsc{Value-Type}, and well-formedness of $S$, there should be $R[r_\reg{rsp}]=(v, 0)$ where $v=\sigma(\spVar)+c$. Thus, $S$ satisfies constraints in \textsc{Dyn-Callq} and is allowed to execute $\instVar$.
Furthermore, according to the simulation relation and the fact that $\spVar$ is not affected by our transformation, i.e., $\text{isNonChangeExp}(\spVar)$, there should also be $R'[r_\reg{rsp}]=(v,0)$. Hence, $S'$ can execute all instructions in $\Compiler(\instVar)$ without getting stuck.

Second, we construct the next states $S_1$ and $S_1'$ as follows:
\begin{align*}
    R_1 & = R[r_\reg{rsp}\mapsto (v-8, 0)]\\
    M_1 & = M[[v-8, v) \mapsto (\text{nextPc}(P, \pcVar), 0)]\\
    \Phi_1 & = (\pcVar_c, (\sigma\circ \sigmaCall, []\rightarrow[]));(\pcVar, (\sigma_f, \sigma_b)); \Phi\\
    \pcVar_c' & = \text{getBlockPc}(\Compiler(P), f_c) \quad \pcVar'' = \text{getInstPc}(\Compiler(P), \inst{callq}{f_c\opMode{\sigmaCall', \sigmaRet}})\\
    R_1' & = R'[r_\reg{rsp}\mapsto (v-8, 0), r_1\mapsto (v_1+\delta, 0), \dots, r_n\mapsto (v_n+\delta, 0)]\\
    M_1' & = M'[[v-8, v) \mapsto (\text{nextPc}(\Compiler(P), \pcVar''), 0)]\\
    \Phi_1' & = (\pcVar_c', (\sigma'\circ\sigmaCall', []\rightarrow[]));(\pcVar', (\sigma_f', \sigma_b')); \Phi'\\
    S_1 & = (R_1, M_1, \Phi_1) \qquad S_1'=(R_1', M_1', \Phi_1').
\end{align*}

Third, we prove that $\hasVal{P, \PType}{S_1'\prec S_1}$ by proving all requirements in \textsc{Simulation-Relation-Other} are satisfied.
We start with proving $\hasVal{P, \PType}{S_1'\prec_\text{helper}S_1}$.
According to \textsc{Simulation-Relation-Helper}, as the other requirements are straightforward, we only prove the following statements:
\begin{itemize}
    \item $\hasVal{\Delta_c}{\sigma_c'\prec\sigma_c}$ where $\sigma_c'=\sigma'\circ\sigmaCall'$ and $\sigma_c=\sigma\circ\sigmaCall$:
    for all $x\in \dom{\sigma_c}$ and $\hasVal{\Delta_c}{\text{isNonChangeExp}(x)}$, there must be $x\neq \ptrVar_i$ for $i=1, \dots, n$.
    Furthemore, according to \textsc{State-Subtype} and the subtype relation at call, there should be $\hasVal{\Delta}{\sigmaCall(\Delta_c)}$, which implies that $\hasVal{\Delta}{\text{isNonChangeExp}(\sigmaCall(x))}$.
    Thus, we have
    \begin{equation*}
        \sigma_c'(x)=\sigma'(\sigmaCall'(x)) = \sigma'(\sigmaCall(x))=\sigma(\sigmaCall(x))=\sigma_c(x),
    \end{equation*}
    so the original statement holds.
    \item $\hasVal{}{(R_1', \sigma_c')\prec_{\RType_c}(R_1, \sigma_c)}$: this can be proved similarly to the case for $\instr{jne}$, where the major difference is that we will apply Lemma~\ref{lemma:sim-val-substitute-general} instead of Lemma~\ref{lemma:sim-val-substitute-special}.
    \todo{Consider to complete the reasoning here since it will not be too long.}
    \item $\hasVal{}{(M_1', \sigma_c')\prec_{\MType_c}(M_1, \sigma_c)}$: 
    this can be proved by applying Lemma~\ref{lemma:sim-mem-substitute-general}.
    \todo{Double check whether all assumptions in the lemma are satisfied.}
\end{itemize}

\textbf{Case $\inst{retq}{}$.}
In this case, there should be $\Compiler(\instVar)=\inst{retq}{}$.
According to our assumption on the program, $f$ must not be the top-level function, i.e., $f \neq \mathit{main}$. 
Then, we can denote 
\begin{equation*}
    \Phi = (\pcVar_p, (\sigmaFp, \sigmaBp));\Phi_0 \quad
    \Phi' = (\pcVar_{p}', (\sigmaFp', \sigmaBp')); \Phi_0'.
\end{equation*}
According to \textsc{Typing-Func}, the state type of $\pcVar$ is $(\Delta, \RType, \MType)\coloneq \text{getStateType}(f)(f_\text{ret})$.
It also implies that $\RType[r_\reg{rsp}]=(\spVar, 0)$.
Then, according to \textsc{Reg-Type}, \textsc{Value-Type}, well-formedness of $S$, and simulation relation between $S'$ and $S$, there should be $R[r_\reg{rsp}]=R'[r_\reg{rsp}]=(v, 0)$, where $v=\sigma(\spVar)=\sigma_f(\spVar)$.
Furthermore, according to $\hasVal{P, \PType}{S'\prec S}$ and \textsc{Simulation-Relation-Other}, there should exist $R_p$, $M_p$, $R_p'$, $M_p'$ such that
\begin{align*}
    & \sRet =[\sigma_f(\spVar), \sigma(\spVar)+8)=[v, v+8)\\
    & M[\sRet] = (\text{nextPc}(P, \pcVar_p), 0) \quad
    M'[\sRet] =(\text{nextPc}(\Compiler(P), \pcVar_p'), 0) \\
    & \text{getInst}(P, \pcVar_p)=\inst{callq}{f\opMode{\sigmaCall,\sigmaRet}} \quad
    \text{getInst}(\Compiler(P), \pcVar_p')=\inst{callq}{f\opMode{\sigmaCall',\sigmaRet}} \\
    & \pcVar_{p0}'=\text{getCallPrefixPc}(\Compiler(P), \pcVar_p')\\
    & \forall x\not\in \text{getDom}(\MType, \sigma_f)\cup \sRet.\; M[x]=M_p[x]\\
     & \forall x\not\in \text{getShiftedDom}(\MType, \omega, \sigma_f, \sigma_f')\cup \sRet.\; M'[x]=M_p'[x]\\
    & \sigma_f=(\sigmaFp\cup \sigmaBp)\circ \sigmaCall \quad
    \sigma_f'=(\sigmaFp'\cup \sigmaBp')\circ \sigmaCall'\\
    & \hasVal{P, \PType}{(R_p', M_p', (\pcVar_{p0}', (\sigmaFp', \sigmaBp')); \Phi_0')\prec (R_p, M_p, \Phi)}.
\end{align*}
Thus, $S$ and $S'$ satisfy all constraints in \textsc{Dyn-Retq} and can execute $\instr{retq}$.

Second, we construct the next states $S_1$ and $S_1'$ as follows:
\begin{align*}
    R_1 & = R[r_\reg{rsp}\mapsto (v+8, 0)] \quad
    R_1' = R'[r_\reg{rsp}\mapsto (v+8, 0)]\\
    \Phi_1 & = (\text{nextPc}(P, \pcVar_p), (\sigmaFp, (\sigma_b \circ \sigmaRet^{-1})\cup \sigmaBp);\Phi_0\\
    \Phi_1' & = (\text{nextPc}(\Compiler(P), \pcVar_p'), (\sigmaFp', (\sigma_b' \circ \sigmaRet^{-1})\cup \sigmaBp');\Phi_0\\
    S_1 & = (R_1, M, \Phi_1) \quad\quad  S_1'=(R_1', M', \Phi_1').
\end{align*}

Third, we prove that $\hasVal{P, \PType}{S_1'\prec S_1}$.
As other requirements are relatively straightforward to prove, we only illustrate how to prove $\hasVal{P, \PType}{S_1'\prec_\text{helper}S_1}$.
Denote 
\begin{align*}
    (\Delta_0, \RType_0, \MType_0) & = \text{getStateType}(\PType, \pcVar_p)\\
    (\Delta_1, \RType_1, \MType_1) & = (\sigmaCall\cup \sigmaRet)(\Delta, \RType, \MType)\\
    (\Delta_2, \RType_2, \MType_2) & =
    (\Delta_0\cup \Delta_1, \RType_1[r_\reg{rsp}\mapsto \RType_0[r_\reg{rsp}]], \text{updateMem}(\MType_0, \MType_1))\\
    \sigma_1 & = \sigmaFp\cup (\sigma_b\circ \sigmaRet^{-1})\cup \sigmaBp \\
    \sigma_1' & = \sigmaFp'\cup (\sigma_b'\circ \sigmaRet^{-1})\cup \sigmaBp'.
\end{align*}

According to \textsc{Simulation-Relation-Helper}, this holds as long as the following statements hold:
\begin{itemize}
    \item $\hasVal{P, \PType}{\text{nextPc}(\Compiler(P), \pcVar_p') \prec \text{nextPc}(P, \pcVar_p)}$: $\hasVal{P, \PType}{S'\prec S}$ implies that $\hasVal{P, \PType}{\pcVar_{p0}'\prec \pcVar_p}$.
    According to the definition of $\pcVar_{p0}'$, the statement should hold since both PCs are referring to the next instruction after two matched calls.
    \item $(\Delta_2, \RType_2, \MType_2) =\text{getStateType}(\PType, \text{nextPc}(P, \pcVar_p))$ and $\hasType{P, \PType}{(R_1, M, \sigma_1)}{(\Delta_2, \RType_2, \MType_2)}$: this is implied from Theorem~\ref{thm:type-safety-revision}.
    \item $\hasVal{\Delta_2}{\sigmaFp'\prec \sigmaFp}$ and
    $\hasVal{\Delta_2}{((\sigma_b'\circ \sigmaRet^{-1})\cup \sigmaBp')\prec ((\sigma_b\circ \sigmaRet^{-1})\cup \sigmaBp)}$:
    The first one holds since $\dom{\sigmaFp}\subseteq\text{getVars}(\Delta_0)$ and $\hasVal{\Delta}{\sigmaFp\prec \sigmaFp'}$ implied from $\hasVal{P, \PType}{S'\prec S}$.
    The second one can be proved by discussing about whether the variables belong to $\dom{\sigmaBp}$ or $\dom{\sigmaRet^{-1}}$.
    \item $\hasVal{}{(R_1', \sigma_1')\prec_{\RType_2} (R_1, \sigma_1)}$:
    We need to prove that for each $r\in \RType_2$, $(R_1'[r], \sigma_1')\prec_{\RType_2[r]}(R_1[r], \sigma_1)$.
    We consider the following three cases and briefly explain the idea to prove the statement:
    \begin{itemize}
        \item $r=r_\reg{rsp}$: according to \textsc{Typing-Callq},
        there should exist $c$ such that $\RType_2[r_\reg{rsp}]=\RType_0[r_\reg{rsp}]=(\spVar+c,0)$.
        Then, from \textsc{Typing-Callq}, there should be $\sigmaCall(\spVar)=\spVar+c-8$.
        The relation can be proved by applying this relation and the assumption that $\text{isNonChangeExp}(\spVar)$.
        \item $r\neq r_\reg{rsp}$, and $r$ is a callee-saved register: 
        \textsc{Typing-Func} implies that in the callee function's type context, the dependent type of $r$ is unchanged before and after the function call and is a non-top type variable. This implies that the value of $r$ is also unchanged before and after executing $f$.
        We can then prove the statement by applying the above relation in the caller's type context.
        \item $r\neq r_\reg{rsp}$, and $r$ is not a callee-saved register:
        \textsc{Typing-Func} implies that $\text{isNonChangeExp}(\RType[r])$. We can then prove the statement by discussing whether the dependent type of $r$ is $\top$ or not.
        \todo{Finish this!}
    \end{itemize}
    \item $\hasVal{}{(M', \sigma_1')\prec_{\MType_2} (M, \sigma_1)}$:
    According to Lemma~\ref{lemma:same-trans-strategy-mtype},  
    \begin{equation*}
        \omega_0 \coloneqq \text{getTransStrategy}(\MType_2)=\text{getTransStrategy}(\MType),
    \end{equation*}
    so all assertions in $\text{getTransStrategy}$ are also satisfied by $\MType_2$.

    For each $s\in \dom{\MType_2}$, denote
    \begin{align*}
        & \MType_2[s]=(\sval_2, (e_2, \tau_2)) \qquad \MType_0[s]=(\sval_0, (e_0, \tau_0))\\
        & \delta_s = \text{getShift}(\MType_2, \omega_0, \sigma_1, \sigma_1', s).
    \end{align*}
    According to the assertions in $\text{updateMem}$, $\tau_2=\tau_0$, so
    \begin{equation*}
        \delta_s = \text{getShift}(\MType_2, \omega_0, \sigma_1, \sigma_1', s)=\text{getShift}(\MType_0, \omega_0, \sigma_1, \sigma_1', s).
    \end{equation*}

    According to Sim-Mem, as the other requirements of the simulation relation can be easily proved by unfolding $\hasVal{P, \PType}{(R_p', M_p', (\pcVar_{p0}', (\sigmaFp', \sigmaBp')); \Phi_0')\prec (R_p, M_p, \Phi)}$, we focus on proving that if $\sigma_1(\sval_2)\neq \emptyset$, then
    \begin{equation}
      (M'[\sigma_1(\sval_2)+\delta_s], \sigma_1')\prec_{e_2, \tau_2} (M[\sigma_1(\sval_2)], \sigma_1).
      \label{eq:sim-ret-goal}
    \end{equation}

    According to Algorithm~\ref{alg:update-mem}, when executing $\text{updateMem}$ to generate $\MType_2$, we generate a map $S$ that maps each slot in $\dom{\MType_2}$ to slots in $\dom{\MType_1}$ that belong to $s$.
    Consider the following cases:
    \begin{itemize}
        \item $s=\sigmaCall([\spVar, \spVar+8))=[\spVar+c-8, \spVar+c)$ (where $\sigmaCall(\spVar)=\spVar+c-8$ and $\RType_2[r_\reg{rsp}]=\RType_0[r_\reg{rsp}]=(\spVar+c, \_)$): according to \textsc{Typing-Callq}, $\MType_0[s]=(\emptyset, \_)$. According to \textsc{Typing-Func}, for all $s_1\in\dom{\MType_1}$, $s_1\cap s=\emptyset$, so $S[s]=\emptyset$. These two facts imply that $\sval_2=\emptyset$, so $\sigma_1(\sval_2)=\emptyset$ and we do not need to worry about this case.
        \item $S[s]=\emptyset$ and $s\neq [\spVar+c-8, \spVar+c)$: in this case, we have $\sval_2=\sval_0$ and $e_2=e_0$.
        Since $\sigma_1(\sval_2)\neq \emptyset$, then $\sigma_1(\sval_0)\neq \emptyset$.
        According to the simulation relation between $R_p$, $M_p$ and $R_p'$, $M_p'$,
        there should be
        \begin{equation}
            (M_p'[\sigma_1(\sval_0)+\delta_s], \sigma_1')\prec_{e_0, \tau_0} (M_p[\sigma_1(\sval_0)], \sigma_1).
            \label{eq:sim-ret-before-call}
        \end{equation}
        Recall that
        \begin{align*}
            & \forall x\not\in \text{getDom}(\MType, \sigma_f)\cup \sRet.\; M[x]=M_p[x]\\
            & \forall x\not\in \text{getShiftedDom}(\MType, \omega, \sigma_f, \sigma_f')\cup \sRet.\; M'[x]=M_p'[x].
        \end{align*}
        It is also straightforward to derive that
        \begin{align*}
            & \sigma_1(\sval_2) \cap (\text{getDom}(\MType, \sigma_f)\cup \sRet) = \emptyset\\
            & (\sigma_1(\sval_2)+\delta_s) \cap (\text{getShiftedDom}(\MType, \omega, \sigma_f, \sigma_f') \cup \sRet) = \emptyset.
        \end{align*}
        by applying Lemma~\ref{lemma:possible-slot-non-overlap},
        so $M'[\sigma_1(\sval_2)+\delta_s]=M_p'[\sigma_1(\sval_0)+\delta_s]$ and $M[\sigma_1(\sval_2)]=M_p[\sigma_1(\sval_0)]$.
        Thus, we can derive (\ref{eq:sim-ret-goal}) from (\ref{eq:sim-ret-before-call}).

        \item $\text{cardinal}(S[s])>1$: according to $\text{updateMem}$, $\sval_2=\sval_0=\emptyset$, so $\sigma_1(\sval_2)=\emptyset$ and we do not need to worry about this case.
        
        \item $\text{cardinal}(S[s])=1$: denote $S[s]=\qty{s_1}$ and $\MType_1[s_1]=(\sval_1, (e_1, \tau_1))$.
        Then, $\sval_2=(\sval_0\backslash s_1)\cup \sval_1$.
        According to the definition of $\MType_1$, there exists $s_c\in \MType$ such that $s_c=(\sigmaCall \cup \sigmaRet)(s_1)$.
        Denote $\MType[s_c]=(\sval_c, (e_c, \tau_c))$.
        Then, there should be $\MType_1[s_1]=(\sigmaCall \cup \sigmaRet)(\sval_c, (e_c, \tau_c))$.

        Denote $(\Delta_c, \RType_c, \MType_c)=\text{getStateType}(\PType, \text{getFuncPc}(P, f))$.
        According to assertions in $\text{updateMem}$ (Algorithm~\ref{alg:update-mem}), there should be $\neg \text{isSpill}(s_c)$. 
        Then, by applying Lemma~\ref{lemma:same-taint-mtype}, \ref{lemma:same-trans-strategy-mtype}, and \ref{lemma:sim-mem-same-shift}, we can get
        \begin{align*}
            \delta_s
            & = \text{getShift}(\MType_2, \omega_0, \sigma_1, \sigma_1', s)
            = \text{getShift}(\MType_0, \omega_0, \sigma_1, \sigma_1', s)\\
            & = \text{getShift}(\MType_c, \omega, \sigma, \sigma', s_c)
            = \text{getShift}(\MType, \omega, .\sigma, \sigma', s_c).
        \end{align*}

        By unfolding the assumption $\hasVal{P, \PType}{S' \prec S}$, we can get
        \begin{equation}
            \sigma(\sval_c) = \emptyset \lor (M'[\sigma(\sval_c)+\delta_s], \sigma')\prec_{e_c, \tau_c} (M[\sigma(\sval_c)], \sigma).
            \label{eq:sim-ret-rel}
        \end{equation}

        According to Lemma~\ref{lemma:call-ret-substitute}, for all $e$, $\sigma(e) = \sigma_1((\sigmaCall\cup \sigmaRet)(e))$.
        Similarly, for all $e$ such that $\sigmaCall(e)=\sigmaCall'(e)$,
        $\sigma'(e)= \sigma_1'((\sigmaCall'\cup \sigmaRet)(e)) = \sigma_1'((\sigmaCall\cup \sigmaRet)(e))$.

        According to \textsc{Typing-Func}, we also have $\sval = \emptyset \lor \text{isNonChangeExp}(e_c, \tau_c)$.
        Then, according to the definition of $\CompilerPtr$ (Algorithm~\ref{alg:compiler-ptr}), $\sval=\emptyset \lor \sigmaCall(e_c)=\sigmaCall'(e_c)$.
        Then, we can derive this from (\ref{eq:sim-ret-rel}) by discussing whether $e_c$ is $\top$ or not:
        \begin{equation}
            \sigma_1(\sval_1) = \emptyset \lor (M'[\sigma_1(\sval_1)+\delta_s], \sigma_1')\prec_{e_1, \tau_1} (M[\sigma_1(\sval_1)], \sigma_1).
            \label{eq:sim-ret-sub-rel}
        \end{equation}

        We then consider the following cases.

        If $\sval_2=\sval_1$, then $e_2=e_1$. The statement can be proved by applying (\ref{eq:sim-ret-sub-rel}) and discussing whether $\sigma_1(\sval_1)=\emptyset$ and whether $e_1=\top$.

        If $\sval_1\subsetneqq \sval_2$, then $\hasVal{\Delta}{\text{isNonChangeExp}(e_0)}$, $\hasVal{\Delta}{\text{isNonChangeExp}(e_1)}$, and $e_2=\top$.
        We only need to prove that for all $x\in \sigma_1(\sval_2)$, $M'[x+\delta_s]=M[x]$.
        If $x\in \sval_0\backslash s_1$, we can prove this by applying (\ref{eq:sim-ret-before-call}) (which is derived from the simulation relation between $R_p$, $M_p$, and $R_p'$, $M_p'$) and $\hasVal{\Delta}{\text{isNonChangeExp}(e_0)}$.
        If $x\in \sval_1$, we can prove this by applying (\ref{eq:sim-ret-sub-rel}) and $\hasVal{\Delta}{\text{isNonChangeExp}(e_1)}$.
    \end{itemize}
    \todo{Double check this after finishing a raw version! especially notations}
\end{itemize}
\end{proof}


\begin{lemma}
For all $P$, $\PType$, $\pcVar_0$, $\pcVar_1$, and for $i\in\qty{0,1}$, $(\Delta_i, \RType_i, \MType_i)=\text{getStateType}(\PType, \pcVar_i)$,
if $\text{getFunc}(P, \pcVar_0)=\text{getFunc}(P, \pcVar_1)$, then for all $s\in\dom{\MType_0}$ such that 
$\neg \text{isSpill}(s)$
and $\MType_0[s]=(\_, (\_, \tau_0))$, $\MType_1[s]=(\_, (\_, \tau_1))$, there must be $\tau_0=\tau_1$.
\label{lemma:same-taint-mtype}
\end{lemma}
\begin{proof}
This lemma means that each non-spill memory slot has the same taint type at all instructions in the same function.

First, we prove that the statement holds when $\pcVar_0$ and $\pcVar_1$ refer to two consecutive instructions in the same basic block, i.e., $\pcVar_1=\text{nextPc}(P, \pcVar_0)$.
Denote $\instVar=\text{getInst}(P, \pcVar_0)$, then according to our assumption, $\instVar$ is not $\instr{jmp}$ or $\instr{retq}$. Consider the following cases for $\instVar$ (other cases can be proved similarly):

\textbf{Case $\inst{movq}{r, \memOp{r_d}{r_i}{i_s}{i_d}\opMode{\sOp,\tauOp}}$.}
For $s\neq \sOp$, \textsc{Typing-Movq-r-m} implies that $\MType_0[s]=\MType_1[s]$, so $\tau_0=\tau_1$.

For $s=\sOp$, since $\neg \text{isSpill}(s)$, then $\tau_0=\tauOp=\tau_1$ according to \textsc{Typing-Movq-r-m} and \textsc{Typing-StoreOp-Non-Spill}.

\textbf{Case $\inst{movq}{\memOp{r_d}{r_i}{i_s}{i_d}\opMode{\sOp,\tauOp}, r}$.}
\textsc{Typing-Movq-m-r} implies that $\MType_0=\MType_1$, so the statement is true.

\textbf{Case $\inst{jne}{\ell\opMode{\sigmaOp}}$.}
Since $\pcVar_1=\text{nextPc}(P, \pcVar_0)$, $\pcVar_1$ refers to the next instruction when the branch is not taken.
Then, according to \textsc{Typing-Jne}, there should be $\MType_0=\MType_1$, so the statement is true.

\textbf{Case $\inst{callq}{f_c\opMode{\sigmaCall,\sigmaRet}}$.}
Denote $(\Delta_c, \RType_c, \MType_c)=\PType(f_c)(f_{c\text{ret}})$.
Then, $\MType_1 = \text{updateMem}(\MType_0, (\sigmaCall\cup \sigmaRet)(\MType_c))$.
According to the definition of $\text{updateMem}$ (Algorithm~\ref{alg:update-mem}), there should be $\tau_0=\tau_1$.

Second, we prove that the statement holds when $\pcVar_0$ refers to a branch instruction, and $\pcVar_1$ is its branch target.
We only illustrate the case where $\pcVar_0$ refers to the branch instruction $\inst{jne}{\ell\opMode{\sigmaOp}}$, since other cases can be proved similarly.
\textsc{Typing-Jne} and \textsc{State-Subtype} imply that $\dom{\MType_0}=\dom{\MType_1}$ and $\MType_0[s]\sqsubseteq \sigmaOp(\MType_1[s])$.
Since $\neg\text{isSpill}(s)$,
then \textsc{Mem-Slot-Subtype} implies that $\tau_0=\sigmaOp(\tau_1)$.
According to \textsc{Typing-Jne}, $\text{getTaintVar}(\dom{\sigmaOp})=\emptyset$, so $\tau_0=\sigmaOp(\tau_1)=\tau_1$.

By combining the above two cases ($\pcVar_0$ and $\pcVar_1$ refer to consecutive instruction or branch/target instructions), we can further derive that the statement holds for all $\pcVar_0$ and $\pcVar_1$ that belong to the same function.
\todo{I may need to explain that we assume there is no dead block somewhere.}
\end{proof}

\begin{lemma}
For all $P$, $\PType$, $\pcVar_0$, $\pcVar_1$, and for $i\in\qty{0,1}$, $(\Delta_i, \RType_i, \MType_i)=\text{getStateType}(\PType, \pcVar_i)$,
if $\text{getFunc}(P, \pcVar_0)=\text{getFunc}(P, \pcVar_1)$, and $\MType_0$ satisfies the assertion in getTransStrategy (Algorithm~\ref{alg:trans-strategy}), then there should be $\text{getTransStrategy}(\MType_0)=\text{getTransStrategy}(\MType_1)$.
\label{lemma:same-trans-strategy-mtype}
\end{lemma}
\begin{proof}
This lemma means that memory state types of instructions in the same function have the same transformation strategy map.
For each non-local-stack slot $s$, $s$ is not a spill slot either, so $\MType_0[s]$ and $\MType_1[s]$ have the same taint type according to Lemma~\ref{lemma:same-taint-mtype}.
Thus, $\MType_1$ also satisfies the assertion in \text{getTransStrategy} (Algorithm~\ref{alg:trans-strategy}).

Then, we can denote $\omega_i=\text{getTransStrategy}(\MType_i)$, $i\in\qty{0, 1}$. 
We aim to prove that $\omega_0=\omega_1$.

Recall that in getTransStrategy (Algorithm~\ref{alg:trans-strategy}), we first build a map from each base pointer to a set of taint types of slots referenced by the pointer. Denote the map for $\MType_i$ as $T_i$, $i\in \qty{0, 1}$.
Note that if the base pointer of a slot is not $\spVar$, then the slot must not be a local stack slot.
Hence, for all $\ptrVar\in \dom{T_1}$, $\ptrVar\neq \spVar$, $T_0[\ptrVar]=T_1[\ptrVar]$.

For each $s\in\dom{\MType_0}$ with its base pointer denoted as $\ptrVar=\text{getPtr}(s)$, we consider the following cases:
\begin{itemize}
    \item $\ptrVar\neq \spVar$ and $\text{cardinal}(T_0[\ptrVar])=1$:
    $\ptrVar\neq \spVar$ implies that $T_0[\ptrVar]=T_1[\ptrVar]$, so $\text{cardinal}(T_0[\ptrVar])=\text{cardinal}(T_1[\ptrVar])=1$. Thus, $\omega_0(s)=\omega_1(s)=\transPtr$.
    \item $\ptrVar\neq \spVar$ and $\text{cardinal}(T_0[\ptrVar])\neq 1$:
    $\ptrVar\neq \spVar$ implies that $T_0[\ptrVar]=T_1[\ptrVar]$, so $\text{cardinal}(T_0[\ptrVar])=\text{cardinal}(T_1[\ptrVar])\neq 1$. Thus, $\omega_0(s)=\omega_1(s)=\transOp$.
    \item $\ptrVar=\spVar$: in this case, there should be $\omega_0(s)=\omega_1(s)=\transOp$.
\end{itemize}
Therefore, $\omega_0=\omega_1$.
\end{proof}

\begin{lemma}
    For all $\sigma$, $\sigma'$, $\MType$ and $\omega=\text{getTransStrategy}(\MType)$, if there exists $M$, $M'$ such that $\hasType{}{M}{\sigma(\MType)}$ and $\hasVal{}{(M', \sigma')\prec_{\MType}(M, \sigma)}$, then
    for all $s\in \dom{\MType}$, 
    \begin{align*}
        \text{getShiftedSlot}(\MType, \omega, \sigma, \sigma', s)
        & =
        \begin{cases}
            \sigma(s)\text{ or }\sigma(s)+\delta & \sigma(s)\subseteq \stackPub\\
            \sigma(s) & \text{otherwise},
        \end{cases}\\
        \text{getPossibleSlot}(\MType, \omega, \sigma, \sigma', s)
        & \subseteq 
        \begin{cases}
            \sigma(s)\cup (\sigma(s)+\delta) & \sigma(s)\subseteq \stackPub\\
            \sigma(s) & \text{otherwise},
        \end{cases}
    \end{align*}
    and
    \begin{equation*}
        \text{getShiftedSlot}(\MType, \omega, \sigma, \sigma', s) \subseteq
        \text{getPossibleSlot}(\MType, \omega, \sigma, \sigma', s).
    \end{equation*}
    \label{lemma:shifted-slot}
\end{lemma}
\begin{proof}
\textsc{Sim-Mem} implies that
\begin{align*}
    s' & \coloneq \text{getShiftedSlot}(\MType, \omega, \sigma, \sigma', s)\\
    & = \sigma(s)+\text{getOpShift}(\MType, \omega, s)+
    (\sigma'(\text{getPtr}(s))-\sigma(\text{getPtr}(s)))\\
    & = \sigma(s)+\text{getOpShift}(\MType, \omega, s)+\text{getPtrShift}(\MType, \omega, \sigma, s).
\end{align*}
We also denote
\begin{align*}
    s'' & \coloneq \text{getPossibleSlot}(\MType, \omega, \sigma, \sigma', s) \\
    & =
    \begin{cases}
        \sigma(s)\cup (\sigma(s)+\delta) & \omega(s)=\transOp\\
        \sigma(s)+(\sigma'(\text{getPtr}(s))-\sigma(\text{getPtr}(s))) & \omega(s)=\transPtr.
    \end{cases}
\end{align*}
We consider the following three cases:
\begin{itemize}
    \item $\sigma(s)\subseteq \stackPub$ and $\omega(s)=\transOp$: 
    \begin{align*}
        s'
        & = \sigma(s)+\text{getOpShift}(\MType, \omega, s)+0\\
        & = \sigma(s) \text{ or } \sigma(s)+\delta,\\
        s''& =\sigma(s)\cup (\sigma(s)+\delta).
    \end{align*}
    This implies that $s'\subseteq s''$.
    \item $\sigma(s) \subseteq \stackPub$ and $\omega(s)=\transPtr$:
    \begin{align*}
        s'
        & = \sigma(s)+0+(\sigma'(\text{getPtr}(s))-\sigma(\text{getPtr}(s)))\\
        & = \sigma(s)+ \text{getPtrShift}(\MType, \omega, \sigma, s)\\
        & = \sigma(s) \text{ or } \sigma(s)+\delta\\
        s'' & = \sigma(s)+ (\sigma'(\text{getPtr}(s))-\sigma(\text{getPtr}(s)))= s'\\
        & =\sigma(s)+\text{getPtrShift}(\MType, \omega, \sigma, s) \\ 
        & \subseteq \sigma(s)\cup (\sigma(s)+\delta).
    \end{align*}
    \item $\sigma(s)\not\subseteq \stackPub$:
    this implies that 
    \begin{equation*}
        \sigma'(\text{getPtr}(s))-\sigma(\text{getPtr}(s))=\text{getPtrShift}(\MType, \omega, \sigma, s)=0.
    \end{equation*}
    \textsc{Sim-Mem} requires $\omega(s)=\transOp \Rightarrow \sigma(s)\subseteq \stackPub$, so in this case $\omega(s)\neq \transOp$, i.e., $\omega(s)=\transPtr$.
    Therefore,
    \begin{align*}
        s' & = \sigma(s)+0+0 =\sigma(s)\\
        s'' & =\sigma(s)+0 =\sigma(s) = s' \subseteq \sigma(s)\cup (\sigma(s)+\delta).
    \end{align*}
\end{itemize}
\end{proof}

\begin{lemma}
    For all $\sigma$, $\sigma'$, $\MType$ and $\omega=\text{getTransStrategy}(\MType)$, if there exists $M$, $M'$ such that $\hasType{}{M}{\sigma(\MType)}$ and $\hasVal{}{(M', \sigma')\prec_{\MType}(M, \sigma)}$, then
    for all $s_1, s_2\in \dom{\MType}$, $\sigma(s_1)\cap \sigma(s_2)=\emptyset$, denoting    
    \begin{align*}
        s_1' & =\text{getPossibleSlot}(\MType, \omega, \sigma, \sigma', s_1)\\
        s_2' & = \text{getPossibleSlot}(\MType, \omega, \sigma, \sigma', s_2),
    \end{align*}
    there should be $s_1'\cap s_2'=\emptyset$.
    \label{lemma:possible-slot-non-overlap}
\end{lemma}
\begin{proof}
\textsc{Mem-Type} implies that for $i\in \qty{1, 2}$, $\sigma(s_i)\subseteq \stackPub\lor \sigma(s_i) \subseteq \otherPub\cup \otherSec$.
We then consider the following four cases about $s_1$ and $s_2$:
\begin{itemize}
    \item $\sigma(s_1)\subseteq \stackPub$ and $\sigma(s_2)\subseteq \stackPub$:
    this implies that $(\sigma(s_1)+\delta) \subseteq \stackSec$ and $(\sigma(s_2)+\delta) \subseteq \stackSec$. Since $\stackPub \cap \stackSec = \emptyset$ and $\sigma(s_1) \cap \sigma(s_2)=\emptyset$, then
    \begin{align*}
        s_1' \cap s_2'
        & \subseteq (\sigma(s_1)\cup (\sigma(s_1)+\delta)) \cap (\sigma(s_2)\cup (\sigma(s_2)+\delta))\\
        & = (\sigma(s_1) \cap \sigma(s_2)) \cup (\sigma(s_1) \cap (\sigma(s_2)+\delta)) \\
        & \quad \cup ((\sigma(s_1)+\delta) \cap \sigma(s_2)) \cup ((\sigma(s_1)+\delta) \cap (\sigma(s_2)+\delta))\\
        & = \emptyset \cup \emptyset \cup \emptyset \cup \emptyset = \emptyset.
    \end{align*}
    \item $\sigma(s_1)\subseteq \stackPub$ and $\sigma(s_2)\not\subseteq \stackPub$:
    \begin{align*}
        s_1' \cap s_2'
        & \subseteq (\sigma(s_1)\cup (\sigma(s_1)+\delta)) \cap \sigma(s_2)\\
        & \subseteq (\stackPub \cup \stackSec) \cap (\otherPub \cup \otherSec)\\
        & = \emptyset.
    \end{align*}
    \item $\sigma(s_1)\not\subseteq \stackPub$ and $\sigma(s_2)\subseteq \stackPub$: the proof is similar to the last case.
    \item $\sigma(s_1)\not\subseteq \stackPub$ and $\sigma(s_2)\not\subseteq \stackPub$:
    \begin{equation*}
        s_1' \cap s_2' \subseteq \sigma(s_1) \cap \sigma(s_2) = \emptyset.
    \end{equation*}
\end{itemize}
\end{proof}

\begin{lemma}
    For all $\sigma$, $\sigma'$, $\MType$ and $\omega=\text{getTransStrategy}(\MType)$, if there exists $M$, $M'$ such that $\hasType{}{M}{\sigma(\MType)}$ and $\hasVal{}{(M', \sigma')\prec_{\MType}(M, \sigma)}$, then
    for all $s_1, s_2\in \dom{\MType}$, $\sigma(s_1)\cap \sigma(s_2)=\emptyset$, denoting    
    \begin{align*}
        s_1' & =\text{getShiftedSlot}(\MType, \omega, \sigma, \sigma', s_1)\\
        s_2' & = \text{getShiftedSlot}(\MType, \omega, \sigma, \sigma', s_2),
    \end{align*}
    there should be
    \begin{equation*}
        s_1'\cap s_2' = s_1'\cap \sigma(s_2) = \sigma(s_1)\cap s_2' = \emptyset.
    \end{equation*}
    \label{lemma:shifted-slot-non-overlap}
\end{lemma}
\begin{proof}
    Denote $s_i''=\text{getPossibleSlot}(\MType, \omega, \sigma, \sigma', s_i)$, $i\in \qty{1, 2}$.
    According to Lemma~\ref{lemma:shifted-slot}, there should be $\sigma(s_i)\subseteq s_i''$ and $s_i' \subseteq s_i''$.
    Thus, the statement can be proved by applying Lemma~\ref{lemma:possible-slot-non-overlap}.
\end{proof}

\begin{lemma}
For all $e$, $\tau$, $v$, $t$, $v'$, $t'$, $\sigma$, $\sigma'$, $\Delta$,
if 
\begin{equation*}
    \hasVal{}{((v', t'), \sigma')\prec_{e, \tau} ((v, t), \sigma)} \quad
    \hasVal{\Delta}{\sigma'\prec \sigma},
\end{equation*}
then for all $\sigma_m$, $\sigma_m'$, $e_1$, $\tau_1$, $\Delta_1$ where 
\begin{align*}
    & \hasVal{\Delta}{e=\sigma_m(e_1)\lor (\text{isNonChangeExp}(e)\land \sigma_m(e_1)=\top)} \\
    & \hasVal{\Delta}{\sigma_m(\Delta_1)} \qquad
    \forall x\in\dom{\sigma_m}.\; \sigma_m(x)\neq \top 
\end{align*}
the following statements hold:
\begin{itemize}
    \item if $\sigma_m'(e_1)=\sigma_m(e_1)=\top$, then $\hasVal{}{((v', t'), \sigma'\circ \sigma_m')\prec_{e_1, \tau_1} ((v, t), \sigma\circ \sigma_m)}$;
    \item if $\sigma_m'(e_1)=\sigma_m(e_1)+c \neq \top$ for some constant number $c$, then $\hasVal{}{((v'+c, t'), \sigma'\circ \sigma_m')\prec_{e_1, \tau_1} ((v, t), \sigma\circ \sigma_m)}$.
\end{itemize}
\label{lemma:sim-val-substitute-general}
\end{lemma}
\begin{proof}
Denote $\sigma_1=\sigma\circ \sigma_m$, $\sigma_1'=\sigma'\circ \sigma_m'$.
Consider the following cases:
\begin{itemize}
    \item $\sigma_m'(e_1)=\sigma_m(e_1)=\top$: according to \textsc{Sim-Val} and our assumption, there should be $t=t'$, so we just need to prove that
    \begin{equation*}
        (e_1=\top \land v'=v)\lor (e_1\neq \top \land v'=\sigma_1'(e_1)\land v=\sigma_1(e_1)).
    \end{equation*}
    Since for all $x\in\dom{\sigma_m}$, $\sigma_m(x)\neq \top$, then $\sigma_m(e_1)=\top$ implies that $e_1=\top$.
    Consider the following cases:
    \begin{itemize}
        \item $e=\top$: according to \textsc{Sim-Val}, $\hasVal{}{((v', t'), \sigma')\prec_{e, \tau} ((v, t), \sigma)}$ implies that $v=v'$, so the statement is true.
        \item $e\neq \top$: in this case, $\sigma_m(e_1)=\top \neq e$. Then, according to our assumption, there must be $\text{isNonChangeExp}(e)$, so $v'=\sigma'(e)=\sigma(e)=v$ and the statement is true.
    \end{itemize}
    \item $\sigma_m'(e_1)=\sigma_m(e_1)+c \neq \top$: according to \textsc{Sim-Val} and our assumption, there should be $t=t'$, so we just need to prove that
    \begin{equation*}
        (e_1=\top \land v'+c=v)\lor (e_1\neq \top \land v'+c=\sigma_1'(e_1)\land v=\sigma_1(e_1)).
    \end{equation*}
    According to our assumption, there must be $\sigma_m(e_1)=e$ and $\sigma_m'(e_1)=\sigma_m(e_1)+c=e+c$.
    Then, 
    \begin{align*}
        \sigma_1(e_1) & = \sigma(e) = v &
        \sigma_1'(e_1) & = \sigma'(e)+c = v'+c.
    \end{align*}
    Thus, the statement is true.
\end{itemize}
\end{proof}

\begin{lemma}
For all $e$, $\tau$, $v$, $t$, $v'$, $t'$, $\sigma$, $\sigma'$, $\Delta$,
if 
\begin{equation*}
    \hasVal{}{((v', t'), \sigma')\prec_{e, \tau} ((v, t), \sigma)} \quad
    \hasVal{\Delta}{\sigma'\prec \sigma},
\end{equation*}
then for all $\sigma_m$, $e_1$, $\tau_1$, $\Delta_1$ where 
\begin{align*}
    & \hasVal{\Delta}{e=\sigma_m(e_1)\lor (\text{isNonChangeExp}(e)\land \sigma_m(e_1)=\top)} \\
    & \hasVal{\Delta}{\sigma_m(\Delta_1)} \qquad
    \forall x\in\dom{\sigma_m}.\; \sigma_m(x)\neq \top,
\end{align*}
then there should be
\begin{equation*}
    \hasVal{}{((v', t'), \sigma'\circ \sigma_m)\prec_{e_1, \tau_1} ((v, t), \sigma\circ \sigma_m)}.
\end{equation*}
\label{lemma:sim-val-substitute-special}
\end{lemma}
\begin{proof}
This can be proved by applying Lemma~\ref{lemma:sim-val-substitute-general} with $\sigma_m'=\sigma_m$ and $c=0$.
\end{proof}

\begin{lemma}
For all $(\Delta, \RType, \MType)$, $(\Delta_1, \RType_1, \MType_1)$, $\sigma$, $\sigma'$, $\sigma_1$, $\sigma_1'$, $\sigma_m$, $\sigma_m'$, if
\begin{align*}
    & \sigma(\Delta)=\mathit{true} \quad \hasVal{\Delta}{\sigma'\prec \sigma} \quad \sigma_1=\sigma \circ \sigma_m \quad \sigma_1'=\sigma'\circ \sigma_m'\\
    & \hasVal{\Delta}{\sigma_m(\Delta_1)} \quad \forall x\in \dom{\sigma_m}.\; \sigma_m(x)\neq \top \quad \sigma_m(\spVar)= \spVar+c\\
    & \forall x,\sigma_m(x)=e,\text{isPtr}(x).\; \hasVal{\Delta}{\text{isNonChangeExp}(e-\text{getPtr}(e))}\\
    & (\_, \sigma_m')=\CompilerPtr((\Delta, \RType, \MType), \sigma_m(\Delta_1, \RType_1, \MType_1), \sigma_m)\\
    & \omega=\text{getTransStrategy}(\MType) \quad \omega_1=\text{getTransStrategy}(\MType_1),
\end{align*}
then for all $s$, $s_1$ where
\begin{equation*}
    \hasVal{\Delta}{\sigma_m(s_1)\subseteq s \land \MType[s]\sqsubseteq \sigma_m(\MType_1)} \quad
    \neg \text{isSpill}(s),
\end{equation*}
then there should be
\begin{equation*}
    \text{getShift}(\MType_1, \omega_1, \sigma_1, \sigma_1', s_1)=\text{getShift}(\MType, \omega, \sigma, \sigma', s).
\end{equation*}
\label{lemma:sim-mem-same-shift}
\end{lemma}
\begin{proof}
According to the definition of $\text{getShift}$, we just need to prove that
\begin{multline*}
    \sigma_1'(\text{getPtr}(s_1))-\sigma_1(\text{getPtr}(s_1)) + \text{getOpShift}(\MType_1, \omega_1, s_1)\\
    = \sigma'(\text{getPtr}(s)) - \sigma(\text{getPtr}(s)) + \text{getOpShift}(\MType, \omega, s).
\end{multline*}

Denote $\text{getPtr}(s_1)=x_1$, $\sigma_m(x_1)=e_x$. According to our assumption, $\hasVal{\Delta}{\text{isNonChangeExp}(e_x-\text{getPtr}(e_x))}$, so $\sigma'(e_x-\text{getPtr}(e_x))=\sigma(e_x-\text{getPtr}(e_x))$.

Furthermore, according to \textsc{Mem-Slot-Subtype}, $\hasVal{\Delta}{\MType[s]\sqsubseteq \sigma_m(\MType_1)}$ and $\neg \text{isSpill}(s)$
imply that $\sigma_1(\tau_1)=\sigma(\tau)$.

Following the definition of $\CompilerPtr$ (Algorithm~\ref{alg:compiler-ptr}), we just need to consider the following cases:
\begin{itemize}
    \item $\omega_1(s_1)=\omega(s)$: according to the definition of $\CompilerPtr$ (Algorithm~\ref{alg:compiler-ptr}), $\sigma_m'(x_1)=\sigma_m(x_1)=e_x$, so
    \begin{align*}
        \sigma_1'(x_1)-\sigma_1(x_1) & = \sigma'(\sigma_m'(x_1)) - \sigma(\sigma_m(x_1))\\
        & = \sigma'(e_x) - \sigma(e_x)\\
        & = \sigma'(\text{getPtr}(e_x)) - \sigma(\text{getPtr}(e_x))\\
        & = \sigma'(\text{getPtr}(\sigma_m(s_1))) - \sigma(\text{getPtr}(\sigma_m(s_1)))\\
        & = \sigma'(\text{getPtr}(s)) - \sigma(\text{getPtr}(s)).
    \end{align*}
    According to the definition of $\text{getOpShift}$, $\sigma_1(\tau_1)=\sigma(\tau)$ and $\omega_1(s_1)=\omega(s)$ implies that
    \begin{equation*}
        \text{getOpShift}(\MType_1, \omega_1, s_1)=\text{getOpShift}(\MType, \omega, s).
    \end{equation*}
    Therefore, the statement holds. 
    \item $\omega_1(s_1)=\transPtr$, $\omega(s)=\transOp$:
    If $\sigma_m(\tau_1)\neq 0$, then $\tau\neq 0$.
    According to the definition of $\CompilerPtr$, $\sigma_m'(x_1)=\sigma_m(x_1)+\delta=e_x+\delta$. Then,
    \begin{align*}
        \sigma_1'(x_1)-\sigma_1(x_1) & = \sigma'(\sigma_m'(x_1)) - \sigma(\sigma_m(x_1))\\
        & = \sigma'(e_x+\delta) - \sigma(e_x)\\
        & = \delta+ \sigma'(\text{getPtr}(e_x)) - \sigma(\text{getPtr}(e_x))\\
        & = \delta+ \sigma'(\text{getPtr}(\sigma_m(s_1))) - \sigma(\text{getPtr}(\sigma_m(s_1)))\\
        & = \delta+ \sigma'(\text{getPtr}(s)) - \sigma(\text{getPtr}(s)).
    \end{align*}
    Furthermore, in this case there should be
    \begin{equation*}
        \text{getOpShift}(\MType_1, \omega_1, s_1) = 0 \quad \text{getOpShift}(\MType, \omega, s)=\delta.
    \end{equation*}
    Therefore, the statement holds.

    If $\sigma_m(\tau_1)=0$, then $\tau=0$.
    According to the definition of $\CompilerPtr$, $\sigma_m'(x_1)=\sigma(x_1)=e_x$. Similar to the first case, we have $\sigma_1'(x_1)-\sigma_1(x_1)=\sigma'(\text{getPtr}(s)) - \sigma(\text{getPtr}(s))$.
    Furthermore, there should also be
    \begin{equation*}
        \text{getOpShift}(\MType_1, \omega_1, s_1) = \text{getOpShift}(\MType, \omega, s)=0.
    \end{equation*}
    Therefore, the statement holds. 
\end{itemize}
\end{proof}

\begin{lemma}
For all $\MType$, $M$, $M'$, $\sigma$, $\sigma'$, $\Delta$, $\RType$, if 
\begin{equation*}
\sigma(\Delta)=\qty{\mathit{true}} \quad
\hasType{}{M}{\sigma(\MType)} \quad
\hasVal{}{(M', \sigma')\prec_{\MType} (M, \sigma)} \quad
\hasVal{\Delta}{\sigma'\prec \sigma},
\end{equation*}
then for all $\MType_1$, $\sigma_m$, $\sigma_m'$, $\Delta_1$, $\RType_1$, where $\MType_1$ satisfies the assertion in getTransStrategy (Algorithm~\ref{alg:trans-strategy}) and
\begin{align*}
& \forall s_1\in\dom{\MType_1}.\; \exists s.\; \hasVal{\Delta}{\sigma_m(s_1)\subseteq s\land \MType[s]\sqsubseteq \sigma_m(\MType_1[s_1])}\\
& \hasVal{\Delta}{\sigma_m(\Delta_1)}
\quad \forall x\in \dom{\sigma_m}.\; \sigma_m(x)\neq \top
\quad \sigma_m(\spVar)=\spVar+c,\\
& \forall x, \sigma_m(x)=e, \text{isPtr}(x).\; \hasVal{\Delta}{\text{isNonChangeExp}(e - \text{getPtr}(e))}\\
& (\_, \sigma_m') = \CompilerPtr((\Delta,\RType,\MType),\sigma_m(\Delta_1,\RType_1,\MType_1),\sigma_m),
\end{align*}
then there should be 
\begin{equation*}
\hasVal{}{(M', \sigma'\circ \sigma_m') \prec_{\MType_1} (M, \sigma \circ \sigma_m)}.
\end{equation*}
\label{lemma:sim-mem-substitute-general}
\end{lemma}
\begin{proof}
Denote
\begin{align*}
    & \omega = \text{getTransStrategy}(\MType) && \omega_1 = \text{getTransStrategy}(\MType_1)\\
    & \sigma_1=\sigma\circ \sigma_m && \sigma_1'=\sigma'\circ \sigma_m'.
\end{align*}
To prove the lemma, we need to prove that for each $s_1\in \dom{\MType_1}$, where 
\begin{align*}
    & \MType_1[s_1]=(\sval_1, (e_1, \tau_1))
    && x_1 = \text{getPtr}(s_1) \\
    & \deltaOp[1] = \text{getOpShift}(\MType_1, \omega_1, s_1)
    && \deltaPtr[1] = \sigma_1'(x_1)-\sigma_1(x_1),
\end{align*}
the following statements hold:
\begin{itemize}
    \item $\omega_1(s_1) = \transOp \Rightarrow \sigma_1(s_1)\subseteq \stackPub \land (\text{isLocalStack}(s_1) \lor \tau_1 \in \qty{0, 1})$;
    \item $\deltaPtr[1] = \text{getPtrShift}(\MType_1, \omega_1, \sigma_1, s_1)$;
    \item $\sigma_1(\sval_1) = \emptyset \lor ((M'[\sigma_1(\sval_1)+\deltaPtr[1]+\deltaOp[1]], \sigma_1') \prec_{e_1, \tau_1} (M[\sigma_1(\sval_1)], \sigma_1))$.
\end{itemize}
According to our assumption, there must exists $s\in\dom{\MType}$ such that $\hasVal{\Delta}{\sigma_m(s_1)\subseteq s\land \MType[s]\sqsubseteq \sigma_m(\MType_1[s_1])}$.
Denote $\MType[s]=(\sval, (e, \tau))$.
We prove the three statements as follows:
\begin{itemize}
    \item $\omega_1(s_1) = \transOp \Rightarrow \sigma_1(s_1)\subseteq \stackPub \land (\text{isLocalStack}(s_1) \lor \tau_1 \in \qty{0, 1})$:
    according to our assumption, $\MType_1$ satisfies the assertion in Algorithm~\ref{alg:trans-strategy}, so when $\omega_1(s_1)=\transOp$, $\text{isLocalStack}(s_1)\lor \tau_1\in \qty{0,1}$.

    We just need to prove that when $\omega_1(s_1)=\transOp$, $\sigma_1(s_1)\subseteq \stackPub$.
    Since $\hasVal{\Delta}{\sigma_m(s_1)\subseteq s}$, then $\sigma_1(s_1)=(\sigma\circ\sigma_m)(s_1) \subseteq \sigma(s)$.

    Furthermore, according to the assertion in $\CompilerPtr$ (Algorithm~\ref{alg:compiler-ptr}), when $\omega_1(s_1)=\transOp$, there should be $\omega(s)=\transOp$.
    Then, $\hasVal{}{(M', \sigma')\prec_\MType(M, \sigma)}$ and \textsc{Sim-Mem} implies that $\sigma(s)\subseteq \stackPub$.
    Therefore, $\sigma_1(s_1)\subseteq \sigma(s)\subseteq \stackPub$.
    \item $\deltaPtr[1] = \text{getPtrShift}(\MType_1, \omega_1, \sigma_1, s_1)$:
    denote $\sigma_m(x_1)=e_x$. According to our assumption, $\hasVal{\Delta}{\text{isNonChangeExp}(e_x-\text{getPtr}(e_x))}$, so $\sigma'(e_x-\text{getPtr}(e_x))=\sigma(e_x-\text{getPtr}(e_x))$.
    According to \textsc{Sim-Mem}, $\hasVal{}{(M', \sigma')\prec_\MType(M, \sigma)}$ implies that $\sigma'(\text{getPtr}(s))-\sigma(\text{getPtr}(s))=\text{getPtrShift}(\MType, \omega, \sigma, s)$.
    Consider the following cases: 
    \begin{itemize}
        \item $\omega_1(s_1)=\omega(s)=\transOp$: according to the definition of $\CompilerPtr$ (Algorithm~\ref{alg:compiler-ptr}), $\sigma_m'(x_1)=\sigma_m(x_1)=e_x$. Then,
        \begin{align*}
            \deltaPtr[1]& = \sigma_1'(x_1) - \sigma_1(x_1)\\
            & = \sigma'(\text{getPtr}(s)) - \sigma(\text{getPtr}(s))\\
            & = \text{getPtrShift}(\MType, \omega, \sigma, s) = 0\\
            & = \text{getPtrShift}(\MType_1, \omega_1, \sigma_1, s_1),
        \end{align*}
        where the second step can be derived similarly to the proof of Lemma~\ref{lemma:sim-mem-same-shift}.
        \item $\omega_1(s_1)=\omega(s)=\transPtr$: according to the definition of $\CompilerPtr$ (Algorithm~\ref{alg:compiler-ptr}), $\sigma_m'(x_1)=\sigma_m(x_1)=e_x$.
        According to Algorithm~\ref{alg:trans-strategy}, $s_1$ and $s$ are not referenced by the stack pointer, so they are not spill slots. Then, $\hasVal{\Delta}{\MType[s]\sqsubseteq \sigma_m(\MType_1[s_1])}$ implies that $\hasVal{\Delta}{\tau=\sigma_m(\tau_1)}$, so $\sigma_1(\tau_1)=\sigma(\tau)$.
        Then, we can derive
        \begin{align*}
            \deltaPtr[1]& = \sigma_1'(x_1) - \sigma_1(x_1)\\
            & = \sigma'(\text{getPtr}(s)) - \sigma(\text{getPtr}(s))\\
            & = \text{getPtrShift}(\MType, \omega, \sigma, s)\\
            & = \text{getPtrShift}(\MType_1, \omega_1, \sigma_1, s_1),
        \end{align*}
        where 
        the second step can be derived similarly to the proof of Lemma~\ref{lemma:sim-mem-same-shift},
        and the last step holds because $\sigma_1(\tau_1)=\sigma(\tau)$ and $\sigma_1(s_1)\subseteq \sigma(s)$.
        \item $\omega_1(s_1)=\transPtr$, $\omega(s)=\transOp$: 
        if $\sigma_m(\tau_1)\neq 0$,
        according to the definition of $\CompilerPtr$ (Algorithm~\ref{alg:compiler-ptr}), $\sigma_m'(x_1)=\sigma_m(x_1)+\delta=e_x+\delta$.
        \begin{align*}
            \deltaPtr[1]& = \sigma_1'(x_1) - \sigma_1(x_1)\\
            & = \delta+ \sigma'(\text{getPtr}(s)) - \sigma(\text{getPtr}(s))\\
            & = \delta+\text{getPtrShift}(\MType, \omega, \sigma, s) = \delta\\
            & = \text{getPtrShift}(\MType_1, \omega_1, \sigma_1, s_1),
        \end{align*}
        where the second step can be derived similarly to the proof of Lemma~\ref{lemma:sim-mem-same-shift}.

        If $\sigma_m(\tau_1)=0$, then $\sigma_m'(x_1)=\sigma_m(x_1)=e_x$ and $\sigma_1(\tau_1)=\sigma(\sigma_m(\tau_1))=0$.
        So
        \begin{align*}
            \deltaPtr[1]& = \sigma_1'(x_1) - \sigma_1(x_1)\\
            & = \sigma'(\text{getPtr}(s)) - \sigma(\text{getPtr}(s))\\
            & = \text{getPtrShift}(\MType, \omega, \sigma, s)=0\\
            & = \text{getPtrShift}(\MType_1, \omega_1, \sigma_1, s_1),
        \end{align*}
        where the second step can be derived similarly to the proof of Lemma~\ref{lemma:sim-mem-same-shift}, and the last step holds since $\sigma_1(\tau_1)=0$ implies that $\text{getPtrShift}(\MType_1, \omega_1, \sigma_1, s_1)=0$ or $\delta$.
    \end{itemize}
    Therefore, the statement is true.
    \item $\sigma_1(\sval_1) = \emptyset \lor ((M'[\sigma_1(\sval_1)+\deltaPtr[1]+\deltaOp[1]], \sigma_1') \prec_{e_1, \tau_1} (M[\sigma_1(\sval_1)], \sigma_1))$:
    if $\sigma_1(\sval_1)= \emptyset$, the statement holds.

    If $\text{isSpill}(s)$, then according to the assertions in $\CompilerPtr$ (Algorithm~\ref{alg:compiler-ptr}), $\sval_1=\emptyset$, so $\sigma_1(\sval_1)=\emptyset$ and the statement holds.

    If $\neg \text{isSpill}(s)$ and $\sigma_1(\sval_1)\neq \emptyset$, we need to prove that
    \begin{equation*}
        (M'[\sigma_1(\sval_1)+\deltaPtr[1]+\deltaOp[1]], \sigma_1') \prec_{e_1, \tau_1} (M[\sigma_1(\sval_1)], \sigma_1).
    \end{equation*}
    Denote $\deltaOp=\text{getOpShift}(\MType, \omega, s)$, $\deltaPtr=\sigma'(\text{getPtr}(s))-\sigma(\text{getPtr}(s))$.
    In this case, we can apply Lemma~\ref{lemma:sim-mem-same-shift} and get that $\deltaPtr[1]+\deltaOp[1]=\deltaPtr+\deltaOp$.
    Since $\sigma(\sigma_m(\sval_1))=\sigma_1(\sval_1)\neq \emptyset$, then $\sigma_m(\sval_1)\neq \emptyset$.
    Then, $\hasVal{\Delta}{\MType[s]\sqsubseteq \sigma_m(\MType_1[s_1])}$ and \textsc{Mem-Slot-Subtype} imply that
    \begin{align*}
        & \hasVal{\Delta}{\sigma_m(\sval_1)\subseteq \sval}\\
        & \hasVal{\Delta}{e=\sigma_m(e_1) \lor (\text{isNonChangeExp}(e)\land \sigma_m(e_1)=\top)}.
    \end{align*}
    Then, we can get $\sigma_1(\sval_1)=\sigma(\sigma_m(\sval_1))\subseteq \sigma(\sval)$.
    Since $\sigma_1(\sval_1)\neq \emptyset$, then $\sigma(\sval)\neq \emptyset$.
    According to \textsc{Sim-Mem}, $\hasVal{}{(M', \sigma')\prec_\MType(M, \sigma)}$ implies that
    \begin{equation}
        (M'[\sigma(\sval)+\deltaPtr+\deltaOp], \sigma')\prec_{e, \tau} (M[\sigma(\sval)], \sigma).
        \label{eq:lemma-sub-sim-mem-before}
    \end{equation}
    
    Consider the following cases:
    \begin{itemize}
        \item $\hasVal{\Delta}{\sigma_m(\sval_1) = \sval}$: in this case, $\sigma_1(\sval_1)=\sigma(\sval)$, so $M[\sigma_1(\sval_1)]=M[\sigma(\sval)]$. Furthermore, since $\deltaPtr[1]+\deltaOp[1]=\deltaPtr+\deltaOp$, then $M'[\sigma_1(\sval_1)+\deltaPtr[1]+\deltaOp[1]]=M'[\sigma(\sval)+\deltaPtr+\deltaOp]$.
        According to the assertions in $\CompilerPtr$ (Algorithm~\ref{alg:compiler-ptr}), there should be $\sigma_m(e_1)=\sigma_m'(e_1)$.
        Thus, the statement can be proved by applying Lemma~\ref{lemma:sim-val-substitute-general}.
        \item $\hasVal{\Delta}{\sigma_m(\sval_1) \subsetneqq \sval}$: note that $e$ describes the value of memory data in region $\sval$, and $e_1$ describes the value of memory data in region $\sval_1$, while in this case, the size of $\sigma_m(\sval_1)$ is smaller than the size of $\sval$.
        The size difference implies that $e=\sigma_m(e_1)\neq \top$ cannot be true.
        Hence, there should be $\hasVal{\Delta}{\text{isNonChangeExp}(e)\land \sigma_m(e_1)=\top}$.
        By discussing whether $e=\top$ and applying \textsc{Sim-Val} and \textsc{Sim-Var-Map}, we can derive $M'[\sigma(\sval)+\deltaPtr+\deltaOp]=M[\sigma(\sval)]$ from (\ref{eq:lemma-sub-sim-mem-before}).
        In other words, for all $x\in \sigma(\sval)$, $M'[x+\deltaPtr+\deltaOp]=M[x]$.
        Since $\sigma_1(\sval_1)\subseteq \sigma(\sval)$ and $\deltaPtr[1]+\deltaOp[1]=\deltaPtr+\deltaOp$, then for all $x\in \sigma_1(\sval_1)$, $M'[x+\deltaPtr[1]+\deltaOp[1]]=M[x]$.
        Therefore, the statement is true.
    \end{itemize}
\end{itemize}
\end{proof}

\subsection{Public Noninterference}
\label{sec:public-noninterference}
In this section, we prove that $\Compiler$ transforms a well-formed program such that the output program satisfies public noninterference.


First, as shown in the following lemma, we prove that a well-formed program never writes secrets to the non-stack public region ($\otherPub$).
\begin{lemma}
If $\typeChecked{P, \PType}{S}$, $S\xrightarrow{\instVar} S_1$, and $\instVar$ stores $(v, t)$ to address $v_a\in \otherPub$, then $t=0$.
\label{lemma:orig-store-other-pub}
\end{lemma}
\begin{proof}
We prove for the case where $\instVar=\inst{movq}{r_1, \memOp{r_b}{r_i}{i_s}{i_d}\opMode{s, \tau}}$.
For other kinds of instructions that write to memory, it can be proved similarly.

Denote
$S=(R, M, (\pcVar, (\sigma_f, \sigma_b));\Phi)$ and $\sigma=\sigma_f\cup \sigma_b$.
According to store typing rules and well-formedness rules,
$[v_a, v_a+8)\subseteq \sigma(s)$.
\textsc{Mem-Type} implies that either $\sigma(s)\subseteq \otherPub$ or $\sigma(s)\cap \otherPub=\emptyset$.
Since $v_a\in \otherPub$ and $v_a\in \sigma(s)$, $\sigma(s)\subseteq \otherPub$.

Furthermore, denote $(\Delta_1, \RType_1, \MType_1)=\text{getStateType}(\PType, \text{nextPc}(\pcVar))$.
Then, $\MType_1$ is the memory type after the store and $\MType_1[s]=(\_, (\_, \tau))$ according to the store typing rules.
According to \textsc{Mem-Type}, there should be $\sigma(\tau)=0$.
\textsc{Dyn-Movq-r-m} also constrains that $t \Rightarrow \sigma(\tau)$, so $t=0$.
\end{proof}

Next, we prove the following lemma, which states that the transformed program never writes secrets to all public regions ($\stackPub\cup \otherPub$).
\begin{lemma}
If $\typeChecked{P, \PType}{S}$, $\hasVal{P, \PType}{S'\prec S}$,
$S'\xrightarrow{\Compiler(\instVar)}^* S_1'$, and $\Compiler(\instVar)$ stores $(v', t')$ to address $v_a'\in \stackPub\cup \otherPub$, then $t'=0$.
\label{lemma:trans-store-pub}
\end{lemma}
\begin{proof}
We prove for the case where $\instVar=\inst{movq}{r_1, \memOp{r_b}{r_i}{i_s}{i_d}\opMode{s, \tau}}$.
For other kinds of instructions that write to memory, it can be proved similarly.

Denote $S=(R, M, (\pcVar, (\sigma_f, \sigma_b));\Phi)$, $S'=(R', M', (\pcVar', (\sigma_f', \sigma_b'));\Phi')$, $\sigma=\sigma_f\cup \sigma_b$, and $\sigma'=\sigma_f'\cup \sigma_b'$.
Suppose $\instVar$ stores $(v, t)$ to address $v_a$.
Then, \textsc{Dyn-Movq-r-m} implies that $t\Rightarrow \sigma(\tau)$, $R[r_1]=(v, t)$ and $R'[r_1]=(v', t')$.
Since $\hasVal{P, \PType}{S'\prec S}$, then $t=t'$ according to \textsc{Sim-Reg}.

Furthermore, denote $(\Delta_1, \RType_1, \MType_1)=\text{getStateType}(\PType, \text{nextPc}(\pcVar))$ and $\omega=\text{getTransStrategy}(\PType, \text{nextPc}(\pcVar))$.
According to the proof on the simulation relation for this $\instr{movq}$ instruction, there should be $\MType_1[s]=(\_, (\_, \tau))$,
$[v_a, v_a+8)\subseteq \sigma(s)$ and $[v_a', v_a'+8)\subseteq \text{getShiftedSlot}(\MType_1, \omega, \sigma, \sigma', s)$.

\textsc{Mem-Type} implies that $\sigma(s)\subseteq \stackPub\lor \sigma(s)\subseteq \otherPub \lor \sigma(s)\subseteq \otherSec$ and
\begin{align*}
    s' & \coloneq \text{getShiftedSlot}(\MType_1, \omega, \sigma, \sigma', s)\\
    & = \sigma(s) + \text{getOpShift}(\MType_1, \omega, s) + (\sigma'(\text{getPtr}(s)) - \sigma(\text{getPtr}(s)))\\
    & = \sigma(s) + \text{getOpShift} (\MType_1, \omega, s) + \text{getPtrShift}(\MType_1, \omega, s).
\end{align*}
We consider the following cases:
\begin{itemize}
    \item $\sigma(s)\subseteq \stackPub$:
    According to Lemma~\ref{lemma:shifted-slot}, $s'=\sigma(s)$ or $\sigma(s)+\delta$.
    If $s'=\sigma(s)$ and $\omega(s)=\transPtr$, then
    \begin{equation*}
        s'=\sigma(s) \Rightarrow \text{getPtrShift}(\MType_1, \omega, \sigma, s)=0 \Rightarrow \sigma(\tau)=0 \Rightarrow t'=t = 0.
    \end{equation*}
    
    If $s'=\sigma(s)$ and $\omega(s)=\transOp$, then
    \begin{equation*}
        s'=\sigma(s) \Rightarrow \text{getOpShift}(\MType_1, \omega, s)=0 \Rightarrow \tau=0 \Rightarrow t'=t = 0.
    \end{equation*}
    
    If $s'=\sigma(s)+\delta$, then $s'\in \stackSec$, which contradicts our assumption where $v_a'\in s'$ and $v_a'\in \stackPub\cup \otherPub$. Thus, we can ignore this case.
    \item $\sigma(s)\subseteq \otherPub$:
    Suppose the store value of the original instruction $\instVar$ is $(v, t)$.
    According to Lemma~\ref{lemma:orig-store-other-pub}, $t=0$, so $t'=t=0$.
    \item $\sigma(s)\subseteq \otherSec$:
    According to Lemma~\ref{lemma:shifted-slot},
    $s'=\sigma(s)\in \otherSec$. Since $v_a'\in s'$, this contradicts our assumption that $v_a'\in \stackPub\cup \otherPub$, and we can ignore this case.
\end{itemize}
\end{proof}



We then define public equivalence between two abstract machine states, which constrains that their taint bits represent whether their registers/memory slots share the same value. They should also have the same values in the public memory regions.
\begin{definition}[Public Equivalence]
We denote $(v,t)\simeq_\text{taint}(v',t')$ if $t=t'=1 \lor (v=v'\land t=t'=0)$.
Then, states $S=(R, M, (\pcVar, \_);\Phi)$ and $S'=(R', M', (\pcVar', \_);\Phi)$ are publicly equivalent, i.e., $S\simeq_\text{pub}S'$ if
\begin{align*}
    &\pcVar=\pcVar' \qquad
    \forall r.\; R[r]\simeq_\text{taint} R'[r] \qquad
    \forall x.\; M[x]\simeq_\text{taint} M'[x] \\
    &\forall x\in \stackPub\cup \otherPub.\; M[x]=M'[x]=(\_, 0).
\end{align*}
\label{def:pub-eq}
\end{definition}

Finally, we can prove that the transformed program satisfies public noninterference, as shown in Lemma~\ref{lemma:pub-noninterference} and Theorem~\ref{thm:public-noninterference-revision}.
\shixin{Self note: when running $\Compiler(P)$, not every internal state can be considered as well-formed. For example, when finishing executing part of $\Compiler(\instVar)$, we cannot describe/constrain well-formedness of that state.}
\begin{lemma}
For all $S_0$, $S_1$, $S_0'$, $S_1'$
if $S_0'\simeq_\text{pub} S_1'$ and for $i\in \qty{0,1}$,
\begin{align*}
    & \typeChecked{P, \PType}{S_i}
    && \hasVal{P, \PType}{S_i'\prec S_i},
\end{align*}
then either $S_0'$, $S_1'$ are termination states or there exists $S_0''$, $S_1''$ such that
\begin{align*}
    & S_i'\xrightarrow{\Compiler(\instVar_i), o_i}^* S_i'', i\in \qty{0, 1}
    && o_1 = o_2 &&  S_0''\simeq_\text{pub}S_1''.
\end{align*}
where $o_i$ records all branch targets, load/store addresses, and store values to $\stackPub\cup \otherPub$ when executing $\Compiler(\instVar_i)$. 
\label{lemma:pub-noninterference}
\end{lemma}
\begin{proof}
$S_0'\simeq_\text{pub}S_1'$ implies that they share the same PC, so we can denote them as 
\begin{align*}
    & S_0'=(R_0, M_0, (\pcVar, (\sigma_{f0}, \sigma_{b0}));\Phi_0)
    && S_1'=(R_1, M_1, (\pcVar, (\sigma_{f1}, \sigma_{b1}));\Phi_1).
\end{align*}

If one of them is a termination state, the other must also be a termination state.
If none of them is a termination state, then according to Theorem~\ref{thm:functional-correctness-revision},
there exist $S_0''$, $S_1''$ such that $S_i'\xrightarrow{\Compiler(\instVar), o_i}^* S_i''$, $i\in \qty{0, 1}$.
Note that they execute the same instruction(s) since they share the same PC.
Consider the following cases for $\instVar$:

\textbf{Case $\inst{movq}{r, \memOp{r_d}{r_i}{i_s}{i_d}\opMode{\sOp, \tauOp}}$.}
We have 
\begin{equation*}
    \Compiler(\instVar)=\inst{movq}{r, \memOp{r_d}{r_i}{i_s}{\delta_0+i_d}\opMode{\sOp+\delta_0, \tauOp}}
\end{equation*}
According to \textsc{Dyn-Movq-r-m} and public equivalence,
\begin{align*}
    R_0[r] & = (v_1, t_1) & R_1[r] & = (v_2, t_2) \quad (v_1, t_1)\simeq_\text{taint}(v_2, t_2) \\
    R_0[r_d] & =R_1[r_d]=(v_d, 0) & R_0[r_i] & =R_1[r_i]=(v_i, 0).
\end{align*}
So $S_0$ and $S_1$ store to the same address $v_a\coloneq \delta_0+v_d+v_i\times i_s+i_d$.

Furthermore, if $\Compiler(\instVar)$ stores to the public region, i.e.,
$v_a\in \stackPub\cup \otherPub$, then $t_1=t_2=0$ according to Lemma~\ref{lemma:trans-store-pub}.
Thus, $(v_1, t_1)\simeq_\text{taint}(v_2, t_2)$ implies that $v_1=v_2$.
Therefore, $o_1=o_2$.
We can also easily derive that $S_0''\simeq_\text{pub}S_1''$ by applying \textsc{Dyn-Movq-r-m}.

\textbf{Case $\inst{movq}{\memOp{r_d}{r_i}{i_s}{i_d}\opMode{\sOp, \tauOp}, r}$.}
We have  
\begin{equation*}
    \Compiler(\instVar)=\inst{movq}{\memOp{r_d}{r_i}{i_s}{\delta_0+i_d}\opMode{\sOp+\delta_0, \tauOp}, r}
\end{equation*}
where $\delta_0=0$ or $\delta$.
Similarly, \textsc{Dyn-Movq-m-r} and public equivalence imply that both machines load from the same untainted address, so $o_1=o_2$.
$S_0''\simeq_\text{pub}S_1''$ can also be derived by applying \textsc{Dyn-Movq-r-m}.

\textbf{Case $\inst{jne}{\ell\opMode{\sigma}}$.}
We have $\Compiler(\instVar)=\instVar$.
According to \textsc{Dyn-Jne} and public equivalence,
$R_0[\reg{ZF}]=R_1[\reg{ZF}]=(b, 0)$, which implies that both machines will jump to the same PC.
Therefore, $o_1=o_2$ and $S_0''\simeq_\text{pub}S_1''$.

\textbf{Case $\inst{callq}{f\opMode{\sigmaCall, \sigma_\text{retq}}}$.}
We have
\begin{equation*}
    \Compiler(\instVar)=\inst{addq}{\delta, r_1};\dots \inst{addq}{\delta, r_n}; \inst{callq}{f\opMode{\sigmaCall', \sigma_\text{retq}}}.
\end{equation*}
Note that the $\instr{addq}$ instructions do not access memory or affect taint equivalence ($\simeq_\text{taint}$ relation) of register values. They also only update the PC to the next instruction.
When executing $\inst{callq}{f\opMode{\sigmaCall', \sigma_\text{retq}}}$,
according to \textsc{Dyn-Callq}, both machines will jump to $f$, i.e., the same PC.
Furthermore, both machines store the same return address $\text{nextPC}(P, \pcVar)$ to the same untainted address.
Therefore, the statement also holds for this case.

\textbf{Case $\inst{retq}{}$.}
We have $\Compiler(\instVar)=\instVar$.
According to \textsc{Dyn-Retq}, both machines pop the return address from the same untainted address.
The return address is also untainted and therefore the same for both machines.
Therefore, the statement holds.

We omit other cases since the statement can be proved similarly.
\end{proof}

\begin{definition}[Well-Formed Transformed State]
$S'$ is a well-formed state for $\Compiler(P)$ if there exists $S$ such that
$\typeChecked{P, \PType}{S}$ and $\hasVal{P, \PType}{S' \prec S}$.
\end{definition}

\begin{theorem}[Public Noninterference]
If a program $P$ is well-typed, then $\Compiler(P)$ satisfies software public noninterference.
\label{thm:public-noninterference-revision}
\end{theorem}
\begin{proof}
According to Definition~\ref{def:pub-noninterference}, we just need to prove that for all well-formed initial states $S_0$ and $S_1$, if $S_0\simeq_\text{pub}S_1$, then 
$\swcontract{\Compiler(P)}{pub}{}(S_0)=\swcontract{\Compiler(P)}{pub}{}(S_1)$
and
$\swcontract{\Compiler(P)}{ct}{}(S_0)=\swcontract{\Compiler(P)}{ct}{}(S_1)$.
This can be proved by applying Lemma~\ref{lemma:pub-noninterference} repeatedly.
\end{proof}

\todo{@adam: Do I need to further explain the relation between well-formed initial states of $\PType$ and $\Compiler(\PType)$?}
\todo{I feel that something is missing here: I did not explain that my type constraints still allow us to have tainted memory.}



\shixin{
Need to add the following to guarantee functional correctness:
\begin{itemize}
    \item Define $\text{getNonChangeVar}(\Delta)$: all vars except for pointers and callee's type at the beginning of the function should be added to this set. Branch conditions should only depends on these vars (after simplification).
    Also need to define isNonChangeExp
    \item For subtype relation $e\sqsubseteq \top$, there must be $\text{getVars}(e)\subseteq \text{getNonChangeVar}(\Delta)$.
    This has two implications:
    \begin{itemize}
        \item When updating state types with non-branch inst, we should be careful when make dest reg/mem slot's type to $\top$ - check source types to only contain non change var.
        \item When checking subtype relation of state types between blocks, we should also check this constraint.
    \end{itemize}
    \item Implementation:
    \begin{itemize}
        \item Add $\bot$ to indicate that reg is invalid.
        \item Check the above requirements.
    \end{itemize}
\end{itemize}
}



\fi

\end{document}